\documentclass[a4paper,11pt]{article}

\usepackage[utf8]{inputenc}
\usepackage{amsmath,amssymb}
\usepackage{framed,caption,color,url}
\usepackage{standalone}
\usepackage{times}
\usepackage{mdframed}
\usepackage{algorithm2e,algorithmicx,algpseudocode,tikz,wrapfig}
\usetikzlibrary{decorations.pathreplacing,calligraphy,arrows, arrows.meta}
\usepackage{graphicx}               
\usepackage{cite} 
\usepackage{nicefrac} 
\usepackage{paralist} 
\usepackage{boxedminipage} 
\usepackage[pdfstartview=FitH,colorlinks,linkcolor=blue,filecolor=blue,citecolor=blue,urlcolor=blue,pagebackref=true]{hyperref}
\usepackage{upgreek}
\usepackage{centernot}
\usepackage{mathtools}
\usepackage{cleveref} 
\usepackage{subcaption}
\usepackage{stmaryrd} 

\ifdefined\LLNCS
\else
\usepackage{fullpage}
\usepackage{amsthm}
\fi

\newcommand{\spn}{\mathrm{Span}}

\newcommand{\fhat}[2]{\ifthenelse{\equal{#2}{}}{\hat{f}[#1]}{\ifthenelse{\equal{#2}{0}}{\hat{f}[\emptyset]}{\hat{f}[#1_{\leq #2}]}}}
\newcommand{\gain}[2]{\ifthenelse{\equal{#2}{}}{g[#1]}{g[#1_{\leq #2}]}}

\newcommand{\pr}[2][]{\Pr_{\ifthenelse{\isempty{#1}}{}{{#1}}}\left[{#2}\right]}

\usepackage{dsfont}

\newcommand{\E}{\Exp}

\newcommand{\remove}[1]{}

\newcommand{\ol}{\overline}
\newcommand{\wt}[1]{\widetilde{#1}}

\newcommand{\etal}{et~al.\ }

\newcommand{\resp}{resp.,\ }
\newcommand{\ie}{i.e.,\ }

\newcommand{\wrt} {with respect to\ }

\newcommand{\ceil}[1]{\lceil #1 \rceil}
\newcommand{\floor}[1]{\lfloor #1 \rfloor}

\newcommand{\bra}[1]{\langle#1\rvert}
\newcommand{\braket}[1]{\langle #1 \rangle}
\newcommand{\ket}[1]{\lvert#1\rangle}
\newcommand{\set}[1]{\{ #1 \}}
\newcommand{\inner}[2]{\langle#1\vert#2\rangle}

\newcommand{\C}{\mathbb{C}}

\newcommand{\bits}{\{0,1\}}
\newcommand{\xor}{\oplus}

\newcommand{\size}[1]{\left|#1\right|}
\newcommand{\abs}[1]{\size{#1}}
\newcommand{\norm}[1]{\left\lVert#1\right\rVert}

\newcommand{\N}{{\mathbb N}}

\newcommand{\cA}{{\mathcal A}}
\newcommand{\cB}{{\mathcal B}}
\newcommand{\cC}{{\mathcal C}}
\newcommand{\cD}{{\mathcal D}}
\newcommand{\cE}{{\mathcal E}}
\newcommand{\cF}{{\mathcal F}}
\newcommand{\cG}{{\mathcal G}}
\newcommand{\cH}{{\mathcal H}}

\newcommand{\cR}{{\mathcal R}}
\newcommand{\cS}{{\mathcal S}}
\newcommand{\cT}{{\mathcal T}}

\newcommand{\cX}{{\mathcal X}}
\newcommand{\cY}{{\mathcal Y}}

\newcommand{\bfr}{\mathbf{r}}

\newcommand{\bfu}{\mathbf{u}}

\newcommand{\bfx}{\mathbf{x}}
\newcommand{\bfy}{\mathbf{y}}

\usepackage{bm}

\newcommand{\veps}{\varepsilon}

\newcommand{\poly}{\operatorname{poly}}
\newcommand{\polylog}{\operatorname{polylog}}

\newcommand{\Exp}{\operatorname*{\mathbb{E}}}
\newcommand{\Ex}{\Exp}

\newcommand{\negl}{\operatorname{negl}}

\newcommand{\class}[1]{\ensuremath{\mathbf{#1}}}

\ifdefined\LLNCS
\newtheorem{observation}[theorem]{Observation}
\newtheorem{fact}[theorem]{Fact}

\newtheorem{construction}[theorem]{Construction}
\newtheorem{algorithm1}[theorem]{Algorithm}
\newtheorem{assumption}[theorem]{Assumption}

\makeatletter
\renewcommand{\@Opargbegintheorem}[4]{%
  #4\trivlist\item[\hskip\labelsep{#3#2\@thmcounterend}]}
\makeatother

\else

\newtheorem{theorem}{Theorem}[section]
\theoremstyle{plain}

\newtheorem{lemma}[theorem]{Lemma}
\newtheorem{corollary}[theorem]{Corollary}

\newtheorem{observation}[theorem]{Observation}

\theoremstyle{definition}

\newtheorem{definition}[theorem]{Definition}

\newtheorem{assumption}[theorem]{Assumption}

\theoremstyle{definition}
\newtheorem{remark}[theorem]{Remark}

\fi

\ifdefined\LLNCS
\else

\newcommand{\sdotfill}{\textcolor[rgb]{0.8,0.8,0.8}{\dotfill}} 

\makeatletter
\def\th@protocol{%
\normalfont 
\setbeamercolor{block title example}{bg=orange,fg=white}
\setbeamercolor{block body example}{bg=orange!20,fg=black}
\def\inserttheoremblockenv{exampleblock}
}
\makeatother
\theoremstyle{protocol}
\newtheorem{proto}[theorem]{Protocol}
\newtheorem{protoc}[theorem]{Protocol}

\fi

\newcommand{\namedref}[2]{#1~\ref{#2}}

\newcommand{\torestate}[3]{%
\expandafter \def \csname BBRESTATE #2 \endcsname{#3}
\theoremstyle{plain}
\newtheorem{BBRESTATETHMNUM#2}[theorem]{#1}
\begin{BBRESTATETHMNUM#2}\label{#2}\csname BBRESTATE #2 \endcsname   \end{BBRESTATETHMNUM#2}
\newtheorem*{BBRESTATETHMNONNUM#2}{\namedref{#1}{#2}}
}

\newcommand{\restate}[1]{\begin{BBRESTATETHMNONNUM#1}[Restated] \csname BBRESTATE #1 \endcsname
\end{BBRESTATETHMNONNUM#1}}

\newcommand{\trdist}{\Delta} 
\newcommand{\tr}{\mathrm{tr}}
\newcommand{\HL}{H_L}
\newcommand{\evolt}{t} 

\newcommand{\OracleH}{\mathcal{O}_{H}} 
\newcommand{\OracleSparse}{\mathcal{O}_L} 
\newcommand{\pquery}{p}
\newcommand{\querynum}{q}
\newcommand{\hash}{\Pi}
\newcommand{\hashstep}{f}
\newcommand{\reduction}{\mathcal{R}}
\newcommand{\Hsimalg}{\mathcal{A}} 

\newcommand{\sPi}{\text{S}\Pi}
\newcommand{\xb}{\bar{x}}
\newcommand{\Pitl}{\tilde{\Pi}}
\newcommand{\sPitl}{\text{S}\tilde{\Pi}}

\newcommand{\nc}{\newcommand}
\newcommand{\nn}{\nonumber}
\newcommand{\sPitlk}{\sPitl^{\otimes k}}
\newcommand{\Proj}[1]{\left|#1\right\rangle\left\langle #1\right|}
\nc{\vev}[1]{\langle#1\rangle}
\nc{\dxb}{\Delta \xb}
\nc{\probone}{O\left(\ell\sqrt{\frac{k}{N}}\right)}
\nc{\proboneq}{O\left(q\sqrt{\frac{k}{N}}\right)}

\usepackage{cleveref}

\title{On the Impossibility of General Parallel Fast-forwarding of \\ Hamiltonian Simulation}

\ifdefined\authorNote
\newcommand{\Knote}[1]{{\color{red} [{K:} #1]}}
\newcommand{\nai}[1]{{\color{blue} [{Nai:} #1]}}
\newcommand{\YT}[1]{{\color{brown} [{YT:} #1]}}
\newcommand{\yuching}[1]{{\color{magenta} [{yuching:} #1]}}
\newcommand{\YC}[1]{{\color{olive} [{YC:} #1]}}
\newcommand{\hannote}[1]{\textcolor{blue}{\small {\textbf{(Han:} #1\textbf{) }}}}
\else
\newcommand{\Knote}[1]{}
\newcommand{\nai}[1]{}
\newcommand{\YT}[1]{}
\newcommand{\yuching}[1]{}
\newcommand{\YC}[1]{}
\newcommand{\hannote}[1]{}
\fi

\newcommand{\email}[1]{\href{mailto:#1}{#1}}

\author{
    Nai-Hui Chia\thanks{Rice University, USA, \email{nc67@rice.edu}.}  \and
	Kai-Min Chung\thanks{Academia Sinica, Taiwan, \email{kmchung@iis.sinica.edu.tw}.}  \and
	Yao-Ching Hsieh\thanks{University of Washington, Seattle, USA, \email{ychsieh@cs.washington.edu}, work done at Academia Sinica.} \and
	Han-Hsuan Lin\thanks{National Tsing Hua University, Taiwan, \email{linhh@cs.nthu.edu.tw}.} \and
	Yao-Ting Lin\thanks{UCSB, USA, \email{yao-ting\_lin@ucsb.edu}. Part of
    the work was done when working at Academia Sinica.} \and
	Yu-Ching Shen \thanks{Academia Sinica, Taiwan, \email{yuching@iis.sinica.edu.tw}.}
}

\date{}

\begin{document}
\maketitle

\begin{abstract}
     Hamiltonian simulation is one of the most important problems in the field of quantum computing. There have been extended efforts on designing algorithms for faster simulation, and the evolution time $T$ for the simulation turns out to largely affect algorithm runtime. While there are some specific types of Hamiltonians that can be fast-forwarded, i.e., simulated within time $o(T)$, for large enough classes of Hamiltonians (e.g., all local/sparse Hamiltonians), existing simulation algorithms require running time at least linear in the evolution time $T$. On the other hand, while there exist lower bounds of $\Omega(T)$  \emph{circuit size} for some large classes of Hamiltonian,  these lower bounds do not rule out the possibilities of Hamiltonian simulation with large but ``low-depth'' circuits by running things in parallel. Therefore, it is intriguing whether we can achieve fast Hamiltonian simulation with the power of parallelism.
    
    In this work, we give a negative result for the above open problem, showing that sparse Hamiltonians and (geometrically) local Hamiltonians cannot be parallelly fast-forwarded. In the oracle model, we prove that there are time-independent sparse Hamiltonians that cannot be simulated via an oracle circuit of depth $o(T)$. In the plain model, relying on the random oracle heuristic, we show that there exist time-independent local Hamiltonians and time-dependent geometrically local Hamiltonians that cannot be simulated via an oracle circuit of depth $o(T/n^c)$, where the Hamiltonians act on $n$-qubits, and $c$ is a constant.
\end{abstract}

\tableofcontents

\section{Introduction}

Simulating a physical system with a specified evolution time is an essential approach to
study the properties of the system. In particular, given a Hamiltonian $H$ that presents the physical system of interest and the evolution time $t$, the goal is to use some well-studied physical system as a simulator to implement $e^{-iHt}$ when $H$ is time-independent or $\exp_{\cT} \left(-i\int_0^{t} H(t') dt' \right)$ for time-dependent $H$, where $\exp_\cT$ denotes the time-ordered matrix exponential. Intuitively, a simulator simulates a Hamiltonian $H$ step by step and thus requires time at least linear in $t$. On the other hand, if one can use a well-studied physical system (e.g., digital or quantum computers) to simulate the Hamiltonian of interest with time significantly less than the specified evolution time, it can significantly benefit our study of physics. Following this line of thought, a fundamental question for simulating Hamiltonians is: 
\begin{center}
    \textsl{Can a simulator simulate Hamiltonians in time strictly less than the evolution time?}
\end{center}
This is called \emph{fast-forwarding of Hamiltonians}. In this work, we investigate the possibility of achieving fast-forwarding of Hamiltonian using quantum computation.

It is known that quantum algorithms can fast-forward some Hamiltonians. Atia and Aharonov~\cite{Atia_2017} showed that commuting local Hamiltonians and quadratic fermionic Hamiltonians with evolution time $t=2^{\Omega(n)}$ can be \emph{exponentially} fast-forwarded by quantum algorithms, where the Hamiltonian applies on $n$ qubits. This result implies the existence of quantum algorithms that simulate the two classes of Hamiltonians in $\poly(n)$ time. Gu et al.~\cite{Gu_2021} showed that more Hamiltonians could be exponentially or polynomially fast-forwarded, such as the exponential fast-forwarding for block diagonalizable Hamiltonians and polynomial fast-forwarding method for frustration-free Hamiltonians at low energies. 

The existence of general fast-forwarding methods for Hamiltonians using quantum computers has also been studied. In particular, people investigated whether all Hamiltonians with some ``succinct descriptions'', such as local and sparse Hamiltonians, can be fast-forwarded. Berry et al.~\cite{berry2007efficient} proved no general fast-forwarding for all sparse Hamiltonians of $n$ qubits for evolution time $t=n\pi/2$ by a reduction from computing the parity of a binary string. In particular, computing the parity of an $n$-bit string requires $\Omega(n)$ quantum queries, and they showed that any algorithm that simulates the corresponding Hamiltonian in time $o(n)$ will violate the aforementioned query lower bound of parity. Atia and Aharonov~\cite{Atia_2017} further showed that if all 2-sparse row-computable\footnote{Given the row index, one can efficiently compute the column indices and values of the non-zero entries of the row.} Hamiltonians with evolution time $2^{o(n)}$ can be simulated in quantum polynomial time, then $\class{BQP} = \class{PSPACE}$. In other words, the result in~\cite{Atia_2017} rules out the possibility of exponential fast-forwarding for Hamiltonians with evolution time superpolynomial in the dimension under well-known complexity assumptions. Haah et al.~\cite{Haah_2021} showed that there exists a piecewise constant bounded 1D time-dependent Hamiltonian\footnote{The Hamiltonian is 1D local, and there is a time slicing such that $H(t)$ is time-independent within each time slice. See~\cite{Haah_2021} for the formal definition.} $H(t)$ on $n$ qubits, such that any quantum algorithm simulating $H(t)$ with evolution time $T$ requires $\Omega(nT)$ gates.

All these works, however, mainly considered lower bounds on the \emph{number of gates} required for simulation, and it does not rule out the possibility that one can complete the simulation with time strictly less than $t$ by using \emph{parallelism}. Briefly, if many local gates in an algorithm operate on disjoint sets of input, then these gates can be applied together, and the efficiency of the algorithm is captured by the \emph{circuit depth} instead of the number of gates. For instance, the result in~\cite{berry2007efficient} was based on the fact that the query complexity of parity is $\Omega(n)$ and thus, one needs $\Omega(t)$ queries to simulate the Hamiltonian evolution; however, if one runs $t$ queries in parallel, it is possible to solve the parity problem with one layer of queries. Therefore, a direct translation of~\cite{berry2007efficient}  does not rule out the possibility of constant depth simulation of the Hamiltonian evolution for time $t$ by running $O(t)$ simulations in parallel.

Parallel runtime (i.e., the circuit depth in the quantum circuit model) could be another suitable notion for capturing the efficiency of the Hamiltonian simulation. Broadly speaking, any physically controllable and implementable system can be used as a simulator, so-called quantum analogue computing\cite{CZ12, CMP18}; instead of having one local interaction at each time step, a simulator that is realized by some physical system will have the whole system evolve together. From a computational perspective, a positive result of fast-forwarding Hamiltonians using parallel algorithms can imply that the simulation can be done in time strictly less than the specified evolution time given sufficient computational resources. In particular, if there exists an algorithm that simulates all Hamiltonians in time less than the simulation time $t$, we might be able to further reduce the runtime by recursively applying the algorithm with sufficient quantum resources. Hence, such algorithms can be a powerful tool for studying quantum physics. In fact, parallel quantum algorithms for Hamiltonian simulation have been studied and showed some advantages. Zhang et al.~\cite{zhang2021parallel} presented a parallel quantum algorithm, whose parallel runtime (circuit depth) has a doubly logarithmic dependency on the simulation precision. Moreover, Haah et al.~\cite{Haah_2021} showed that the circuit depth of their algorithm for simulating geometrically constant-local Hamiltonians can be reduced to $O(t\cdot \polylog(tn/\epsilon))$ by using ancilla qubits. In the last, choosing other physical systems similar to the target Hamiltonians as simulators is possible to gain advantages, which is the idea of quantum analogue computing.

We first explore the possibility of achieving fast-forwarding of Hamiltonians using parallel quantum algorithms, i.e., quantum algorithms that have circuit depth strictly less than the simulation time. We call this \emph{parallel fast-forwarding.} Our first goal is to address the following question: 
\begin{center}
    \emph{For all sparse or local Hamiltonians, do there exist quantum algorithms that simulate the Hamiltonians with circuit depth strictly less than the required evolution time?} 
\end{center}

Furthermore, we noticed that more general simulators (in addition to quantum circuit models) are widely considered for Hamiltonian simulation, such as quantum analogue computing. So, we are also wondering about the following question.

\begin{center}
\emph{For all sparse or local Hamiltonians, does there exist a natural simulator that simulates the Hamiltonians with evolution time strictly less than the required evolution time?}
\end{center}

\subsection{Our Results}

In the work, we give negative answers to the above questions. Roughly speaking, we show that under standard cryptographic assumptions, there exists Hamiltonians that cannot be parallelly fast-forwarded by quantum computers and any simulators that are geometrically local physical systems.

We define parallel fast-forwarding as follows: 

\begin{definition}[Parallel fast-forwarding] 
\label{dfn:pff_informal}
Let $\mathbf{H}$ be a subset of all normalized Hamiltonian ($\|H\|=1$) and $\mathbf{H}_n$ be the subset of Hamiltonian in $\mathbf{H}$ which acts on $n$ qubits. We say that the set $\mathbf{H}$ can be $(T(\cdot), g(\cdot), \varepsilon(\cdot))$-parallel fast forwarded if there exists an efficient classical algorithm $\cA(1^n,t)$ that outputs a circuit $C_{n,t}$, i.e., $\{C_{n,t}\}$ is a uniform quantum circuit, such that for all $n\in\mathbb{N}, t\leq T(n)$, $C_{n,t}$ satisfies the following two properties.
\begin{itemize}
\item The circuit $C_{n,t}$ has depth at most $g(t)$.
\item For all $H\in\mathbf{H}_n, \ket{\psi}\in\mathbb{C}^{2^n}$, the circuit $C_{n,t}(H, \ket{\psi})$ (or $C_{n,t}^H(\ket{\psi})$ under the oracle setting) has output state that is $\varepsilon(n)$ close to the Hamiltonian evolution outcome $e^{-iHt}\ket{\psi}$.
\end{itemize}
In other words, there exists uniform quantum circuit $C_{n,t}$ such that for every Hamiltonian $H\in\mathbf{H}_n$, the evolution of $H$ to time $t$ up to some predetermined time bound $T(\cdot)$ can be simulated by $C$.
\end{definition}

Compared to~\cite{Atia_2017}, Definition~\ref{dfn:pff_informal} focuses on $C$'s circuit depth instead of the number of gates and requires the depth of $C$ to be smaller than $t$ rather than being $\poly(n)$. In particular, when $t = \mathsf{superpoly}(n)$, the definition in~\cite{Atia_2017} can only be satisfied by a circuit $C$ that has gate number superpolynomially smaller than $t$, and $C$ that has circuit depth slightly less than $t$ can satisfy Definition~\ref{dfn:pff_informal}. Therefore, we can also interpret the no fast-forwarding theorem in~\cite{Atia_2017} as refuting the possibility of achieving Definition~\ref{dfn:pff_informal} with gate number (and also circuit depth) superpolynomially smaller than $t = \mathsf{superpoly}(n)$ for a specific family of Hamiltonians. However, given that negative result, one might ask the following question: 
\begin{center}
    \emph{Can we achieve parallel fast-forwarding with $g(t)$ slightly smaller than $t$, such as $g(t) = \sqrt{t}$?} 
\end{center}

In this work, we address the aforementioned question and show impossibility results for parallel fast-forwarding with circuit depth $g(t)$ slightly smaller than $t$ for local or sparse Hamiltonians. Our first result is an unconditional\footnote{That is, the result holds without making any computational assumptions.} result under the oracle model.\footnote{By the oracle model we mean that the algorithm can only access the Hamiltonian by making (quantum) queries to the oracle that encodes the description of the Hamiltonian. See~\Cref{sec:oracle} for the definition.}

\begin{theorem}[No parallel fast-forwarding for sparse Hamiltonians relative to random permutations, simplified version of Theorem~\ref{thm:lower_bound_oracle}]
Relative to a random permutation oracle over $n$-bit strings, for any polynomial $T(\cdot)$, there exists a family of time-independent sparse Hamiltonians $\mathbf{H}$ such that 
$\mathbf{H}$ cannot be $(T(\cdot), g(\cdot), \varepsilon(\cdot))$-fast forwarded for some $g=\Omega(t)$ and $\varepsilon=\Omega(1)$.
\label{thm:sh_informal}
    
\end{theorem}

To obtain no fast-forwarding result in the standard model,\footnote{By the standard (plain) model we mean that the algorithm is given the classical description of the Hamiltonian as input, which is the standard setting of the Hamiltonian simulation problem. Moreover, there is no oracle that can be accessed by algorithms. We will use the terms ``standard model'' and ``plain model`` interchangeably throughout this work.} we rely on cryptographic assumptions that provide hardness against low-depth algorithms. 
We assume the existence of \emph{iterative parallel-hard functions}, formally defined in Definition~\ref{def:p-hard_funcs}. Roughly speaking, an iterative parallel-hard function is a function of the form $f(k,x) = g^{(k)}(x):= \underbrace{g(g(\dots g(x)))}_{k\text{ times}}$, such that $g$ is efficiently computable (by some circuit of size $s$), but $g^{(k)}(x)$ is not computable for circuits with depth much less than $k$.

With such a cryptographic assumption, we obtained the following two no-fast-forwarding theorems under the standard model.

\begin{theorem}[No parallel fast-forwarding for local Hamiltonians, simplified version of Theorem~\ref{thm:lower_bound_plain}]
\label{thm:lh_informal}
Assuming the existence of iterative parallel-hard functions with size parameter $s(n)$, then for every polynomial $T(n)$, there exists a family of time-independent local Hamiltonians $\mathbf{H}$
such that 
$\mathbf{H}$ cannot be $(T(\cdot), g(\cdot), \varepsilon(\cdot))$-fast forwarded for some $g=\Omega(t/s(n))$ and $\varepsilon=\Omega(1)$.
\end{theorem}

\begin{theorem}[No parallel fast-forwarding for time-dependent geometrically local Hamiltonians, simplified version of Theorem~\ref{thm:lower_bound_dep}]
\label{thm:glh_informal}
Assuming the existence of iterative parallel-hard functions with size parameter $s(n)$, then for every polynomial $T(n)$, there exists a family of time-dependent geometrically local Hamiltonian $H$ 
such that 
$\mathbf{H}$ cannot be $(T(\cdot), g(\cdot), \varepsilon(\cdot))$-fast forwarded for some $g=\Omega(t/ns(n))$ and $\varepsilon=\Omega(1)$.
\end{theorem}

Some loss in parameters are hidden in Theorem~\ref{thm:lh_informal} and Theorem~\ref{thm:glh_informal}. Readers are referred to the full theorems in Section~\ref{sec:plain} for precise parameters.

We note that the existence of parallel-hard functions with an iterative structure is widely used in cryptography. Our definition of iterative parallel-hard functions adapts from the iterated sequential function proposed by Boneh et al. \cite{Boneh18VDP}. Functions of this form play a crucial role in the recent construction of verifiable delay functions (VDF)\cite{Pie19, Wes20, EFK20}. 
In contrast to its wide usage, there have not been many proposals on candidates for such iterative hard functions. Iterative squaring\cite{RSW96}, which is probably the most widely used candidate, is not hard against quantum circuits. There are some recent attempts toward constructing iterated quantum-hard functions from isogenies \cite{FMP19,CRT21}, but these assumptions are much less well-studied.

As a concrete instantiation of our iterated parallel-hard function, we adopted a hash chain, which is also widely assumed to be hard to compute within low depth. In Section \ref{sec:t-hashchain}, we justify the quantum parallel hardness of the hash chain by showing a depth lower bound of computing the hash chain in the quantum random oracle model~\cite{BDFLSZ11}.

Our results in Theorem~\ref{thm:sh_informal}, Theorem~\ref{thm:lh_informal}, and Theorem~\ref{thm:glh_informal} imply that no quantum algorithm can simulate certain families of local or sparse Hamiltonians with circuit depth polynomially smaller than $t$. For instance, suppose $t = n^{c}$ for some constant $c$ and $s(n) = n^2$, then by Theorem~\ref{thm:lh_informal}, no quantum algorithm can simulate the local Hamiltonians with circuit depth smaller than $t^{c-2}$.

Since local Hamiltonians are sparse, Theorem~\ref{thm:lh_informal} also implies no parallel fast-forwarding of sparse Hamiltonians in the standard model. Finally, Theorem~\ref{thm:glh_informal} and Theorem~\ref{thm:lh_informal} are incomparable due to the fact that the Hamiltonians in Theorem~\ref{thm:glh_informal} are time-dependent and the depth lower bound has a factor of $n$. 

It is worth noting that the results above show no parallel fast-forwarding when using ``quantum circuits'' as simulators, which does not directly imply hardness results when considering other physical systems as simulators. Especially, choosing physical systems that are similar to the Hamiltonians to be simulated is possible to gain advantages, and physical systems naturally evolve the whole system together instead of applying local operators one by one. Therefore, it is nontrivial whether similar results hold for other simulators. Fortunately, we are able to generalize Theorem~\ref{thm:glh_informal} and Theorem~\ref{thm:lh_informal} to show that natural simulators that are geometrically local Hamiltonians cannot do much better than quantum circuits.

\begin{theorem}
[No fast-forwarding for local Hamiltonians with natural simulators, simplified version of Corollary~\ref{cor:lower_bound_plain}]
\label{thm:informal_naruto_1}
Assuming the existence of iterative parallel-hard function with size parameter $s(n)$, then for every polynomial $T(n)$, there exists a family of time-independent local Hamiltonians $\mathbf{H}$ over $\widetilde{O}(n)$ qubits satisfying the following. For any geometrically constant-local Hamiltonian $H_B$ acting on $\poly(n)$ qubits, using $H_B$ to simulate any $H_A\in\mathbf{H}$ for any evolution time $t\in[0, s(n)T(n)]$ needs an evolution time at least $(t/2s(n)-O(s(n)))/\polylog(tn)$.
\end{theorem}

\begin{theorem}[No fast-forwarding for geometrically local Hamiltonians with natural simulators, simplified version of Corollary~\ref{cor:lower_bound_tdh}]
\label{thm:informal_naruto_2}
 Assuming the existence of iterative parallel-hard functions with size parameter $s(n)$, then for every polynomial $T(n)$, there exists a family of time-dependent geometrically local Hamiltonians $\mathbf{H}$ over $\widetilde{O}(n)$ qubits satisfying the following. 
For any geometrically constant-local Hamiltonian $H_B$ acting on $\poly(n)$ qubits, using $H_B$ to simulate any $H_A\in\mathbf{H}$ for any evolution time $t\in[0, ns(n)T(n)]$ needs an evolution time at least $\left( \frac{t}{ns(n)}-O(ns(n))-\polylog (n) \right)/\polylog(tn)$.
\end{theorem}

\section{Technical Overview}

\paragraph{The main idea} Our idea is to reduce some tasks that have a circuit or query depth lower bounds (i.e., parallel-hard problems) to simulating specific Hamiltonians with evolution time $t$, such that the existence of parallel fast-forwarding of the Hamiltonians will contradict the circuit depth lower bound and also violate the parallel hardness of the task. For instance, one can reduce parity, which is not in $\class{QNC}^0$ (the class of all constant-depth bounded fan-in circuits), to simulate a corresponding Hamiltonian $H$ with some time $t$, such that $e^{-iHt}$ outputs the parity of the input. Along this line, if $e^{-iHt}$ can be implemented by a constant-depth quantum circuit, we can compute parity -- this violates the quantum circuit lower bound on parity! Following the same idea, one can also derive some no-go results for parallel fast-forwarding from unstructured search, where the $k$-parallel quantum query complexity is ${\Theta}(\sqrt{N/k})$, where $k$-parallel means each ``query layer'' can have $k$ queries in parallel~\cite{JMD13,zalka1999grover}.

However, there are several challenges: First, those above-mentioned parallel-hard problems can be solved in depth smaller than the input size. This could result in a Hamiltonian simulation in which the evolution time is smaller than the number of qubits. Although this might still lead to an impossibility result for parallel fast-forwarding of an $o(n)$ evolution time, parallel fast-forwarding algorithms for such a short evolution time seem not that useful. In fact, to the best of our knowledge, it is not easy to find a problem that can be computed in quantum polynomial time while having a quantum depth strictly greater than the input size using $\polylog(n)$ parallel queries. So, one technical contribution of our work is finding such problems and proving their quantum depth. 

Second, finding appropriate reductions from the parallel-hard problems to Hamiltonians of our interest and preserving the input size and the quantum depth is also challenging. Note that we are focusing on sparse or local Hamiltonians with evolution time, a polynomial in the number of qubits. One intuitive approach is trying the circuit-to-Hamiltonian reduction in~\cite{nagaj2010fast,KSV02}. Briefly, the reduction uses a $t$-depth circuit on $n$ qubits to simulate a local Hamiltonian on $n+t$ qubits with time $t$, where the additional $t$ qubits are for the ``clock register''. This, as mentioned above, has an evolution time smaller than the number of qubits. In this work, we find reductions that map a $d$-depth $n$-qubit quantum computation with $d=\poly(n)$ to a local or sparse Hamiltonian with the number of qubits and evolution time ``close to'' $n$ and $t$ respectively. 

Another challenge is that we need the parallel-hard problem as an iterative structure to show no parallel fast-forwarding theorems. More specifically, our goal is to prove that some Hamiltonians cannot be parallel fast-forwarded with \emph{any evolution time in the specified range}. Therefore, we might need a sequence of parallel-hard problems such that there are corresponding parallel-hard problems \emph{for all $t$} in the range. 
In addition, given a parallel-hard problem with an iterative structure, it is not trivial how to reduce it to one Hamiltonian $H$ with different evolution times $t$ such that simulating $H$ for different $t$ gives the corresponding answers.   

\paragraph{Parallel hardness of the underlining assumptions} One candidate for parallel-hard problems with an iterative structure of our purpose is the \emph{hash chain}. Roughly speaking, let $\cX$ be a finite set and $h:\cX\rightarrow \cX$ be a hash function. An $s$-chain of $h$ is a sequence $x_0,x_1,\dots,x_s\in\cX$ such that $x_{i+1} = h(x_i)$ for any $i\in [s-1]$. Given quantum oracle access to $h$, the goal of the algorithm is to find an $s$-chain. Classically, it was proven that classical algorithms require query depth of at least $s$ to output an $s$-chain with constant probability. A similar result also holds for quantum algorithms that make quantum queries to the hash function~\cite{chung2021compressed}. Along this line, the hash chain seems ideal for our purpose because $s$ can be a polynomial in $\log(|\cX|)$ and has the iterative parallel hardness.

However, a hash function is generally irreversible, and this fails standard approaches for reducing the problem to Hamiltonian simulation. Briefly, one encodes $h$ as a Hamiltonian $H$ such that evaluating $h$ is equivalent to applying $e^{-iH}$. Since $e^{-iH}$ is a unitary that is reversible, evaluating $h$ also needs to be reversible. Here, we give \emph{permutation chain} and \emph{twisted hash chain} that are iteratively parallel-hard and the underlining function is reversible. However, the reversibility imposes another challenge, as the ability to query the inverse of the permutation breaks the known composed oracle techniques used to prove the hardness of hash chain~\cite{chung2021compressed}\footnote{One can use the technique in~\cite{zhandry13} to convert random permutations to random functions, but the conversion only works when the algorithm has no access to the inverse oracle.}. Therefore we tailored a two-step-hybrid argument to prove the hardness of the random permutation chain with the ability to query the inverse of the permutations.

Note that for oracle lower bounds of parallel query algorithms, while ~\cite{JMD13} gives optimal bounds by generalizing the adversary method, it is notoriously hard to find the suitable adversary matrices. Therefore we derive the query lower bounds for our problems by crafting a hybrid argument and using the compressed oracle technique~\cite{zhandry2019record} respectively.

\subsection{No parallel fast-forwarding for sparse Hamiltonians relative to random permutation oracle}
\label{subsec:sh_informal}
We first introduce the permutation chain and demonstrate how to prove Theorem~\ref{thm:sh_informal} via the \emph{graph-to-Hamiltonian reduction} based on the permutation chain. This shows no parallel fast-forwarding for sparse matrices relative to a random permutation oracle.

\paragraph{Permutation chain} One of the reversible parallel-hard problem we formulated is the \emph{permutation chain}. In this problem, we are given as inputs $q$  permutations of $N:=2^n$ elements $\Pi_1,\Pi_2,\dots,\Pi_q$.\footnote{They can be viewed as a special case of one permutation of $qN$ elements.} Let $\sPi$ be the unitary that enables one to query to each $\Pi_i$ and their inverses in superposition. Let $\xb_1=1$ and $\xb_{i+1}=\Pi_i(\xb_i)$ so that $\xb_{q+1}=\Pi_q(\dots \Pi_2(\Pi_1(1)))$. With $q$ queries to $\sPi$, it is easy to calculate $\xb_i$, while we prove that it is only possible to calculate $\xb_i$ with probability $O(q\sqrt{k/N})$ using $\floor{(q-1)/2}$ $k$-parallel queries\footnote{$k$-parallel means each ``query layer'' can have $k$ queries in parallel} to $\sPi$. Therefore if we have $q,k=O(\polylog(N))$, the success probability is negligible in $n$, even when $k$ is larger than $q$ and having access to the inverses of $\Pi_1,\Pi_2,\dots,\Pi_q$. 

To bound the success probability, we employed a two-step hybrid. First, we show that we can replace each $\sPi$ with $\sPitl$. $\sPitl$ is a set of functions that return zeros almost everywhere except at $\{\xb_i\}$, where they behave the same as $\sPi$ (see Figure~\ref{fig:hash_chain_permutation}(a)(b)). We prove that we can approximately simulate one call to a random $\sPi$ with two calls to $\sPitl$. Now, $\sPitl$ looks like a constant zeros function, we can erase some of its values without getting caught. In the second step, we show that we can release the $\Pitl_i$'s on a finely controlled schedule, with only negligible change in the output probability. Define $\Pi^\bot$ to be a constant zero function. Define $\sPitl_\ell$ to be the unitary corresponding to $\Pitl_1,\Pitl_2,\dots,\Pitl_{\ell}, \Pi^\bot, \dots$, \ie all but the first $\ell$ permutations are erased (see Figure~\ref{fig:hashchian_hybrid_2}). We show that if we replace the first $k$-parallel queries of $\sPitl$ with $\sPitl_1$, second $k$-parallel queries of $\sPitl$ with $\sPitl_2$, third $k$-parallel queries of $\sPitl$ with $\sPitl_3$, etc, we can only be caught with negligible probability. Intuitively, this is because while we are at the $i$-th query layer, it is hard to find any non-zero values of $\Pitl_{i+1}, \dots,\Pitl_{q}$. Therefore, if an algorithm only makes $q-1$ queries to $\sPitl$, we can replace the queries with $\sPitl_1, \sPitl_2, \dots, \sPitl_{q-1}$. It is impossible to find $\xb_{q+1}$ with non-negligible probability since these oracles do not have information of $\Pitl_q$.

\paragraph{Graph-to-Hamiltonian reductions}
The purpose of graph-to-Hamiltonian reduction is using quantum walk on a line \cite{Childs_2003} to solve the permutation chain. Briefly, we use a graph to encode the permutation chain and let Hamiltonian $H$ be the adjacency matrix that represents the graph. Then, the time evolution operator $e^{-iHt}$ helps to find the solution of permutation chain. Therefore, a low-depth Hamiltonian simulation algorithm for $H$ could result in breaking the hardness of permutation chain. This gives our first impossibility result of parallel fast-forwarding for sparse Hamiltonians.

Let $\Pi_1,\Pi_2,\dots\Pi_L$ be $L$ permutations over $N$ elements.
We use a graph with $N(L+1)$ vertices in which each vertex labelled by $(j,x)$ to record the permutation chain, where $j\in \{0,1,\dots,L\}$ and $x\in [N]$.
The vertices $(j,x)$ and $(j+1, x')$ are adjacent if and only if $x' = \Pi_{j+1}(x)$.
The construction of the graph has followings properties.
First, the graph consists of $N$ disconnected line because each $\Pi_j$ is a permutation. 
Second, each vertex $(q,x)$ that connects to $(0, x_0)$ satisfies $x_q = \Pi_q(\Pi_{q-1}\cdots(\Pi_1(x_0))$.
To solve the permutation chain problem, we start from the vertex $(0, x_0)$ and walk along the connected line.
When stopping at a vertex $(q, x_q)$, the pair $(x_q, x_0)$ would be a solution of permutation chain. It is obvious that the adjacency matrix of the corresponding graph is \emph{sparse}. We let the Hamiltonian $H$ determining the dynamics of the walk be the adjacency matrix of the graph, and our goal is to find $(x_q, x_0)$ by simulating $e^{-iHt}$ given \emph{sparse access} to $H$.

There are two main challenges for building such a reduction: First, we need to implement the sparse oracle access to the corresponding Hamiltonian. This requires oracle access to the permutation and inverse permutation oracle. More specific, we need to implement two oracles that are used in the Hamiltonian simulation algorithm to execute the quantum walk.
The first one is the entry oracle $\OracleH$, which answers the element value of $H$ when queried on the matrix index.
The second one is the sparse structure oracle $\OracleSparse$, which answers the indices of the nonzero entries when queried on the row index.
To implement $\OracleH$, it is equivalent to checking if two vertices $(j, x)$ and $(j+1, x')$ are adjacency.
It can be done by querying $\Pi_{j+1}(x)$.
To implement $\OracleSparse$, it is equivalent to finding the vertices that are adjacent to $(j,x)$.
Finding $(j+1,x')$ adjacent to $(j,x)$ can be done by querying $\Pi_{j+1}(x)$, but finding $(j-1,x'')$ needs to query $\Pi_{j}^{-1}(x)$. Hence, we need to consider the security of permutation chain when the inversion oracle $\Pi_{j}^{-1}$ is given to the adversary. We bypass this challenge by showing that the our permutation chain is secure against quantum adversaries even if inverse permutation oracle is given as we previously discussed.

Second, we need to show that the simulation algorithm is able to walk fast enough so that simulating $H$ for evolution time close to the length of the chain gives the solution to the permutation chain. To be more precise, we aim to design the system such that after walking for time $t$, it reaches the vertex further than $t$ with high probability. Recall that $H$ determining the dynamics of the walk is the adjacency matrix of the graph corresponding to the permutation chain. We observe that for such quantum walk system, it indeed reaches some points beyond $t$ for the walking time $t$ with high probability.
At any time $t$, the system is described by the quantum state $e^{-iHt}\ket{0, x_0}$. 
The probability of stopping on the vertex $(q, x_q)$ at time $t$ is $P(q) =\abs{\bra{q,x_q}e^{-iHt}\ket{0,x_0}}^2$.
We have $\abs{\bra{q,x_q}e^{-iHt}\ket{0,x_0}} = qJ_{q}(2t)/t$ for $t\in[0, L/2]$, where $J_q(\cdot)$ is the $q$-th order Bessel function \cite{Childs_2003}.
By the properties of Bessel function, we show that $\sum_{q=\ceil{t}}^{L} P(q)= O(1)$, which means that the probability of stopping at a vertex $(l, x_l)$ such that $l>t$ is high.
As a result, it breaks the hardness of permutation chain if $e^{-iHt}$ can be implemented with $o(t)$ queries.

\subsection{No parallel fast-forwarding for (geometrically) local Hamiltonians in the plain model}
\label{subsec:lh_informal}

To show no fast-forwarding of (geometrically) local Hamiltonians in the plain model, the combination of the permutation chain and the graph-to-Hamiltonian reduction used in Section~\ref{subsec:sh_informal} might be insufficient. First, it is unclear how to instantiate random permutation oracle. In addition, even if we can translate the permutation chain to a parallel-hard quantum circuit in the plain model, the graph-to-Hamiltonian reduction inherently provides sparse oracle access to the Hamiltonian from oracle access to the permutation chain. However, we need to have the full classical descriptions of each local term for simulating local Hamiltonians. 

Observing these difficulties, we introduce the twisted hash chain and the circuit-to-Hamiltonian reduction for proving Theorem~\ref{thm:lh_informal} and Theorem~\ref{thm:glh_informal}. 

\paragraph{Twisted hash chain}
In order to implement a reversible operation (or a permutation), we follow the idea of the Feistel network~\cite{Luby88}. Roughly speaking, the Feistel network is an implementation of block ciphers by using cryptographic hash functions. By means of chaining quantum query operators as in Figure~\ref{fig:feistel}, the outputs in each layer satisfy $x_i = H(x_{i-1}) \xor x_{i-2}$. Therefore, we can think of it as a ``quantum version'' of the Feistel networks. Informally, the goal of the algorithm is to output the head and tail of a chain of length $q+1$ by using at most $q$ depth of queries.

For proving the parallel hardness, we use the compressed oracle technique by Zhandry~\cite{zhandry2019record}. In particular, the analysis is undergone in the framework of Chung \etal~\cite{chung2021compressed} where they generalize the technique to the parallel query model. Our proof is inspired by the parallel hardness of the standard hash chain proven in~\cite{chung2021compressed}. For technical reasons, the challenge is the following: in the twisted hash chain problem, the algorithm is not required to output \emph{all} elements of the chain and their hash values. Therefore, we cannot directly apply the tools provided in~\cite{chung2021compressed}. In addition, we cannot simply ask the algorithm to spend extra queries for outputting the hash values since this would lead to a trivial bound (we call the extra queries for generating the whole chain the ``verification'' procedure). Instead, we need a more fine-grained analysis of the verification procedure. First, we notice that since $x_i = H(x_{i-1})\xor x_{i-2}$, the verification requires \emph{sequential} queries. Therefore, unlike Theorem~5.9 in \cite{chung2021compressed} where the verification procedure only requires parallel queries, the analysis for our purpose is more involved.

We bypass the aforementioned issue by reduction. Suppose there is an algorithm $\cA$ outputs $x_0,x_q,x_{q+1}$ such that $x_0,\dots,x_{q+1}$ form a $(q+1)$-chain by making $q$ $k$-parallel queries. Then we can construct a reduction $\cB$ which first runs $\cA$ and obtain $x_0,x_q,x_{q+1}$. Next, $\cB$ starts with $x_0,x_{q+1}$ and queries each element of the chain iteratively \emph{in parallel} until approaching $x_{q-1},x_{2q}$. If $\cA$ successfully outputs a $(q+1)$-chain, then it implies that $\cB$ also outputs the \emph{complete} $(2q+1)$-chain with hash values but $H(x_{2q+1})$ by making a total of $2q$ $k$-parallel queries. As a result, it remains to analyze the success probability of making the last additional query on $x_{2q+1}$. In this way, it significantly simplifies the proof.

\paragraph{Circuit-to-Hamiltonian reductions}
For our results in the plain model, we leverage the power of the random oracle heuristic. From the parallel hardness of twisted hash chain, we can obtain a heuristically parallel-hard circuit that preserves the iterative structure. Evaluation of this circuit to large depth directly translates to computing a hash chain of large length, which is assumed to be hard for low depth circuit. To translate the hardness to a no parallel fast-forwarding result, we embed the computation of the circuit to a Hamiltonian via two different approaches.

To embed circuit computation to a time-independent Hamiltonian, we use the technique from Nagaj \cite{nagaj2010fast}, which demonstrate how to transform a circuit computation with size $T$ to a Hamiltonian evolution problem of time $O(T\log T)$. In our work, we make two major modification upon Nagaj's technique. First, we observed that Nagaj's technique fits well with our iterated structure of circuit. At a high level, simulating Hamiltonian obtained from Nagaj's compiler can be interpreted as a quantum walk on a line, where each point on the line correspond to a computation step/gate of the circuit. Again by the detailed analysis on Bessel function that we used in the graph-to-Hamiltonian reduction, we observe that we can obtain a "depth $O(t)$" intermediate state of computing $C$ by evolving $H$ for time $O(t)$. This not only gives a better fast-forward lower bound, but also allows us to obtain a Hamiltonian that is hard to fast-forward on \emph{every} evolution time within time bound $T$. Second, Nagaj's construction gives a Hamiltonian of $O(n+T)$-qubits, where $n$ is the circuit input size and $T$ is the circuit size. This is an issue because it restricts our no fast-forwarding results to evolution times small than the Hamiltonian size. We overcome this by introducing a new design for the clock state via the Johnson graph. Our restructured clock state allows a fine-grained tradeoff between the locality parameter and the Hamiltonian size.

For our second construction, we achieve the geometrically local property with the power of time-dependent Hamiltonians. Our idea is to use the piecewise-time-independent construction from \cite{Haah_2021}, in which simulating the Hamiltonian for each time segment on the initial state behaves equivalently to applying a gate on the state. We take one step further by transforming our circuit to contain gates operating on neighboring gates only. This gives us a geometrically 2-local Hamiltonian which is hard to fast-forward. Combined with the algorithm that simulates geometrically local Hamiltonians also by \cite{Haah_2021}, our result tightens the gap between upper bounds and lower bounds to a small polynomial in qubit number $n$.

\begin{remark} Two things worth to be noted for the two approaches in Section~\ref{subsec:sh_informal} and Section~\ref{subsec:lh_informal}:
\begin{itemize}
    \item If one can instantiate random permutations by hash functions or other algorithms without using keys, one can obtain Theorem~\ref{thm:lh_informal} and Theorem~\ref{thm:glh_informal} by combining the permutation chain and the circuit-to-Hamiltonian reduction.
    \item The combination of the twisted hash chain and the circuit-to-Hamiltonian reduction can give no parallel fast-forwarding for Hamiltonians in the random oracle model. This is similar to Theorem~\ref{thm:sh_informal}; however, Theorem~\ref{thm:sh_informal} using the permutation chain and the graph-to-Hamiltonian reduction provides a better size of the Hamiltonians. In particular, the Hamiltonian in Theorem~\ref{thm:sh_informal} has the number of qubits independent of the evolution time, while the Hamiltonians given from the circuit-to-Hamiltonian reduction has the number of qubits that is poly-logarithmic in the evolution time.
\end{itemize}
\end{remark}

\section{Open Questions}

In this work, we showed that the existence of a parallel-hard problem with an iterative structure implies no parallel fast-forwarding of sparse and (geometrically) local Hamiltonians under cryptographic assumptions. Along this line, the first question that is natural to ask is whether there exist more Hamiltonians that have succinct descriptions and cannot be parallelly fast-forwarded under other computational assumptions.

We are also wondering whether the existence of parallel-hard problems with an iterative structure is equivalent to no parallel fast-forwarding. This is equivalent to proving or disproving that no parallel fast-forwarding results in parallel-hard problems with an iterative structure. Intuitively, One can show that the existence of Hamiltonians that cannot be parallelly fast-forwarded implies some quantum circuits that have no smaller circuit depth. This follows from the fact that if one can implement a quantum circuit with a depth smaller than the quantum simulation algorithm for the Hamiltonian, one can achieve parallel fast-forwarding. However, this task asks the algorithm to output quantum states close to $e^{-iHt}\ket{\psi}$ and thus is not a ``classical computational problem'' as parallel-hard problems with an iterative structure.

In addition, we want to match the upper and lower bounds for parallel fast-forwarding of Hamiltonian simulation. For instance, for geometrically local Hamiltonians, the algorithms in~\cite{Haah_2021} require depth $O(t\cdot \polylog(tn/\epsilon))$, where $n$ is the number of qubits and $\epsilon$ is the precision parameter. There is still a $O(1/ns(n))$ gap compared to our result in Theorem~\ref{thm:glh_informal}. Likewise, our results for sparse (Theorem~\ref{thm:sh_informal}) and local Hamiltonians (Theorem~\ref{thm:lh_informal}) have not matched the upper bounds from known quantum simulation algorithms, such as~\cite{zhang2021parallel, LC17,LC19}.

The questions mentioned above are to investigate the optimal quantum circuit depth for Hamiltonian simulation under certain computational assumptions. Note that the Hamiltonian simulation problem has classical inputs and quantum outputs. Inspired by this, we are wondering a more general question: \emph{is it possible to prove quantum circuit lower bounds for complexity classes that have classical inputs and quantum outputs?} For example, can we unconditionally show quantum circuit depth lower bounds for Hamiltonian simulation or some quantum states with succinct classical descriptions? Note that although showing circuit depth lower bounds for languages is challenging and has some barriers, complexity classes with quantum outputs might have specific properties and provide new insights into showing quantum circuit depth lower bounds. 
\section{Preliminaries and Notation}

\subsection{Notation}
For $n\in\N$, we use $[n]$ to denote the set $\{1, 2, \dots , n\}$.
The trace distance between two density matrices $\rho$ and $\sigma$ is denoted by $\trdist(\rho, \sigma):= \frac{1}{2}\norm{\rho-\sigma}_1= \frac{1}{2}\tr\left(\sqrt{(\rho-\sigma)^{\dagger}(\rho-\sigma)} \right)$.
Let $x_1,x_2$ be $n$-bit strings, we use $x_1 \oplus x_2$ to denote the bitwise XOR of $x_1$ and $x_2$. The Kronecker delta is denoted by $\delta_{jk}$ where $\delta_{jk} = 0$ if $j\neq k$ and $\delta_{jk} = 1$ if $j = k$. 

\subsection{Hamiltonian simulation}

\begin{definition}[Hamiltonian simulation]
A Hamiltonian simulation algorithm $\cA$ takes as inputs the description of the Hamiltonian $H$, an initial state $\ket{\psi_0}$, the evolution time $t \geq 0$ and an error parameter $\epsilon \in (0, 1]$. 
Let $\ket{\wt{\psi_t}}$ be the ideal state under the Hamiltonian $H$ for evolution time $t$ with the initial state $\ket{\psi_0}$.
In other words, $\ket{\wt{\psi_t}} := e^{-iHt}\ket{\psi_0}$ for a time-independent $H$, and $\ket{\wt{\psi_t}} := \exp_{\cT}\left(-i\int_{0}^{t} H(t') dt' \right)\ket{\psi_0}$ for a time-dependent $H$, where $\exp_{\cT}$ is the time-ordered matrix operator.
The goal of $\cA$ is to generate an approximation $\ket{\psi_{t}}$ of the evolved quantum state $\ket{\wt{\psi_{t}}}$ such that 
\[
\trdist\left(\ket{\psi_{t}}\bra{\psi_{t}}, \ket{\wt{\psi}_{t}}\bra{\wt{\psi}_{t}}\right) \leq \epsilon.
\]
\end{definition}

\subsection{Basic quantum computation}
Below, we provide a brief introduction to quantum computation. For more basics, we refer the readers to \cite{nielsen2010quantum}. Throughout this work, we use the standard bra-ket notation.

\begin{definition}[Quantum circuit model]
A quantum circuit consists of qubits, a sequence of quantum gates, and measurements. A qubit is a two-dimensional complex Hilbert space. Each qubit is associated with a \emph{register}. A quantum gate is a unitary operator acting on quantum registers. We say a quantum gate is a $k$-qubit gate if it acts non-trivially on $k$ qubits.
\end{definition}

\begin{theorem}[Universal gate sets \cite{BMP+99}]
There exists a \emph{universal gate set} that consists of a finite number of quantum gates such that  any unitary operator can be approximated by composing elements in the universal gate set within an arbitrary error. Furthermore, every element in the universal gate set is a one- or two-qubit gate.
\end{theorem}

\begin{definition}[Quantum circuit depth]
Given a finite-sized gate set $\cG$, a \emph{$d$-depth quantum circuit} or a \emph{quantum circuit of depth $d$ \wrt $\cG$} consists of a sequence of $d$ layers of gates such that (i) each gate belongs in $\cG$ and (ii) each gate within the same layer acts on disjoint qubits.
We omit the gate set $\cG$ when it is clear from the context.
\end{definition}

\begin{definition}[Quantum query operator]
Given an oracle $f: \bits^n \to \bits^m$,
the \emph{query operator} $\mathcal{O}_f$ is defined as
\[
\mathcal{O}_f\ket{x,y} \coloneqq \ket{x,y\oplus f(x)}.
\]
\end{definition}

\begin{definition}[Parallel quantum query operator]
  \label{dfn:parallel_query}
  Given an oracle $f: \bits^n \to \bits^m$.
  The \emph{$k$-parallel query operator} $\mathcal{O}^{\otimes k}_f$ is defined as
  \[
  \mathcal{O}^{\otimes k}_f\ket{\bfx,\bfy}
  \coloneqq \ket{\bfx,\bfy\xor f(\bfx)},
  \]
  where $\bfx=(x_1,\dots,x_k)$,  $\bfy=(y_1,\dots,y_k)$ and $f(\bfx) \coloneqq (f(x_1),\dots f(x_k))$.
\end{definition}

\subsection{Useful tools}
In this subsection, we introduce several definitions and lemmas for analyzing quantum random walk in Section~\ref{sec:Q_walk} and the clock state construction in Section~\ref{sec:plain}. 

\subsubsection{Bessel functions}
The Bessel functions of the first kind of order $n$ are denoted by $J_n(x)$. We present the required properties of Bessel functions for our use.

\begin{itemize}
    \item The integration form of the Bessel function:
    \begin{equation}
        \label{eq:jn_integral}
        J_{n}(x) = \frac{1}{2\pi}\int_{-\pi}^{\pi}dp\ e^{inp - ix\sin p}
        = \frac{i^{n}}{2\pi} \int_{-\pi}^{\pi}dp\ e^{inp - ix\cos p}.
    \end{equation}
    \item
    The relation between $J_n$ and $J_{-n}$:
    \begin{equation}
        \label{eq:jn_minus_n}
        J_{-n}(x) = (-1)^{n} J_n(x).
    \end{equation}
    \item The recursion formula for integer orders:
    \begin{equation}
        \label{eq:jn_recursion}
        J_{n+1}(x) = \frac{2n}{x} J_{n}(x) - J_{n-1}(x).
    \end{equation}
    \item
    The asymptotic form for large order
    \begin{equation}
    \label{eq:jn_large_order_xltn}
        J_{n}(n\ \mathrm{sech} \xi) \sim \frac{e^{-n(\xi - \tanh \xi)}}{\sqrt{2\pi n\tanh \xi}}
    \end{equation}
    suggests that when $x<\abs{n}$, the value of $J_n(x)$ is exponentially small in $n$.
\end{itemize}

The following lemmas provide upper bounds for Bessel functions for large argument $x$.

\begin{lemma}[Theorem 2 in \cite{Kra06}]\label{lem:bessel}
  Let $n > -1/2$ and $\mu \coloneqq (2n+1)(2n+3)$.
  For any $x > \sqrt{\mu + \mu^{2/3}}/2$, it holds that
\[
J_{n}^2 (x)
\leq \frac{4\left(4x^2-(2n+1)(2n+5)\right)}{\pi\left((4x^2-\mu)^{3/2}-\mu\right)}.
\]
\end{lemma}

By Lemma~\ref{lem:bessel}, we have the following lemma which is more convenient for our use.

\begin{lemma}\label{lem:bessel_tail}
Let $n$ be a positive integer. For any real $x \geq 2n$, it holds that
\[
    J_{n}^2 (x) \leq \frac{2}{n\pi}.\footnote{The bound reminds us of the well known formula: $J_n(x)\sim \sqrt{2/(\pi x)} \left(\cos (x- nx/2 -\pi/4 ) +O(x^{-1})\right)$. However, it holds for $x>n^2$ only.}
\]

\end{lemma}

\begin{proof}
We discuss the behavior of Bessel functions in three cases: $n=1$, $n=2$, and $n\ge 3$.
For $n=1$, the maximum of $J^2_1(x)$ is $0.339\dots$ which is less than $2/\pi \approx 0.637$.
For $n=2$, the maximum of $J^2_2(x)$ is $0.237\dots$ which is less than $1/\pi \approx 0.318$.

Now let us analyze the case in which $n\geq 3$. First, we notice that when $x > 2n$, the conditions in Theorem~\ref{lem:bessel} hold. This is because $\mu+\mu^{2/3} < 2\mu$ and then
\[
\frac{\sqrt{\mu+\mu^{2/3}}}{2}
< \frac{\sqrt{2\mu}}{2}
= \sqrt{2n^2+ 4n +\frac{3}{2}}
< \sqrt{4n^2} = 2n < x,
\]
where the second inequality holds when $n \geq 3$.

Now, we will finish the proof by bounding the numerator and the denominator of the RHS in Lemma~\ref{lem:bessel}. For the numerator, we have
\[
4\left(4x^2 - (2n+1)(2n+5) \right)
< 4\left(4x^2 - (2n+1)(2n+3) \right)
= 4\left(4x^2 - \mu \right).
\]
For the denominator, we will show that
\[
(4x^2 - \mu)^\frac{3}{2} -\mu > \frac{2}{3}(4x^2 - \mu)^\frac{3}{2}
\]
or equivalently
\[
\frac{1}{3}(4x^2 - \mu)^\frac{3}{2} > \mu.
\]

First, since $x \geq \sqrt{2\mu}/2$, we have $4x^2 - \mu \geq \mu$.
Furthermore, when $n \geq 3$ we have $\mu \geq 35$, which would imply $\frac{1}{3} \mu^{3/2} > \mu$.
Hence, we conclude that $\frac{1}{3}(4x^2 - \mu)^\frac{3}{2} \geq \frac{1}{3} \mu^{3/2} > \mu$.
Putting things together, we obtain
\[
J_{n}^2(x) < \frac{4}{\pi}\cdot \frac{(4x^2-\mu)}{\frac{2}{3}(4x^2-\mu)^{3/2}}
= \frac{4}{\pi}\cdot \frac{1}{\frac{2}{3}\sqrt{4x^2 - \mu}}.
\]
When $x > 2n$ and $n > 3$, it holds that $4x^2-\mu \geq 16n^2 - (4n^2 +8n +3) \geq 9n^2$. Therefore, we finally obtain
\[
    J_n^2(x) \leq \frac{4}{\pi}\cdot \frac{1}{\frac{2}{3}\sqrt{4x^2-\mu}}
    \leq \frac{4}{\pi}\cdot \frac{1}{\frac{2}{3}\cdot 3n}
    = \frac{2}{n\pi}.
\]
This finishes the proof.
\end{proof}

\subsubsection{Johnson graph}

\begin{definition}[Johnson Graph]
  \label{def:Johnson_graph}
    For all integers $n\geq k\geq 1$, the \emph{$(n,k)$-Johnson graph} $J_{n,k} = (V,E)$ is an undirected acyclic graph defined as follows.
    \begin{itemize}
        \item $V \coloneqq \set{S \subseteq [n] \colon |S| = k}$, \ie the vertices are the $k$-element subsets of an $n$-element set.
        \item $E \coloneqq \set{(S_0,S_1) \colon |S_0 \cap S_1| = k-1}$, \ie there is an edge if and only if the intersection of the two vertices (subsets) contains $k-1$ elements.\footnote{Equivalently, we can define $E \coloneqq \set{(S_0,S_1) \colon |S_0 \cup S_1| = k+1}$.}
    \end{itemize}
    The number of vertices in $J_{n,k}$ is $\binom{n}{k}$.
    It was proven that for all integers $n\geq k\geq 1$, there exists a Hamiltonian path\footnote{A Hamiltonian path is a path that visits every vertex in the graph exactly once. Do not confuse it with the physical quantity we want to simulate.} in $J_{n,k}$\cite{alspach2012johnson}. 
\end{definition}
\section{Lower Bounding Permutation Chain}
\label{sec:hashchain}

\begin{definition}[Permutation notations]\label{def:perm-notation}
Here we define several notations for the later proofs. 
Let $\Pi_1,\Pi_2,\dots,\Pi_q$ be permutations of $N$ elements. 
Let $\Pi^{-1}_1, \Pi^{-1}_2, \dots ,\Pi^{-1}_q$ be the corresponding inverse permutations.
Define the sets $[-q]:=\{-q,-q+1,\dots,-1\}$ and $[\pm q]:=[q]\cup[-q]$. 
We define the unitary $\sPi$ as the controlled version of the above permutations as
\begin{align}
    \sPi\ket{j,x,r}:=\begin{cases}
    \ket{j,x,r\oplus \Pi_j(x)} &,j>0   \\ 
        \ket{j,x,r\oplus \Pi^{-1}_{|j|}(x)} &,j<0
    \end{cases}  
\end{align}
where $j\in [\pm q]$ and $x,r \in [N]$. \\
We denote the elements of the chain by $\xb_1:=1$ and $\xb_{i+1}:=\Pi_i(\xb_{i})$ for all $i\in [k]$.
Next, we define $\Pitl_i$ to be the ``erased'' $\Pi_i$ for all $i\in[q]$. Formally, $\Pitl_i$ is defined to be the function $[N]\rightarrow [N]\cup \{0\}$ such that
\begin{align}
    \Pitl_i(x)=\begin{cases} \xb_i &,x=\xb_{i-1} \\
    0 &,\text{otherwise}.
    \end{cases}
\end{align}

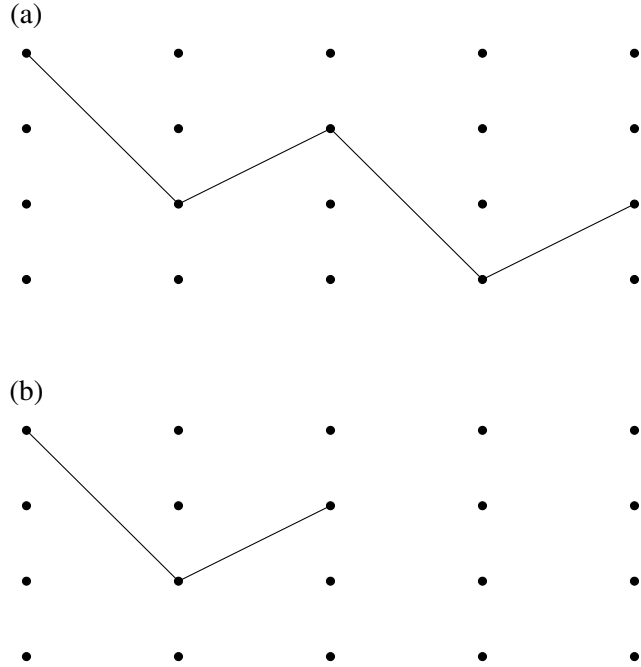
\begin{figure}[ht]
    \centering
    \begin{adjustbox}{max width = \textwidth}
\begin{tikzpicture}
  \node at (0, 3.5) {(a)};

  \filldraw [black] (0,3) circle (1.5pt);
  \filldraw [black] (0,2) circle (1.5pt);
  \filldraw [black] (0,1) circle (1.5pt);
  \filldraw [black] (0,0) circle (1.5pt);

  \filldraw [black] (2,3) circle (1.5pt);
  \filldraw [black] (2,2) circle (1.5pt);
  \filldraw [black] (2,1) circle (1.5pt);
  \filldraw [black] (2,0) circle (1.5pt);
  
  \filldraw [black] (4,3) circle (1.5pt);
  \filldraw [black] (4,2) circle (1.5pt);
  \filldraw [black] (4,1) circle (1.5pt);
  \filldraw [black] (4,0) circle (1.5pt);
  
  \filldraw [black] (6,3) circle (1.5pt);
  \filldraw [black] (6,2) circle (1.5pt);
  \filldraw [black] (6,1) circle (1.5pt);
  \filldraw [black] (6,0) circle (1.5pt);
  
  \filldraw [black] (8,3) circle (1.5pt);
  \filldraw [black] (8,2) circle (1.5pt);
  \filldraw [black] (8,1) circle (1.5pt);
  \filldraw [black] (8,0) circle (1.5pt);

  \draw (0,3) -- (2,1);
  \draw (2,1) -- (4,2);
  \draw (4,2) -- (6,0);
  \draw (6,0) -- (8,1);

  \node at (0, -1.5) {(b)};

  \filldraw [black] (0,-2) circle (1.5pt);
  \filldraw [black] (0,-3) circle (1.5pt);
  \filldraw [black] (0,-4) circle (1.5pt);
  \filldraw [black] (0,-5) circle (1.5pt);

  \filldraw [black] (2,-2) circle (1.5pt);
  \filldraw [black] (2,-3) circle (1.5pt);
  \filldraw [black] (2,-4) circle (1.5pt);
  \filldraw [black] (2,-5) circle (1.5pt);
  
  \filldraw [black] (4,-2) circle (1.5pt);
  \filldraw [black] (4,-3) circle (1.5pt);
  \filldraw [black] (4,-4) circle (1.5pt);
  \filldraw [black] (4,-5) circle (1.5pt);
  
  \filldraw [black] (6,-2) circle (1.5pt);
  \filldraw [black] (6,-3) circle (1.5pt);
  \filldraw [black] (6,-4) circle (1.5pt);
  \filldraw [black] (6,-5) circle (1.5pt);
  
  \filldraw [black] (8,-2) circle (1.5pt);
  \filldraw [black] (8,-3) circle (1.5pt);
  \filldraw [black] (8,-4) circle (1.5pt);
  \filldraw [black] (8,-5) circle (1.5pt);

  \draw (0,-2) -- (2,-4);
  \draw (2,-4) -- (4,-3);

\end{tikzpicture}
\end{adjustbox}
    \caption{
    Schematic diagram of $\sPitl$ and $\sPitl_\ell$. 
    (a) The permutation chain $\xb_1,\xb_2,\dots\xb_{q+1}$ specified by $\sPitl$.
    (b) $\sPitl_\ell$, where the permutations are removed after $\ell$ queries.
    }
    \label{fig:hashchian_hybrid_2}
\end{figure}

Similarly, we define the corresponding controlled unitary $\sPitl$. For all $i\in[k]$, define $\dxb_i:=\xb_{i+1}-\xb_i \mod{N}$. Note that $\sPitl$ can be parameterized by either $\{\xb_2,\dots,\xb_{q+1} \}$ or $\{\dxb_1,\dots,\dxb_q \}$. 
Denote the transformation from $\sPi$ to $\sPitl$ by $\sPitl=F(\sPi)$. 
For all $\ell\in [q]$, define the hybrid oracle $\sPitl_{\ell}$ as
\begin{align}
    \sPitl_{\ell}\ket{j,x,r}:=\begin{cases}
    \ket{j,x,r\oplus \Pitl_j(x)} &,|j|\leq \ell, j>0   \\ 
    \ket{j,x,r} &, j>\ell \\
    \ket{j,x,r\oplus \Pitl^{-1}_{|j|}(x)} &,|j|\leq \ell,j<0 \\
    \ket{j,x,r} &, j<-\ell.
    \end{cases}  
\end{align}
\end{definition}

\begin{theorem}\label{thm:hashchain_main}
Use the notations of Definition~\ref{def:perm-notation}. Let $q,k$ be integers such that $k=O(\polylog(N))$ and $q=O(\polylog(N))$.
For any quantum algorithm $\cA$ using $\floor{(q-1)/2}$ $k$-parallel queries to $\sPi$, we have 
$$\E_{\sPi}\left[\Pr[ \cA^{\sPi} \text{ outputs } \xb_{q+1}]\right] = \proboneq .$$
\end{theorem}

Before proving Theorem~\ref{thm:hashchain_main}, we first introduce several lemmas as follows.
\begin{lemma}\label{lem:2-norm_to_trace}
    For any $\ket{\phi}$, $\ket{\psi}$ such that $\|\ket{\phi}\| = \|\ket{\psi}\| = 1$ and $\|\ket{\phi}-\ket{\psi}\| \leq \veps$, it holds that 
    \[
    \Delta(\ket{\phi}\bra{\phi},\ket{\psi}\bra{\psi}) \leq \veps.
    \]
\end{lemma}
\begin{proof}
    The trace distance between two pure states is given by $\sqrt{1-|\braket{\phi | \psi}|^2}$.
    The Euclidean norm of $\ket{\phi}-\ket{\psi}$ is given by $\|\ket{\phi}-\ket{\psi}\| = \sqrt{(\bra{\phi}-\bra{\psi})(\ket{\phi}-\ket{\psi})} = \sqrt{2-2\mathsf{Re}[\braket{\phi | \psi}]}$, where $\mathsf{Re}[\cdot]$ denote the real part of a complex number. 

    First, it is true that $0 \leq \left( \mathsf{Re}[\braket{\phi | \psi}]-1 \right)^2 + \mathsf{Im}[\braket{\phi | \psi}]^2$, where $\mathsf{Im}[\cdot]$ denote the imaginary part of a complex number.
    Rearranging the terms, we obtain 
    \[
    1-|\braket{\phi | \psi}|^2
    = 1 - (\mathsf{Re}[\braket{\phi | \psi}]^2 + \mathsf{Im}[\braket{\phi | \psi}]^2) 
    \leq 2 - 2\mathsf{Re}[\braket{\phi | \psi}].
    \]
\end{proof}

\begin{lemma}[$q$-bin $k$-parallel Grover search lower bound] \label{lem:k-grover}
Let $\cF$ be the set of all functions $f$ from $[qN]$ to $\bits$ with the following promise. For all $i \in \set{0,1,\dots,q-1}$, it holds that $|\set{x \in \set{iN+1,iN+2,\dots,iN+N} \colon f(x) = 1}| = 1$.
Let $g$ be the constant zero function with domain $[qN]$. 
Then for every algorithm that makes $\ell$ $k$-parallel queries to $f$ (or $g$), the final state of the algorithm, denoted by $\ket{\psi^f}$ (or $\ket{\psi^g}$), satisfies
\[
\Ex_{f\gets\cF}\left[ \Delta\left( \ket{\psi^f}\bra{\psi^f},\ket{\psi^g}\bra{\psi^g} \right) \right] 
= O\left(\ell\sqrt\frac{k}{N}\right).
\]    
\end{lemma}

\begin{proof}
Let $\ket{\psi^g_i} \coloneqq U_i O^{\otimes k}_g \dots U_1 O^{\otimes k}_g U_0\ket{0}$. For any $f\in\cF$, let $\Pi_f$ be the projector acting on the query register of the algorithm that is defined as $\Pi_f \coloneqq \sum_{x:f(x)=1}\ket{x}\bra{x}$. Let $\ol{\Pi}_{f} \coloneqq I-\Pi_{f}$.

Notice that $(O^{\otimes k}_f-O^{\otimes k}_g)\ol{\Pi}_{f}^{\otimes k} = 0$ because $O_f$ and $O_g$ are identical beyond the set of the $1$-preimages.
Therefore, we have
\begin{equation}\label{eq:oracle_diff}
(O^{\otimes k}_f-O^{\otimes k}_g) (I - \ol{\Pi}_{f}^{\otimes k}) = O^{\otimes k}_f-O^{\otimes k}_g.
\end{equation}

\noindent Also, it holds that
\begin{equation}\label{eq:inclu}
I - \ol{\Pi}_{f}^{\otimes k} \leq \sum_{i = j}^k \Pi_{f,j},
\end{equation}
where $\Pi_{f,j}$ denote the operator the acts as $\Pi_{f}$ on the register of the $j$-th query branch and as identity on other registers; for matrices $A,B$, by $A \geq B$ we mean that $A-B$ is a positive semi-definite matrix. Then by standard hybrid arguments~\cite{BBBV97}, we have
\begin{align*}
    \Ex_{f\gets\cF}\left[ \|\ket{\psi^f}-\ket{\psi^g}\| \right]
    & \leq \sum_{i=0}^{\ell-1} \Ex_{f\gets\cF}\left[ \|(O^{\otimes k}_f-O^{\otimes k}_g)\ket{\psi^g_i}\| \right] \\
    & = \sum_{i=0}^{\ell-1} \Ex_{f\gets\cF}\left[ \|(O^{\otimes k}_f-O^{\otimes k}_g) (I - \ol{\Pi}_{f}^{\otimes k}) \ket{\psi^g_i}\| \right] \\
    & \leq \sum_{i=0}^{\ell-1} \Ex_{f\gets\cF}\left[ \|(O^{\otimes k}_f-O^{\otimes k}_g)\| \cdot \| (I - \ol{\Pi}_{f}^{\otimes k}) \ket{\psi^g_i}\| \right],
\end{align*}
where the first equality is due to \eqref{eq:oracle_diff} and the last inequality is due to the fact that $\|A\ket{\phi}\| \leq \|A\| \cdot \|\ket{\phi}\|$, where $\|A\|$ denotes the operator norm of $A$.

Since $\|O_f\| = \|O_g\| = 1$, by the triangle inequality we can bound it as
\begin{align*}
    & \leq 2 \sum_{i=0}^{\ell-1} \Ex_{f\gets\cF}\left[ \|(I - \ol{\Pi}_{f}^{\otimes k})\ket{\psi^g_i}\| \right] \\
    & = 2 \sum_{i=1}^{\ell-1} \Ex_{f\gets\cF}\left[ \sqrt{\bra{\psi^g_i}(I - \ol{\Pi}_{f}^{\otimes k})\ket{\psi^g_i}} \right] \\
    & \leq 2 \sum_{i=0}^{\ell-1} \sqrt{\Ex_{f\gets\cF}\left[ \bra{\psi^g_i}(I - \ol{\Pi}_{f}^{\otimes k})\ket{\psi^g_i} \right]}  & \text{(Jensen's inequality)} \\
    & \leq 2 \sum_{i=0}^{\ell-1} \sqrt{\Ex_{f\gets\cF}\left[ \bra{\psi^g_i}\sum_{j = 1}^k \Pi_{f,j}\ket{\psi^g_i} \right]}  \\
    & \leq 2 \sum_{i=0}^{\ell-1} \sqrt{\sum_{j = 1}^k\Ex_{f\gets\cF}\left[ \bra{\psi^g_i} \Pi_{f,j}\ket{\psi^g_i} \right]}  & \text{(linearity of expectation)} \\
    & = 2 \sum_{i=0}^{\ell-1} \sqrt{\sum_{j = 1}^k\bra{\psi^g_i}\frac{I}{N}\ket{\psi^g_i}} 
    = 2\ell\sqrt\frac{k}{N},
\end{align*}
where the first equality holds since $I - \ol{\Pi}_{f}^{\otimes k}$ is a projection operator; 
the third inequality holds due to \eqref{eq:inclu} 
and the second equality holds because the probability of any $x\in[qN]$ being a $1$-preimage of $f$ is $1/N$.
So we have $\Ex_{f\gets\cF}[\Pi_{f,j}] = I/N$ for every $j$. Finally, by Lemma~\ref{lem:2-norm_to_trace}, we have 
\[
\Ex_{f\gets\cF}\left[ \Delta\left( \ket{\psi^f}\bra{\psi^f},\ket{\psi^g}\bra{\psi^g} \right) \right] 
\leq 2\ell\sqrt{\frac{k}{N}}.
\]
as desired.
\end{proof}

\begin{figure}[ht]
    \centering
    {
\newcommand{\www}{1.5} 
\newcommand{\sss}{2} 

\newcommand{\xone}{0}
\newcommand{\xtwo}{\xone+\www}

\newcommand{\xthree}{\xtwo+\sss}
\newcommand{\xfour}{\xthree+\www}

\newcommand{\xfive}{\xfour+\sss}
\newcommand{\xsix}{\xfive+\www}

\newcommand{\xseven}{\xsix+\sss}
\newcommand{\xeight}{\xseven+\www}

\newcommand{\xnine}{\xeight+\sss}
\newcommand{\xten}{\xnine+\www}

\begin{adjustbox}{max width = \textwidth}
\begin{tikzpicture}
  
  \node at (0.75, 3.5) {(a) $\Pi_i$};
  
  \filldraw [black] (\xone,3) circle (1.5pt);
  \filldraw [black] (\xone,2) circle (1.5pt) node[anchor=east] {$\bar{x}_i$};
  \filldraw [black] (\xone,1) circle (1.5pt);
  \filldraw [black] (\xone,0) circle (1.5pt);

  \filldraw [black] (\xtwo,3) circle (1.5pt);
  \filldraw [black] (\xtwo,2) circle (1.5pt);
  \filldraw [black] (\xtwo,1) circle (1.5pt) node[anchor=west] {$\bar{x}_{i+1}$};
  \filldraw [black] (\xtwo,0) circle (1.5pt);
  
  \draw [very thick] (\xone,3) -- (\xtwo,2);
  \draw [very thick] (\xone,2) -- (\xtwo,1);
  \draw [very thick] (\xone,1) -- (\xtwo,0);
  \draw [very thick] (\xone,0) -- (\xtwo,3);

  \node at (4.25, 3.5) {(b) $\tilde{\Pi}_i$};
  
  \filldraw [black] (\xthree,3) circle (1.5pt);
  \filldraw [black] (\xthree,2) circle (1.5pt) node[anchor=east]{$\bar{x}_{i}$};
  \filldraw [black] (\xthree,1) circle (1.5pt);
  \filldraw [black] (\xthree,0) circle (1.5pt);

  \filldraw [black] (\xfour,3) circle (1.5pt);
  \filldraw [black] (\xfour,2) circle (1.5pt);
  \filldraw [black] (\xfour,1) circle (1.5pt) node[anchor=west] {$\bar{x}_{i+1}$};
  \filldraw [black] (\xfour,0) circle (1.5pt);
  
  \draw [very thick] (\xthree,2) -- (\xfour,1);

  \node at (7.75, 3.5) {(c) $\Pi^{R}_{i}$};
  
  \filldraw [black] (\xfive,3) circle (1.5pt);
  \filldraw [black] (\xfive,2) circle (1.5pt) node[anchor=east] {$\bar{x}_{i}$};
  \filldraw [black] (\xfive,1) circle (1.5pt) node[anchor=east] {$\bar{x}'_{i}$};
  \filldraw [black] (\xfive,0) circle (1.5pt);

  \filldraw [black] (\xsix,3) circle (1.5pt);
  \filldraw [black] (\xsix,2) circle (1.5pt) node[anchor=west] {$\bar{x}'_{-i}$};
  \filldraw [black] (\xsix,1) circle (1.5pt) node[anchor=west] {$\bar{x}_{i+1}$};
  \filldraw [black] (\xsix,0) circle (1.5pt);
  
  \draw [very thick] (\xfive,3) -- (\xsix,3);
  \draw [very thick] (\xfive,2) -- (\xsix,2);
  \draw [very thick] (\xfive,1) -- (\xsix,1);
  \draw [very thick] (\xfive,0) -- (\xsix,0);
  \node at (11.25, 3.5) {(d) $\tilde{\Pi}^{\prime}_{i}$};
  
  \filldraw [black] (\xseven,3) circle (1.5pt);
  \filldraw [black] (\xseven,2) circle (1.5pt) node[anchor=east] {$\bar{x}_{i}$};
  \filldraw [black] (\xseven,1) circle (1.5pt) node[anchor=east] {$\bar{x}'_{i}$};
  \filldraw [black] (\xseven,0) circle (1.5pt);

  \filldraw [black] (\xeight,3) circle (1.5pt);
  \filldraw [black] (\xeight,2) circle (1.5pt) node[anchor=west] {$\bar{x}'_{-i}$};
  \filldraw [black] (\xeight,1) circle (1.5pt) node[anchor=west] {$\bar{x}_{i+1}$};
  \filldraw [black] (\xeight,0) circle (1.5pt);
  
  \draw [very thick, >=stealth', <->] (\xseven+0.05,3) -- (\xeight-0.05,3);
  \draw [very thick, >=stealth', dashed, <-](\xseven+0.07,2) -- (\xeight-0.05,2);
  \draw [very thick, >=stealth', <->](\xseven+0.06,2-0.04) -- (\xeight-0.06,1.04);
  \draw [very thick, >=stealth', dashed, ->](\xseven+0.05,1) -- (\xeight-0.07,1);
  \draw [very thick, >=stealth', <->](\xseven+0.05,0) -- (\xeight-0.05,0);
  
  \node at (14.75, 3.5) {(e) $H(\tilde{\Pi}^{\prime}_{i})$};
  
  \filldraw [black] (\xnine,3) circle (1.5pt);
  \filldraw [black] (\xnine,2) circle (1.5pt) node[anchor=east] {$\bar{x}_{i}$};
  \filldraw [black] (\xnine,1) circle (1.5pt) node[anchor=east] {$\bar{x}'_{i}$};
  \filldraw [black] (\xnine,0) circle (1.5pt);

  \filldraw [black] (\xten,3) circle (1.5pt);
  \filldraw [black] (\xten,2) circle (1.5pt) node[anchor=west] {$\bar{x}'_{-i}$};
  \filldraw [black] (\xten,1) circle (1.5pt) node[anchor=west] {$\bar{x}_{i+1}$};
  \filldraw [black] (\xten,0) circle (1.5pt);
  
  \draw [very thick, >=stealth', <->] (\xnine+0.05,3) -- (\xten-0.05,3);
  \draw [very thick, >=stealth', <->](\xnine+0.06,2-0.04) -- (\xten-0.06,1.04);
  \draw [very thick, >=stealth', <->](\xnine+0.06,1+0.04) -- (\xten-0.06,2-0.04);;
  \draw [very thick, >=stealth', <->](\xnine+0.05,0) -- (\xten-0.05,0);
  
\end{tikzpicture}
\end{adjustbox}
}
    \caption{
    Schematic diagrams of permutation in the hybrid proof. 
    (a) $\Pi_i$ is the permutation given by the problem. 
    (b) $\tilde{\Pi}_i$ maps $\bar{x}_{i}$ to $\bar{x}_{i+1} := \Pi_i(\bar{x}_{i})$ and maps other inputs to a dummy image $0$. 
    (c) $\Pi^{R}_{i}$ is a random permutation that is independent of $\Pi_i$. 
    (d) $\tilde{\Pi}'_{i}$ merges $\tilde{\Pi}_i$ and $\Pi^{R}_{i}$.
    It maps $\bar{x}_{i}$ to $\bar{x}_{i+1}$, and maps other input $x$ to $\Pi^{R}_{i}(x)$. 
    There is a collision on inputs $\bar{x}_{i}$ and $\bar{x}'_{i} := {{\Pi}^{R}}^{-1}_{i}(\bar{x}_{i+1})$.
    When executing the ``inverse'' of ${\tilde{{\Pi}^{\prime}}}_{i}$,
    it follows the rules of ${\Pi^{R}}^{-1}_{i}$. Note that the ``inverse'' is not exactly equal to ${\tilde{{\Pi}^{\prime}}}^{-1}_{i}$. 
    (e) The truth table of $H(\tilde{\Pi'}_{i})$ is equal to $\tilde{\Pi}'_{i}$ except that $H(\tilde{\Pi'}_{i})(\bar{x'}_{i}) = \bar{x'}_{-i}$
    \label{fig:hash_chain_permutation}
    }
\end{figure}
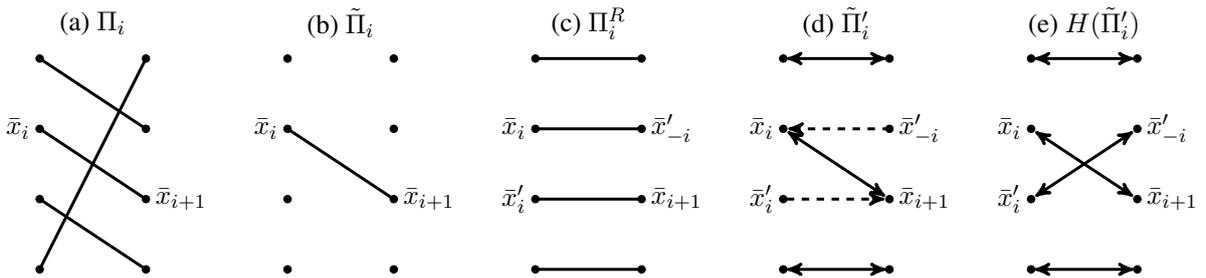

\begin{lemma}\label{lem:pi2pitl}
Use the notations of Definition~\ref{def:perm-notation}. Let $\ell,k$ be integers such that $\ell = O(\polylog(N))$ and $k = O(\polylog(N))$. 
For any quantum algorithm $\cA$ using $\ell$ $k$-parallel queries to $\sPi$
, there is a quantum algorithm $\tilde{\cA}$ using $2\ell$ $k$-parallel queries to $\sPitl$ such that for all $\sPitl$, 
$$ \E_{\sPi \in F^{-1}(\sPitl)}\left[\Pr[ \cA^{\sPi} \neq \tilde{\cA}^{\sPitl} ]\right] = \probone.$$
\end{lemma}
\begin{proof}
Let $\Pi^R_1,\Pi^R_2,\dots,\Pi^R_q$ be permutations of $N$ elements. Let $\sPi^R$ be the corresponding unitary.

We construct $\tilde{\cA}$ as follows: it samples a uniformly random $\sPi^R$ and runs $\cA$ but replaces every query to $\sPi$ with $\sPitl'$ constructed from uniformly random $\sPi^R$, where $\sPitl'$ is defined as
\begin{align}
    \sPitl'\ket{i,x,r}:=\begin{cases} 
    \ket{i,x,r\oplus \Pitl_i(x)} &,\Pitl_i(x)\neq 0, i>0\\
    \ket{i,x,r\oplus \Pi^R_i(x)} &,\Pitl_i(x)= 0, i>0 \\
    \ket{i,x,r\oplus \Pitl^{-1}_{|i|}(x)} &,\Pitl^{-1}_{|i|}(x)\neq 0, i<0\\
    \ket{i,x,r\oplus {\Pi_{|i|}^R}^{-1}(x)} &,\Pitl^{-1}_{|i|}(x)= 0, i<0.
    \end{cases}
\end{align}
One query to $\sPitl'$ can be constructed from two queries to $\sPitl$ coherently by doing the obvious classical calculation and uncomputing the garbage. Therefore $\tilde{\cA}$ uses 2$\ell$ queries to $\sPitl$ as required.

Now we prove that it is very hard to distinguish $\cA$ from $\tilde{\cA}$. This is done by a reduction to the hardness of the modified Grover's search algorithm in Lemma~\ref{lem:k-grover}.

For all $i \in [q]$, define $\xb'_i$ to be $x$ such that $\Pitl'(x)=\xb_{i+1}$ and $x\neq \xb$ . For all $i \in [-q]$, define $\xb'_{i}$ to be $x$ such that $\Pitl'^{-1}_{|i|}(x)=\xb_{|i|}$ and $x\neq \xb_{|i|+1}$. Note that $\xb'_i$ does not exist if $\Pi^R_i(\xb_i)=\Pitl_i(\xb_i)$.

Note that each $\Pitl'_i$ looks like a random permutation except  on the collisions  $\{\xb'_i\}$.  We define a function $H$ which relates these similar $\sPi$ and $\sPitl'$:
\begin{align}
    H(\sPitl')\ket{i,x,r}:=\begin{cases} 
    \ket{i,x,r\oplus \Pitl_i(x)} &,\Pitl_i(x)\neq 0, i>0\\
    \ket{i,x,r\oplus \Pi_i^R(x)} &,\Pitl_i(x)= 0, x\neq \xb'_i, i>0 \\
    \ket{i,x,r\oplus \xb'_{-i}} &,x= \xb'_i, i>0 \\
    \ket{i,x,r\oplus \Pitl^{-1}_{|i|}(x)} &,\Pitl^{-1}_{|i|}(x)\neq 0, i<0\\
    \ket{i,x,r\oplus {\Pi_{|i|}^R}^{-1}(x)} &,\Pitl^{-1}_{|i|}(x)= 0, x\neq \xb'_i, i<0 \\
    \ket{i,x,r\oplus \xb'_{-i}} &,x= \xb'_i, i<0.
    \end{cases}
\end{align}

It is easy to check that for all $\sPitl'$, $H(\sPitl')$ is a valid $\sPi$, and for a given $\sPi$, there are $N^q$ elements in $H^{-1}(\sPi)$. It is also easy to verify that for a fixed $\sPitl$, different $H^{-1}(\sPi)$ partitions all possible $\sPitl'$.

Let $\rho_{\sPi}$ be the final density matrix of $\cA^{\sPi}$. Let $\rho_{\sPitl'}$ be the final density matrix of $\tilde{\cA}^{\sPitl}$ with a fixed $\sPi'$. By the strong convexity of trace distance, we have

\begin{align}\label{eq:first-hybrid-convex}
  \Delta\left(\E_{\sPi \in  F^{-1}(\sPitl)}[\rho_{\sPi}], \E_{\sPitl' \in  H^{-1}(F^{-1}(\sPitl))}[\rho_{\sPitl'}]  \right)  
  \leq \E_{\sPi \in  F^{-1}(\sPitl)} \left[\Delta\left(\rho_{\sPi}, \E_{\sPitl' \in  H^{-1}(\sPi)}[\rho_{\sPitl'}]  \right)\right].
\end{align}

Finally, we prove that $\Delta(\rho_{\sPi}, \E_{\sPitl' \in  H^{-1}(\sPi)}[\rho_{\sPitl'}]  )=\probone$ for all $\sPi$ by a reduction to the modified Grover search problem. Consider a Grover oracle $G$ defined in Lemma~\ref{lem:k-grover} that might be $f$ or $g$. Given free calls to $\sPi$, we use two calls to $G$ to construct an oracle $\sPi^G$ which equals $\sPi$ when $G=g$ and equals a random $\sPitl'$ when $G=f$. The construction is as follows: 
\begin{align}
    \sPi^G\ket{i,x,r}:=\begin{cases} 
    \ket{i,x,r\oplus \Pitl_i(x)} &,\Pitl_i(x)\neq 0, i>0\\
    \ket{i,x,r\oplus \Pi_i(x)} &,\Pitl_i(x)= 0, G((i-1)N+x)=0, i>0 \\
    \ket{i,x,r\oplus \xb_{i+1}} &,\Pitl_i(x)= 0, G((i-1)N+x)=1, i>0 \\
    \ket{i,x,r\oplus \Pitl^{-1}_{|i|}(x)} &,\Pitl^{-1}_{|i|}(x)\neq 0, i<0\\
   \ket{i,x,r\oplus \Pi^{-1}_{|i|}(x)} &,\Pitl^{-1}_{|i|}(x)= 0, G((i-1)N+\Pi^{-1}_{|i|}(x))=0, i<0 \\
    \ket{i,x,r\oplus \xb_{i}} &,\Pitl^{-1}_{|i|}(x)= 0, G((i-1)N+\Pi^{-1}_{|i|}(x))=1, i<0. \\
    \end{cases}
\end{align}
By Lemma~\ref{lem:k-grover}, if we try to distinguish $f$ from $g$ by distinguishing $\cA^{\sPi^G}$ of the two cases, we can only succeed with probability $O\left(\ell\sqrt{\frac{k}{N}}\right)$ since we only have $O(\ell)$ $k$-parallel queries to $G$. Therefore, one can only distinguish $\sPi$ from $\sPitl'$ with  probability $\probone$, \ie 
\begin{align}
  \Delta\left(\rho_{\sPi}, \E_{\sPitl' \in  H^{-1}(\sPi)}\left[\rho_{\sPitl'}\right]  \right)=\probone.  
\end{align}
 Then by \eqref{eq:first-hybrid-convex}, 
 \begin{align}
     \Delta\left(\E_{\sPi \in  F^{-1}(\sPitl)}\left[\rho_{\sPi}\right], \E_{\sPitl' \in  H^{-1}(F^{-1}(\sPitl))}\left[\rho_{\sPitl'}\right]  \right) =\probone.
 \end{align}
  By the operational interpretation of the trace distance, we have 
  \begin{align}
       \E_{\sPi \in F^{-1}(\sPitl)}\left[\Pr[ \cA^{\sPi} \neq \tilde{\cA}^{\sPitl} ]\right] =\probone.
  \end{align} 
\end{proof}

\begin{lemma} \label{lem:pitl-depth-bound}
Use the notations of Definition~\ref{def:perm-notation}. Let $\ell,k$ be integers such that $\ell \in [q]$, $k=O(\polylog(N))$. For all algorithm $\cA$ using $\ell$ $k$-parallel queries to $\sPitl$, w.l.o.g. we can assume the final output of $\cA$ has the form 
$$\ket{\psi}=U_{\ell} \sPitlk U_{\ell-1}\sPitlk \dots U_2 \sPitlk U_1 \sPitlk \ket{\psi_0}.$$
For all $m,p \in [\ell]$, $p\leq m$, define the hybrid state
$$\ket{{\psi}_{m,p}}:=U_{m} \sPitlk U_{m-1}\sPitlk \dots U_{p+1} \sPitlk U_{p} \sPitlk_{p} U_{p-1}\sPitlk_{p-1} \dots U_2 \sPitlk_2 U_1 \sPitlk_1 \ket{\psi_0}.$$
Then for all $\cA$, 
\begin{align}
    \E_{\sPitl} \norm{\ket{\psi}-\ket{{\psi}_{\ell,\ell}}}\leq 2\ell \sqrt{\frac{k}{N}}.
\end{align}
\end{lemma}
\begin{proof}
Note that $\ket{\psi}=\ket{\psi_{\ell,0}}$. By triangle inequality we have 
\begin{align}\label{eq:lem2-hybrid}
     \E_{\sPitl}  \norm{\ket{\psi}-\ket{\psi_{\ell,\ell}}} &\leq \sum_{i=1}^\ell  \E_{\sPitl}  \norm{\ket{\psi_{\ell,i-1}}-\ket{\psi_{\ell,i}}}.
\end{align}
We proceed by proving $\E_{\sPitl}  \norm{\ket{\psi_{\ell,i-1}}-\ket{\psi_{\ell,i}}}=2\sqrt{k/N}$ for all $i\in[\ell]$. Note that 
\begin{align}\label{eq:psii-diff}
 \norm{\ket{\psi_{\ell,i-1}}-\ket{\psi_{\ell,i}}} 
 &= \norm{\ket{\psi_{i,i-1}}-\ket{\psi_{i,i}}}  \nn \\
 &= \norm{U_i \sPitlk \ket{\psi_{i-1,i-1}}-U_i \sPitlk_i\ket{\psi_{i-1,i-1}}}  \nn \\
 &= \norm{ (\sPitlk - \sPitlk_i)\ket{\psi_{i-1,i-1}}}. 
\end{align}

Note that $\sPitl\ket{j,x,r}$ and $\sPitl_i\ket{j,x,r}$ differs only when $j >i$ and $x=\xb_j$ or $j < -i$ and $x=\xb_{j+1}$. Therefore 
\begin{align}
    (\sPitl-\sPitl_i)\ket{j,x,r}&= (\sPitl-\sPitl_i)P_i\ket{j,x,r} \\
      \sPitl(1-P_i)&=\sPitl_i(I-P_i) \\
      (\sPitlk-\sPitlk_i)(I-P_i)^{\otimes k} &=0
\end{align}
where 
\begin{align}\label{eq:Pi}
    P_i:=\left(\sum_{j=i+1}^k\Proj{j,\xb_j}+\sum_{-j=i+1}^k\Proj{j,\xb_{j+1}}\right)\otimes I.
\end{align}
 Note that $P_i$ actually depends on $\sPitl$, but we omit the dependence for cleaner notation. 
Therefore we have

\begin{align}\label{eq:sPi-diff}
   &\norm{ (\sPitlk-\sPitlk_i) \ket{\psi_{i-1,i-1}}}^2 \nn \\
   =&\norm{ (\sPitlk-\sPitlk_i)(I-(I-P_i)^{\otimes k} ) \ket{\psi_{i-1,i-1}}}^2 \nn \\
   \leq &4\norm{ (I-(I-P_i)^{\otimes k} ) \ket{\psi_{i-1,i-1}}}^2  \nn \\
   = &4\vev{{\psi_{i-1,i-1}}| (I-(I-P_i)^{\otimes k} ) |{\psi_{i-1,i-1}}} \nn \\
   \leq &4\vev{{\psi_{i-1,i-1}}| \sum_{j=1}^k P_i^j |{\psi_{i-1,i-1}}} 
\end{align}
where $P_i^j=I^{\otimes j-1}\otimes P_i \otimes I^{\otimes k-j}$. In the fourth line, we use the fact that $I-(I-P_i)^{\otimes k}$ is a projector, and in the fifth line we use the standard union bound calculation.

By \eqref{eq:Pi}, for all $i\in [q]$ and $j\in [k]$, $P_i^j$ is normalized by
\begin{align}
    \sum_{\dxb_i=0}^{N-1} P_i^j =\left(\sum_{j=i+1}^k+\sum_{-j=i+1}^k\right)\Proj{j}\otimes{\sum_{\xb \in [N]} \Proj{\xb}} \leq  I
\end{align}
where we omitted the tensor product of identities and assume $\dxb_{i+1},\dots,\dxb_q$ to be fixed. Therefore
\begin{align}\label{eq:Pnormal}
    &\sum_{\dxb_i=0}^{N-1}\norm{ (\sPitlk-\sPitlk_i) \ket{\psi_{i-1,i-1}}}^2 \nn \\
    \leq & \sum_{\dxb_i=0}^{N-1}4\vev{{\psi_{i-1,i-1}}| \sum_{j=1}^k P_i^j |{\psi_{i-1,i-1}}}  \nn \\
    \leq & 4k
\end{align}

Finally, combining \eqref{eq:Pnormal}, \eqref{eq:sPi-diff}, and \eqref{eq:psii-diff}, we have
\begin{align}
    &\E_{\sPitl}  \norm{\ket{\psi_{\ell,i-1}}-\ket{\psi_{\ell,i}}} \nn \\
    =& \E_{\dxb_1,\dots,\dxb_{i-1}}\E_{\dxb_i} \E_{\dxb_{i+1},\dots,\dxb_{q}} \norm{ (\sPitlk - \sPitlk_i)\ket{\psi_{i-1,i-1}}} \nn \\
    =& \E_{\dxb_1,\dots,\dxb_{i-1}} \E_{\dxb_{i+1},\dots,\dxb_{q}} \frac{1}{N}\sum_{\dxb_i}\norm{ (\sPitlk - \sPitlk_i)\ket{\psi_{i-1,i-1}}} \nn \\
    \leq& \E_{\dxb_1,\dots,\dxb_{i-1}} \E_{\dxb_{i+1},\dots,\dxb_{q}} \frac{1}{N}\sqrt{\sum_{\dxb_i}\norm{ (\sPitlk - \sPitlk_i)\ket{\psi_{i-1,i-1}}}^2}\cdot\sqrt{N} \nn \\
    =&2\sqrt{\frac{k}{N}}.
\end{align}
Plugging back to \eqref{eq:lem2-hybrid} we have
\begin{align}
    \E_{\sPitl}  \norm{\ket{\psi}-\ket{\psi_{\ell,\ell}}} &\leq 2\ell \sqrt{\frac{k}{N}}.
\end{align}

\end{proof}

\begin{corollary} \label{cor:q-1}
Any algorithm $\cA$ using $(q-1)$ $k$-parallel queries to $\sPitl$ can only output $\xb_{q+1}$ with  probability $O\left(\ell \sqrt{\frac{k}{N}}\right)$.
\end{corollary}
\begin{proof}
By Lemma~\ref{lem:2-norm_to_trace} and Lemma~\ref{lem:pitl-depth-bound}, output probability of $\ket{\psi}$ and $\ket{{\psi}_{\ell,\ell}}$ only differ by $O(\ell \sqrt{k/N})$. Since $\ket{\psi_{q-1,q-1}}$ does not depend on $\Pitl_q$, it can only find $\xb^{q+1}$ with probability $1/N$. Thus the maximum probability of $\ket{\psi}$ finding $\xb^{q+1}$  is $O(\ell \sqrt{k/N})+1/N=O(\ell \sqrt{k/N})$.
\end{proof}

\begin{proof}[Proof of Theorem~\ref{thm:hashchain_main}]
Let $\cA$ be an algorithm that uses $\floor{{(q-1)/2}}$ $k$-parallel queries to $\sPi$ and outputs $\xb_{q+1}$ with some probability $p$. By Lemma~\ref{lem:pi2pitl}, there is an algorithm $\tilde{\cA}$ that uses $(q-1)$ $k$-parallel queries to $\sPitl$ and outputs $\xb_{q+1}$ with probability at least $p-\proboneq$. By Corollary~\ref{cor:q-1}, it holds that $p-\proboneq=\proboneq$. Therefore, we have $p=\proboneq$.
\end{proof}

\section{Parallel Hardness of Twisted Hash Chains}
\label{sec:t-hashchain}
The (standard) hash chain problem is a natural candidate for parallel hardness. However, hash functions (modeled as random functions) with equal-length inputs and outputs are not injective with overwhelming probability. Importantly, quantum operations need to be reversible. Therefore, if we evaluate the standard hash chain straight-forwardly, the intermediate hash values are required to be stored in ancilla qubits. Consequently, the width of the circuit would become proportional to the length of the chain, which means the computation could not be done within width $\lambda$. 

Inspired by the construction of the Feistel cipher which implements a permutation by hash functions, we introduce the \emph{twisted hash chain} problem. Let $\cX$ be $\bits^n$. An \emph{$s$-chain} is a sequence $x_0,x_1,\dots,x_s \in \cX$ such that $x_i = H(x_{i-1}) \oplus x_{i-2}$ for $i\in [s]$ where we use the convention that $x_{-1} \coloneqq 0^n$. 
Informally, the task of a $q$-query $k$-parallel algorithm that interacts with a random oracle $H:\cX\to\cY$ with $|\cX| = |\cY|$ is to output $x_0,x_q,x_{q+1}\in\cX$ of a $(q+1)$-chain, \ie there exists a sequence $x_1,\dots,x_{q-1} \in \cX$ such that $x_i = H(x_{i-1}) \oplus x_{i-2}$ for $i\in[q+1]$. The computation of a twisted hash chain is shown in Fig.~\ref{fig:feistel}.

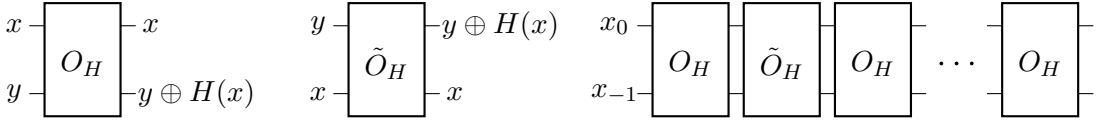
\begin{figure}[h]
  \begin{adjustbox}{max width = \textwidth}
\begin{tikzpicture}
  
  \draw (-0.2, 0.3) -- (1.2, 0.3);
  \node at (-0.4, 0.3) {$y$};
  \node at (2.0, 0.3) {$y\oplus H(x)$};
  
  \draw (-0.2, 1.2) -- (1.2, 1.2);
  \node at (-0.4, 1.2) {$x$};
  \node at (1.4, 1.2) {$x$};
  
  \filldraw [color=black, fill=white, thick] (0,0) rectangle (1, 1.5);
  \node at (0.5, 0.7){$O_H$};


   
   \draw (3.8, 0.3) -- (5.2, 0.3);
   \node at (3.6, 0.3) {$x$};
   \node at (5.4, 0.3) {$x$};
   
   \draw (3.8, 1.2) -- (5.2, 1.2);
   \node at (3.6, 1.2) {$y$};
   \node at (6, 1.2) {$y\oplus H(x)$};
   
   \filldraw [color=black, fill=white, thick] (4,0) rectangle (5, 1.5);
   \node at (4.5, 0.7){$\tilde{O}_H$};


  \draw (7.8, 1.2) -- (11.6, 1.2);
  \draw (7.8, 0.3) -- (11.6, 0.3);
  \node at (7.5, 1.2) {$x_0$};
  \node at (7.5, 0.3) {$x_{-1}$};
  
  \filldraw [color=black, fill=white, thick] (8,0) rectangle (9, 1.5);
  \node at (8.5, 0.7){${O}_H$};
   
  \filldraw [color=black, fill=white, thick] (9.2,0) rectangle (10.2, 1.5);
  \node at (9.7, 0.7){$\tilde{O}_H$};
   
  \filldraw [color=black, fill=white, thick] (10.4,0) rectangle (11.4, 1.5);
  \node at (10.9, 0.7){${O}_H$};

  \filldraw[black] (11.8,0.7) circle (0.5pt);
  \filldraw[black] (12,0.7) circle (0.5pt);
  \filldraw[black] (12.2,0.7) circle (0.5pt);
   
  \draw (12.4, 1.2) -- (13.8, 1.2);
  \draw (12.4, 0.3) -- (13.8, 0.3);
  \filldraw [color=black, fill=white, thick] (12.6,0) rectangle (13.6, 1.5);
  \node at (13.1, 0.7){${O}_H$};

\end{tikzpicture}
\end{adjustbox}
  \caption{Schematic diagram of the twisted hash chain.}
  \label{fig:feistel}
\end{figure}

\noindent In this section, we aim to prove the following theorem.
\begin{theorem}[Twisted hash chain is sequential]\label{thm:twisted_hash_chain}
For any $k$-parallel $q$-query oracle algorithm $\cC$, the probability $p_\cC$ (parameterized by $k$ and $q$) that $\cC$ outputs $x_0,x_q,x_{q+1}\in\cX$ satisfying the following condition:
\begin{itemize}
    \item there exist $x_1,\dots,x_{q-1}\in\cX$ such that $H(x_{i-1}) = x_i \oplus x_{i-2}$ for $i\in[q+1]$, where $x_{-1} \coloneqq 0^n$
\end{itemize}
is at most $F(k,2q) = O(k^4q^4/|\cY|)$, where the function $F$ is defined in Lemma~\ref{lem:chain_all_but_last}.
\end{theorem}

Toward proving the hardness, we exploit the framework of \cite{chung2021compressed}. Below, we borrow the notations and definitions of \cite{chung2021compressed}. Let $H\colon\cX\to\cY$ be a random oracle. Let $\hat{\cY}$ be the dual group of $\cY$. Let $\ol{\cY}$ denote the set $\cY\cup\set{\bot}$. We say that $D:\cX\to\ol{\cY}$ is a \emph{database}. By $\mathfrak{D}$ we mean the set of all databases, \ie the set of all functions from $\cX$ to $\ol{\cY}$. 
For any tuple $\bfx = (x_1,\dots,x_k)$ with pairwise disjoint $x_i\in\cX$, tuple $\bfr = (r_1,\dots,r_k)\in\ol{\cY}^k$ and database $D\in\mathfrak{D}$, we define the database $D[\bfx\mapsto\bfr]$ as 
\[
D[\bfx\mapsto\bfr](x) \coloneqq \begin{cases}
    r_i & \text{ if } x = x_i \text{ for some } i\in[k] \\
    D(x) & \text{ if } x \notin \set{x_1,\dots,x_k}.
\end{cases}
\]
By database property $\mathsf{P}$ we mean a set of databases, that is $P \subseteq \mathfrak{D}$. In this section, we assume that $\cX = \cY$. 
For any database $D\in\mathfrak{D}$ and tuple $\bfx = (x_1,\dots,x_k)$ of pairwise distinct $x_i\in\cX$, we let 
\[
D|^\bfx \coloneqq \set{D[\bfx\mapsto\bfr] \mid \bfr \in \ol{\cY}^k} \subseteq \mathfrak{D}
\]
be the set of databases that coincide $D$ outside of $\bfx$.
Furthermore, for any database property $\mathsf{P}\subseteq\mathfrak{D}$, we let
\[
\mathsf{P}|_{D|^\bfx} \coloneqq \mathsf{P} \cap D|^\bfx.
\]

\begin{definition}[Definition~5.5 in~\cite{chung2021compressed}]
    Let $\mathsf{P},\mathsf{P}'$ be two database properties. Then, the \emph{quantum transition capacity} (of order $k$) is defined as
    \[
        \llbracket \mathsf{P} \xrightarrow{k} \mathsf{P}' \rrbracket 
        \coloneqq \max_{\bfx,\hat{\bfy},D}\|\mathsf{P}'|_{D|^\bfx} \mathsf{cO_{\bfx\hat{\bfy}}}
        \mathsf{P}|_{D|^\bfx}\|,
    \]
    where the maximum is over all possible $\bfx \in \cX^k$, $\hat{\bfy} \in \hat{\cY}^k$ and $D \in \mathfrak{D}$.
    Furthermore, we define
    \[
        \llbracket \mathsf{P} \xRightarrow{k} \mathsf{P}' \rrbracket 
        \coloneqq \sup_{U_1,\dots,U_{q-1}} \left\| \mathsf{P}' U_{q-1}\mathsf{cO}^k U_{q-1} \mathsf{cO}^k \dots U_1\mathsf{cO}^k \mathsf{P} \right\|,
    \]
    where $\|\cdot\|$ is the operator norm; the supremum is over all positive $d\in\mathbb{Z}$ and all unitaries $U_1,\dots,U_{q-1}$ acting on $\C[\cX]\otimes\C[\cY]\otimes\C^d$.\footnote{Namely, over all $q$-query quantum algorithms.}
    For the formal definitions of $\mathsf{cO_{\bfx\hat{\bfy}}}$ and $\mathsf{cO}^k$, we refer to \cite{chung2021compressed}. We note that their definitions are not required for the following proof.
\end{definition}

\begin{definition}
The database property \emph{twisted hash chain of length $s$}, denoted by $\mathsf{TCHN^s}$, is defined as
\[
\mathsf{TCHN^s} 
\coloneqq \set{D \mid \exists x_0,x_1,\dots,x_s \in \cX \colon  x_i = D(x_{i-1})\oplus x_{i-2}, \forall i \in [s]}  \subseteq \mathfrak{D},
\]
where we use the convention that $x_{-1} \coloneqq 0^n$ for convenience.
\end{definition}

\begin{definition}[Definition~5.20 in \cite{chung2021compressed}, with $\ell$ fixed to $1$]
A database transition $\mathsf{P}\to\mathsf{P}'$ is said be \emph{$k$-non-uniformly weakly recognizable} by $1$-local properties\footnote{We refer to Definition~5.10 in \cite{chung2021compressed} for the formal description of local properties.} if for every $\bfx = (x_1,\dots,x_k)$ with pairwise disjoint entries, and for every $D \in \mathfrak{D}$, there exists a family of $1$-local properties $\set{L^{D,\bfx}_i}$ where each $L^{D,\bfx}_i \subseteq \ol{\cY}$ and the support of $L^{D,\bfx}_i$ is $\set{x_i}$ or empty, so that
\[
D[\bfx\mapsto\bfr] \in \mathsf{P} \land [\bfx\mapsto\bfu] \in \mathsf{P}'
\implies \exists i \colon u_i \in L^{D,\bfx}_i \land r_i \neq u_i.
\]
\end{definition}

\begin{theorem}[Theorem~5.23 in~\cite{chung2021compressed}]\label{thm:5.23}
Let $\mathsf{P}$ and $\mathsf{P}'$ be $k$-non-uniformly weakly recognizable by $1$-local properties $L^{\bfx,D}_i$, where the support of $L^{\bfx,D}_i$ is $\set{x_i}$ or empty. Then
\[
\llbracket \bot \xRightarrow{q,k} \mathsf{TCHN}^{q+1}\rrbracket 
\leq \max_{\bfx,D}e \sum_{i}\sqrt{10P[U\in L^{\bfx,D}_i]},
\]
where $U$ is defined to be uniformly random in $\cY$ and $\bot \coloneqq \set{D \mid D(x)=\bot \text{  for all } x \in\cX}$.
\end{theorem}

Here, we define the family of local properties for our purpose. For any $D \in \mathfrak{D}$ and any $\bfx \coloneqq (x_1,\dots,x_k)$ $\in \cX^k$  with disjoint entries, we defined the following $1$-local properties $L^{D,\bfx}_i \subseteq \ol{\cY}$\footnote{Recall that we assume $\cX$ = $\cY$.} with support $\set{x_i}$ for $i\in[k]$ as
\[
L^{D,\bfx}_i \coloneqq L^{D,\bfx}_{i,1} \cup L^{D,\bfx}_{i,2},
\]
where
\[
L^{D,\bfx}_{i,1} \coloneqq \set{x\in\cX \mid D(x) \neq \bot \lor x \in \set{x_1,\dots,x_k}}
\]
and
\[
L^{D,\bfx}_{i,2} \coloneqq 
\set{x \in \cX \mid \exists x',x''\in L^{D,\bfx}_{i,1} \colon x = x' \oplus x''}.
\]

The following lemma, in line with Lemma~2.1 in~\cite{chung2021compressed}, shows that the local properties $\set{L^{D,\bfx}_i}$ recognize the database transition $\neg\mathsf{TCHN^s} \to \mathsf{TCHN^{s+1}}$, and allows us to exploit Theorem~\ref{thm:5.23}.
First, we briefly explain the intuition. We pick an arbitrary $(s+1)$-chain in $D[\bfx\mapsto\bfu]$ and call it the \emph{new chain}. We denote the elements of the new chain by $\hat{x}_0,\hat{x}_1,\dots,\hat{x}_{s+1}$.
There are two possible consequences of this database transition:

First, some branch $x_i$ of the query $\bfx$ becomes the first elements $\hat{x}_0$ of the new chain, \ie $x_i = \hat{x}_0$ and $x_i$ is responded with $D[\bfx\mapsto\bfu](x_i) = u_i$ such that $u_i = D[\bfx\mapsto\bfu](\hat{x}_0) = \hat{x}_1$.
In addition, $D[\bfx\mapsto\bfu](\hat{x}_1) = \hat{x}_0 \oplus \hat{x}_2 \neq \bot$.
This means that $\hat{x}_1$ must either already be sampled ($D(\hat{x}_1) \neq \bot$) or be one of the branch of $\bfx$ ($\hat{x}_1 \in\set{x_1,\dots,x_k}$) (or both). 
In other words, $u_i$ must be in $L^{D,\bfx}_{i,1}$.

Second, some branch $x_i$ of the query $\bfx$ becomes the $(j+1)$-th elements $\hat{x}_j$ ($1\leq j \leq s$) of the new chain, \ie $x_i = \hat{x}_j$ and $x_i$ is responded with $D[\bfx\mapsto\bfu](x_i) = u_i$ such that $u_i = D[\bfx\mapsto\bfu](\hat{x}_j) = \hat{x}_{j-1} \oplus \hat{x}_{j+1}$.
Similarly, we can conclude that either $D(\hat{x}_{j-1}) \neq \bot$ or $\hat{x}_{j-1} \in \set{x_1,\dots,x_k}$ (or both) and either $D(\hat{x}_{j+1}) \neq \bot$ or $\hat{x}_{j+1} \in \set{x_1,\dots,x_k}$ (or both).
This means that $u_i$ must be the XOR of two elements in $L^{D,\bfx}_{i,1}$. That is, $u_i \in L^{D,\bfx}_{i,2}$.

The intuition above can be formalized as the following lemma.
\begin{lemma}\label{lem:local_properties}
    $D[\bfx\mapsto\bfr]\notin\mathsf{TCHN^s}\land D[\bfx\mapsto\bfu] \in \mathsf{TCHN^{s+1}} 
    \implies \exists i\in[k] \colon r_i \neq u_i \land u_i \in L^{D,\bfx}_i$.
\end{lemma}
\begin{proof}
    Suppose $D[\bfx\mapsto\bfu] \in\mathsf{TCHN^{s+1}}$ and let $\hat{x}_0,\hat{x}_1,\dots,\hat{x}_{s+1}$ be such a chain, \ie $\hat{x}_1 = D[\bfx\mapsto\bfu](\hat{x}_0)$ and $\hat{x}_{j+2} = D[\bfx\mapsto\bfu](\hat{x}_{j+1}) \oplus \hat{x}_j$ for $j = 0,1,\dots,q-1$.
    Let $s_\circ$ be the smallest $j$ such that $D[\bfx\mapsto\bfr](\hat{x}_{s_\circ}) \neq D[\bfx\mapsto\bfu](\hat{x}_{s_\circ})$. If $s_\circ \geq s$ or $j$ does not exist, then $D[\bfx\mapsto\bfr] \in \mathsf{TCHN^{s+1}}$, and we are done. 
    Suppose now $0\leq s_\circ\leq s-1$. Since $D[\bfx\mapsto\bfu]$ and $D[\bfx\mapsto\bfu]$ are identical outside of $\bfx$, there exists a coordinate $i$ of $\bfx$ such that $x_i = \hat{x}_{s_\circ}$. 
    Therefore, $u_i = D[\bfx\mapsto\bfu](x_i) = D[\bfx\mapsto\bfu](\hat{x}_{s_\circ}) \neq D[\bfx\mapsto\bfr](\hat{x}_{s_\circ}) = D[\bfx\mapsto\bfr](x_i) = r_i$.
    
    Below, we divide the analysis into three cases according to the value of $s_\circ$:
    \begin{enumerate}
        \item 
        If $s_\circ = 0$, we have $u_i = D[\bfx\mapsto\bfu](\hat{x}_{s_\circ}) = D[\bfx\mapsto\bfu](\hat{x}_0) = \hat{x}_1$.
        And $D[\bfx\mapsto\bfu](\hat{x}_1) = \hat{x}_2 \oplus \hat{x}_0 \neq \bot$, which implies either $D(\hat{x}_1) \neq \bot$ or $\hat{x}_1 \in \set{x_1,\dots,x_k}$ (or both).
        This means $u_i \in L^{D,\bfx}_{i,1}$.
        \item
        If $s_\circ = 1$, we have $u_i = D[\bfx\mapsto\bfu](\hat{x}_{s_\circ}) = D[\bfx\mapsto\bfu](\hat{x}_1) = \hat{x}_2 \oplus \hat{x}_0$.
        And $D[\bfx\mapsto\bfu](\hat{x}_2) = \hat{x}_3 \oplus \hat{x}_1 \neq \bot$, which implies either $D(\hat{x}_2) \neq \bot$ or $\hat{x}_2 \in \set{x_1,\dots,x_k}$ (or both).
        And $D[\bfx\mapsto\bfu](\hat{x}_0) = \hat{x}_1 \neq \bot$, which implies either $D(\hat{x}_0) \neq \bot$ or $\hat{x}_0 \in \set{x_1,\dots,x_k}$ (or both).
        This means $u_i \in L^{D,\bfx}_{i,2}$.
        \item
        If $2 \leq s_\circ \leq s-1$, then $r_i = D[\bfx\mapsto\bfu](\hat{x}_{s_\circ}) = \hat{x}_{s_\circ+1} \oplus \hat{x}_{s_\circ-1}$.
        Similarly, $D[\bfx\mapsto\bfu](\hat{x}_{s_\circ+1}) = \hat{x}_{s_\circ} \oplus \hat{x}_{s_\circ+2} \neq \bot$, which implies either $D(\hat{x}_{s_\circ+1}) \neq \bot$ or $\hat{x}_{s_\circ+1} \in \set{x_1,\dots,x_k}$ (or both).
        And $D[\bfx\mapsto\bfu](\hat{x}_{s_\circ-1}) = \hat{x}_{s_\circ} \oplus \hat{x}_{s_\circ-2} \neq \bot$, which implies either $D(\hat{x}_1) \neq \bot$ or $\hat{x}_1 \in \set{x_1,\dots,x_k}$ (or both).
        This means $u_i \in L^{D,\bfx}_{i,2}$.
    \end{enumerate}
    In all of the above cases, $u_i$ must be in $L^{D,\bfx}_i$ which concludes the proof.
\end{proof}

We need Corollary~4.2 in~\cite{chung2021compressed} which is rephrased from Lemma~5 in~\cite{zhandry2019record}.
\begin{lemma}[Lemma~5 in~\cite{zhandry2019record}]\label{lem:zhandry}
Let $R \subseteq \cX^\ell \times \cY^\ell$ be a relation. 
Let $\cA$ be an algorithm that outputs $\bfx\in\cX^\ell$ and $\bfy\in\cY^\ell$. Let $p$ be the probability that $\bfy = H(\bfx) \coloneqq (H(x_1),\dots,H(x_\ell))$ and $(\bfx,\bfy)\in R$ when $\cA$ has interacted with the standard random oracle, initialized with a uniformly random function $H$.
Similarly, let $p'$ be the probability that $\bfy = D(\bfx) \coloneqq (D(x_1),\dots,D(x_\ell))$ and $(\bfx,\bfy)\in R$ when $\cA$ has interacted with the compressed oracle and $D$ is obtained by measuring its internal state in the computational basis. Then
\[
\sqrt{p} \leq \sqrt{p'} + \sqrt{\frac{\ell}{|\cY|}}.
\]
\end{lemma}

\begin{lemma}\label{lem:capacity}
$\llbracket \bot \xRightarrow{q,k} \mathsf{TCHN}^{q+1}\rrbracket 
\leq  qek \sqrt{\frac{5kq(kq+1)}{|\cY|}}.$
\end{lemma}
\begin{proof}
By Lemma~5.6 in~\cite{chung2021compressed}, we have
\[
\llbracket \bot \xRightarrow{q,k} \mathsf{TCHN}^{q+1}\rrbracket 
\leq \sum_{s=1}^q \llbracket\mathsf{SZ}_{\leq k(s-1)} \setminus \mathsf{TCHN^s} \xrightarrow{k} \mathsf{TCHN^{s+1}}\rrbracket.
\]
Choosing the local properties $\set{L^{D,\bfx}_i}$ as above whenever $D \in \mathsf{SZ}_{\leq k(s-1)}$, and to be constant-false otherwise, Lemma~\ref{lem:local_properties} ensures that we can apply Theorem~5.23 in~\cite{chung2021compressed} to bound quantum transition capacity.
Therefore, applying
Theorem~5.23 in \cite{chung2021compressed}, 
for each $s\in[q]$ we have
\[
\llbracket \mathsf{SZ}_{\leq k(s-1)} \setminus \mathsf{TCHN^s} \xrightarrow{k} \mathsf{TCHN^{s+1}} \rrbracket
\leq e \max_{D,\bfx} \sum_{i=1}^k \sqrt{10\Pr[U\in L^{D,\bfx}_i]} 
\leq ek \sqrt{\frac{5kq(kq+1)}{|\cY|}}.
\]
The last inequality holds because for every $s\in[q]$ and every $D \in \mathsf{SZ}_{\leq k(s-1)}$, it holds that 
\[
|\set{x \in \cX \mid D(x) \neq \bot} \cup \set{x_1\dots,x_k}| \leq k(q-1) + k = kq.
\]
Thus, we have $|L^{D,\bfx}_{i,1}| \leq kq$ and $|L^{D,\bfx}_{i,2}| \leq \binom{kq}{2} = kq(kq-1)/2$ for $i\in[k]$, which then implies $|L^{D,\bfx}_i| \leq kq(kq+1)/2$.
Finally, summing over $s\in[q]$ completes the proof.
\end{proof}

First, note that Lemma~\ref{lem:zhandry} is tailored for algorithms that output \emph{all elements of the chain and their hash values}. However, to obtain a bound when the algorithm $\cA$ is required to output \emph{only $x_0,x_q$ and $x_{q+1}$} is more challenging.
In most of situations, one could define another algorithm $\cB$ that simply runs $\cA$ followed by calculating the whole chain with $q+2$ extra queries. This would increase the number of queries by at most $q+2$. However, this gives us a meaningless bound for the twisted hash chain problem.

Below, we first provide the following lemma for algorithms that do \emph{not} have to output the last hash value $y_{q+1}$. The proof is similar to Theorem~5.9 in~\cite{chung2021compressed}.
\begin{lemma}\label{lem:chain_all_but_last}
For any $k$-parallel $q$-query oracle algorithm $\cA$ that interacts with a standard random oracle, the probability $p_\cA$ (parameterized by $k$ and $q$) that $\cA$ outputs $x_0,x_1,\dots,x_{q+1}\in\cX$ and $y_0,y_1,\dots,y_q \in \cY$ (without $y_{q+1}$) satisfying
\begin{itemize}
    \item $y_i = H(x_i)$ for $0\leq i\leq q$
    \item $y_{i-1} = x_i \oplus x_{i-2}$ for $i \in [q+1]$, where $x_{-1} \coloneqq 0^n$
\end{itemize}
is upper bounded by the function $F(k,q)$ where
\[
F(k,q) \coloneqq
\left( qek \sqrt{\frac{5kq(kq+1)}{|\cY|}} 
+ e(q+2)\sqrt{\frac{5(q+2)(q+3)}{|\cY|}} 
+ \sqrt{\frac{q+2}{|\cY|}} \right)^2
= O\left(\frac{q^4k^4}{|\cY|}\right).
\]
\end{lemma}
\begin{proof}
We define an algorithm $\cB$ that runs $\cA$ and obtains the output $x_0,\dot,x_{q+1},y_0,\dots,y_q$. And then $\cB$ makes a classical query $x_{q+1}$ to the random oracle. Finally, $\cB$ outputs $x_0,\dot,x_{q+1},y_0,\dots,y_q,H(x_{q+1})$.

Let $p_\cB$ be the probability that $\cB$ outputs $x_0,x_1,\dots,x_{q+1}\in\cX$ and $y_0,y_1,\dots,y_{q+1} \in \cY$ satisfying
\begin{itemize}
    \item $y_i = H(x_i)$ for $0\leq i \leq q+1$
    \item $y_{i-1} = x_i \oplus x_{i-2}$ for $i \in [q+1]$, where $x_{-1} \coloneqq 0^n$
\end{itemize}
when $\cB$ interacts with the standard random oracle.

Let $p'_\cB$ be the probability that $\cB$ outputs $x_0,x_1,\dots,x_{q+1}\in\cX$ and $y_0,y_1,\dots,y_{q+1} \in \cY$ satisfying
\begin{itemize}
    \item $y_i = D(x_i)$ for $0\leq i \leq q+1$
    \item $y_{i-1} = x_i \oplus x_{i-2}$ for $i \in [q+1]$, where $x_{-1} \coloneqq 0^n$
\end{itemize} when $\cB$ interacts with the compressed oracle.

Let $\ol{p}'_\cB$ be the probability that $\cB$ outputs $x_0,x_1,\dots,x_{q+1}\in\cX$ and $y_0,y_1,\dots,y_{q+1} \in \cY$ satisfying
\begin{itemize}
    \item $D(x_{i-1}) = x_i \oplus x_{i-2}$ for $i \in [q+1]$, where $x_{-1} \coloneqq 0^n$
\end{itemize} when $\cB$ interacts with the compressed oracle.

We trivially have $p_\cA = p_\cB$ and $p'_\cB \leq \ol{p}'_\cB$. Since $\cB$ now outputs all the hash values as well, we can apply Lemma~\ref{lem:zhandry} to $\cB$ which gives 
\[
\sqrt{p_\cB} \leq \sqrt{p'_\cB} + \sqrt{\frac{q+2}{|\cY|}}.
\]

\noindent In the rest of the proof, it remains to bound $\ol{p}'_\cB$. 
\begin{align*}
    \sqrt{\ol{p}'_\cB} & \leq \sup_{U_1,\dots,U_q} \left\| \sum_\bfx \mathsf{TCHN}_\bfx^{q+1} (\ket{\bfx}\bra{\bfx} \otimes \mathsf{cO}_{x_{q+1}})U_q\mathsf{cO}^k U_{q-1} \mathsf{cO}^k \dots U_1\mathsf{cO}^k\bot \right\| \\
    & \leq \left\| \sum_\bfx \mathsf{TCHN}_\bfx^{q+1} (\ket{\bfx}\bra{\bfx}\otimes \mathsf{cO}_{x_{q+1}})\neg\mathsf{TCHN}^{q+1} \right\| \\
    & + \sup_{U_1,\dots,U_q} \|\mathsf{TCHN}^{q+1}U_q \mathsf{cO}^k U_{q-1} \mathsf{cO}^k \dots U_1\mathsf{cO}^k\bot\| \\
    & \leq \max_\bfx\|\mathsf{TCHN}_\bfx^{q+1}\mathsf{cO}_{x_{q+1}} \neg\mathsf{TCHN}^{q+1}\| 
    + \llbracket \bot \xRightarrow{q,k} \mathsf{TCHN}^{q+1}\rrbracket \\
    & \leq \max_\bfx\|\mathsf{TCHN}_\bfx^{q+1}\mathsf{cO}_{x_{q+1}} \neg\mathsf{TCHN}_\bfx^{q+1}\| 
    + \llbracket \bot \xRightarrow{q,k} \mathsf{TCHN}^{q+1}\rrbracket,
\end{align*}
where the summation is over all $\bfx = (x_0,\dots,x_{q+1}) \in \cX^{q+2}$;
$\set{\ket{\bfx}\bra{\bfx}}$ denotes the measurement acting on $\cB$'s output register to produce the output $\bfx$; 
the database property $\mathsf{TCHN}_\bfx^{q+1}$ is defined as
\[
\mathsf{TCHN}_\bfx^{q+1} \coloneqq 
\set{D \mid x_i = D(x_{i-1}) \xor x_{i-2} \text{ for } i\in[q+1]} \subseteq \mathfrak{D}.
\]
That is, the sequence $x_0,\dots,x_{q+1}$ forms a $(q+1)$-chain.

Now, notice that for every $\bfx \in \cX^{q+2}$,
\begin{align*}
& \|\mathsf{TCHN}_\bfx^{q+1}\mathsf{cO}_{x_{q+1}} \neg\mathsf{TCHN}_\bfx^{q+1}\| \\
= & \|\mathsf{TCHN}_\bfx^{q+1}(\mathsf{cO}_{x_{q+1}} \otimes \mathsf{cO}_{x_q\hat{0}} \otimes \dots \otimes \mathsf{cO}_{x_0\hat{0}}) \neg\mathsf{TCHN}_\bfx^{q+1}\| \\
\leq & \max_{\hat{\bfy}} \|\mathsf{TCHN}_\bfx^{q+1}\mathsf{cO}_{\bfx\hat{\bfy}} \neg\mathsf{TCHN}_\bfx^{q+1}\| 
\leq \llbracket \neg\mathsf{TCHN}_\bfx^{q+1} \xrightarrow{q+2} \mathsf{TCHN}_\bfx^{q+1}\rrbracket,
\end{align*}
where the first equality holds since $\mathsf{cO}_{x\hat{0}}$ is equal to the identity operator for every $x\in\cX$.

Following similar arguments as in Lemma~\ref{lem:local_properties}, we now show there exist local properties that recognize the database transition $\llbracket \neg\mathsf{TCHN}_\bfx^{q+1} \xrightarrow{q+2} \mathsf{TCHN}_\bfx^{q+1}\rrbracket$.
For any tuple $\bfx = (x_1,\dots,x_{q+2})$ with pairwise distinct entries, any tuple $\bfx' = (x'_0,\dots,x'_{q+1})$\footnote{In the rest of the proof, we switch the variable $\bfx$ of $\mathsf{TCHN}_\bfx^{q+1}$ into $\bfx'$ for convenience.} and database $D \in \mathfrak{D}$, we define the following local properties for $i \in [q+2]$
\[
L^{\bfx,D}_i \coloneqq \set{x'_0, \dots, x'_{q+1}} \cup \set{x \mid \exists a,b \in\set{1,\dots,q+2} \colon x = x'_a \oplus x'_b}.
\]
Note that $|L^{\bfx,D}_i| \leq (q+2) + \binom{q+2}{2} = (q+2)(q+3)/2$ for each $i\in[q+2]$.

Suppose $D[\bfx\mapsto\bfr] \notin \mathsf{TCHN}_{\bfx'}^{q+1}$ yet $D[\bfx\mapsto\bfu] \in \mathsf{TCHN}_{\bfx'}^{q+1}$. Then $\set{x'_0, \dots, x'_{q+1}}$ is a $(q+1)$-chain. Let $s_\circ$ be the smallest $j$ such that $D[\bfx\mapsto\bfr](x'_j) \neq D[\bfx\mapsto\bfu](x'_j)$. If $s_\circ = q+1$ or $j$ does not exist, then $D[\bfx\mapsto\bfr] \in \mathsf{TCHN}_{\bfx'}^{q+1}$ and we are done. So we assume $0\leq s_\circ \leq q$. Since $D[\bfx\mapsto\bfr]$ coincides $D[\bfx\mapsto\bfu]$ outside of $\bfx$, there must exists an index $i\in[q+2]$ such that $x_i = x'_{s_\circ}$. Therefore, we have $r_i = D[\bfx\mapsto\bfr](x_i) = D[\bfx\mapsto\bfr](x'_{s_\circ}) \neq D[\bfx\mapsto\bfu](x'_{s_\circ}) = D[\bfx\mapsto\bfu](x_i) = u_i$.

In addition, if $s_\circ = 0$, then $u_i = D[\bfx\mapsto\bfu](x'_0) = x'_1 \in \set{x_0',\dots,x'_{q+1}}$. If $1\leq s_\circ \leq q$, then $u_i = D[\bfx\mapsto\bfu](x'_{s_\circ}) = x'_{s_\circ-1} \xor x'_{s_\circ+1}$ which means $u_i$ is the XOR of two distinct elements in $\set{x'_0, \dots, x'_{q+1}}$. In either case, $u_i$ must lie in $L^{\bfx,D}_i$.
Therefore, by Theorem~5.23 in~\cite{chung2021compressed}, for every $\bfx'\in\cX^{q+2}$ we have
\[
\llbracket \neg\mathsf{TCHN}_{\bfx'}^{q+1} \xrightarrow{q+2} \mathsf{TCHN}_{\bfx'}^{q+1}\rrbracket
\leq e(q+2)\sqrt{\frac{5(q+2)(q+3)}{|\cY|}}.
\]
Thus, we can bound $\max_{\bfx'}\|\mathsf{TCHN}_{\bfx'}^{q+1}\mathsf{cO}_{x'_{q+1}} \neg\mathsf{TCHN}_{\bfx'}^{q+1}\|$ by the above quantity.

Putting things together, we have
\[
p_\cA \leq 
\left(  
\llbracket \bot \xRightarrow{q,k} \mathsf{TCHN}^{q+1}\rrbracket
+ e(q+2)\sqrt{\frac{5(q+2)(q+3)}{|\cY|}} 
+ \sqrt{\frac{q+2}{|\cY|}}
\right)^2
\]
Bounding the first term by Lemma~\ref{lem:capacity}, this concludes the proof.
\end{proof}

Now, we are ready to prove the main theorem.
\begin{proof}[Proof of Theorem~\ref{thm:twisted_hash_chain}]
We finish the proof by reduction. 
Define the algorithm $\cD$ as follows:
\begin{enumerate}
    \item Run $\cC$ and obtain $x_0,x_q$ and $x_{q+1}$.
    \item For $i \in [q-1]$: \\
    - Make a classical $2$-parallel query $(x_{i-1},x_{q+i})$ to the random oracle and then obtain $(H(x_{i-1}),$ $H(x_{q+i}))$. \\
    - Set $x_i \coloneqq H(x_{i-1}) \oplus x_{i-2}$ and $x_{q+i+1} \coloneqq H(x_{q+i}) \oplus x_{q+i-1}$.
    \item Make a classical $2$-parallel query $(x_{q-1},x_{2q})$ to the random oracle and then obtain $(H(x_{q-1}),$ $H(x_{2q}))$.
    \item Output $x_0,x_1,\dots,x_{2q+1}$ and $H(x_0),H(x_1),\dots,H(x_{q-1}),H'(x_q),H(x_{q+1}),\dots,H(x_{2q})$, \\ 
    where $H'(x_q) \coloneqq x_{q-1} \oplus x_{q+1}$\footnote{Note that in Step 2 and 3, $\cD$ makes a total of $2q$ queries including $x_0,\dots,x_{2q}$ except $x_q$.}.
\end{enumerate}
First, it is trivial that $p_\cC$ is equivalent to the probability $p_\cD$ that $H'(x_q) = H(x_q)$ and $\cD$ outputs a $(2q+1)$-chain.
Now, we calculate the total number of queries made by $\cD$. 
In Step 1, $\cD$ makes $q$ $k$-parallel queries to execute $\cC$. 
In Steps 2 and 3, $\cD$ makes $q$ $2$-parallel queries. 
To sum up, $\cD$ makes a total of $2q$ $k$-parallel queries. 
By Lemma~\ref{lem:chain_all_but_last}, the probability $p_\cD$ is at most $F(k,2q)$. 
Therefore, this finishes the proof.
\end{proof}

Considering the situation in which the algorithm is assigned to a particular starting point $x_0\in\cX$ of the chain, we have the following corollary which is trivially implied by Theorem~\ref{thm:twisted_hash_chain}.
\begin{corollary}\label{cor:twisted_hash_chain_puzzle}
For any $k$-parallel $q$-query oracle algorithm $\cE$, the probability $p_\cE$ (parameterized by $k$ and $q$) that the algorithm takes a uniformly random $x_0\in\cX$ as input, and outputs $x_q,x_{q+1}\in\cX$ satisfying
\begin{itemize}
    \item there exist $x_1,\dots,x_{q-1}\in\cX$ such that $H(x_{i-1}) = x_i \oplus x_{i-2}$ for $i\in[q+1]$, where $x_{-1} \coloneqq 0^n$
\end{itemize}
is at most $F(k,2q) = O(k^4q^4/|\cY|)$, where the function $F$ is defined in Lemma~\ref{lem:chain_all_but_last}.
\end{corollary}
\begin{proof}
    We finish the proof by reduction. Let $\cC$ be the algorithm that first samples $x_0\in\cX$ uniformly at random and invokes $\cE(x)$. $\cC$ responds to every $\cE$'s oracle query by its oracle access directly. Then $\cC$ outputs whatever $\cE$ outputs. 
    
    Let $p_\cC$ be the probability defined as in Theorem~\ref{thm:twisted_hash_chain}.
    Since the success of $\cE$ implies the success of $\cC$, we have $p_\cC \geq p_\cE$.
    By Theorem~\ref{thm:twisted_hash_chain} and the construction of $\cC$, we have $F(k,2q) \geq p_\cC$, which concludes the proof.
\end{proof}

\begin{remark}
Here, we explain the challenging issue of our case. Given only $x_0,x_q$ and $x_{q+1}$, in order to output the whole chain, the algorithm cannot make the query in parallel but is required to make \emph{adaptive} queries. For example, to reveal the next point $x_1 = H(x_0)$, the algorithm must first query $x_0$. Therefore, we cannot use Theorem~5.9 in~\cite{chung2021compressed} in a black-box way.
\end{remark}
\section{Quantum Walk on a Line}\label{sec:Q_walk}
Our proof of Hamiltonian simulation lower bound relies on the continuous-time quantum walk on a line \cite{Childs_2003}. We introduce quantum walks on a line in this section.

Consider a particle moving on a graph, which is a line with $L$ vertices.
Each vertex on the line is labeled by an integer $1,2,\dots L$.
We use a quantum state $\ket{j}$ to denote the particle locating at the vertex $j$.
Figure~\ref{fig:1d_walk}(a) illustrates our system, a finite segment with length $L$.
We let the Hamiltonian $\HL$ of the system be the adjacency matrix of the graph.
In physics terminology, $\HL$ couples adjacent vertices with the coupling constant $1$.
We have
\begin{equation}
  \label{eq:h_simple}
  \HL = \sum_{j = 1}^{L-1} \ket{j}\bra{j+1} + \ket{j+1}\bra{j},
\end{equation}
or
\begin{equation*}
  \HL = \left(
  \begin{array}{cccccc}
    0 & 1 & 0 & & & 0\\
    1 & 0 & 1 & \cdots & \cdots & 0\\
    0 & 1 & 0 & & & 0\\
      & \vdots & & \ddots & &\vdots\\
      & \vdots & & & 0 & 1\\
    0 & 0 & 0 & \cdots & 1 & 0
  \end{array}
  \right).
\end{equation*}
in the $\{\ket{j}\}_{j=1}^{L}$ basis.

The dynamics of the particle are determined by the time evolution operator $e^{-i\HL\evolt}$, where $\evolt$ is the evolution time.
We call the dynamics of the system ``quantum walk on a line.''

\begin{figure}[ht]
    \centering
    \begin{adjustbox}{max width = \textwidth}
\begin{tikzpicture}
    \node at (-3.5, 1) {(a)};
    \filldraw [black] (1,1) circle (2.5pt) node[below]{\small{$\ket{1}$}};
    \filldraw [black] (2.5,1) circle (2pt) node[below]{$\ket{2}$};
    \filldraw [black] (4,1) circle (2pt) node[below]{$\ket{3}$};
    
    \filldraw [black] (4.375,1) circle (0.5pt);
    \filldraw [black] (4.75,1) circle (0.5pt);
    \filldraw [black] (5.125,1) circle (0.5pt);
    
    \filldraw [black] (5.5,1) circle (2pt) node[below]{\small{$\ket{L-2}$}};
    \filldraw [black] (7,1) circle (2pt) node[below]{$\ket{L-1}$};
    \filldraw [black] (8.5,1) circle (2.5pt) node[below]{$\ket{L}$};
    
    \draw (1,1) - - (4,1);
    \draw (5.5,1) - - (8.5,1);
    

    \node at (-3.5, -1) {(b)};
    \filldraw [black] (-3.125,-1) circle (0.5pt);
    \filldraw [black] (-2.75,-1) circle (0.5pt);
    \filldraw [black] (-2.375,-1) circle (0.5pt);
    
    \draw (-2,-1) - - (4,-1);
    
    \filldraw [black] (-2,-1) circle (2pt) node[below]{\small{$\ket{-1}$}};
    \filldraw [black] (-0.5,-1) circle (2pt) node[below]{\small{$\ket{0}$}};
    \filldraw [black] (1,-1) circle (2pt) node[below]{\small{$\ket{1}$}};
    \filldraw [black] (2.5,-1) circle (2pt) node[below]{$\ket{2}$};
    \filldraw [black] (4,-1) circle (2pt) node[below]{$\ket{3}$};
    
    \filldraw [black] (4.375,-1) circle (0.5pt);
    \filldraw [black] (4.75,-1) circle (0.5pt);
    \filldraw [black] (5.125,-1) circle (0.5pt);
    
    \filldraw [black] (5.5,-1) circle (2pt) node[below]{\small{$\ket{L-2}$}};
    \filldraw [black] (7,-1) circle (2pt) node[below]{$\ket{L-1}$};
    \filldraw [black] (8.5,-1) circle (2pt) node[below]{$\ket{L}$};
    \filldraw [black] (10,-1) circle (2pt) node[below]{$\ket{L+1}$};
    \filldraw [black] (11.5,-1) circle (2pt) node[below]{$\ket{L+2}$};
    
    \draw (5.5,-1) - - (11.5,-1);
    
    \filldraw [black] (11.875,-1) circle (0.5pt);
    \filldraw [black] (12.25,-1) circle (0.5pt);
    \filldraw [black] (12.625,-1) circle (0.5pt);
    
\end{tikzpicture}
\end{adjustbox}
    \caption{Quantum walk on a line. (a) Quantum walk on a finite segment with length $L$. (b) Quantum walks on an infinite line.}
    \label{fig:1d_walk}
\end{figure}
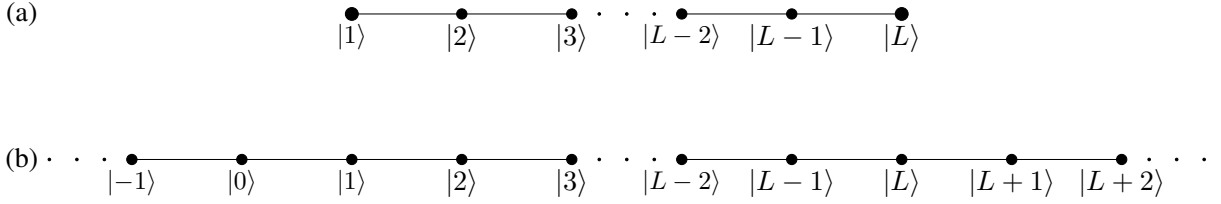

We are interested in the dynamics of a particle initially at the end of the line.
In other words, we consider the evolution of a particle under $\HL$ with the initial state $\ket{1}$.
We have the following result.

\begin{lemma}
\label{lem:simple_h}
  Given a system that evolves under the Hamiltonian $\HL$ described in \eqref{eq:h_simple} with initial state $\ket{1}$, if the system is measured at time $\evolt \in [0, L/2]$ in the $\{\ket{j}\}_{j=0}^{L}$ basis with outcome $l$, the probability that $l>\evolt$ is at least $1/3$.
\end{lemma}

Before the formal proof of Lemma~\ref{lem:simple_h}, we first discuss the general behavior of the quantum walk.
Let the particle initially locate at $\ket{k}$ and evolve under the Hamiltonian $\HL$.
When measuring the system at time $\evolt$ in the $\{\ket{j} \}_{j=1}^{L}$ basis, the probability $P(k,l,\evolt)$ of measurement outcomes being $l$ is
\begin{equation}
  \label{eq:prob_get_l}
  P(k,l,\evolt) = \abs{\bra{l}e^{-i\HL\evolt}\ket{k}}^2.
\end{equation}
By diagonalizing $\HL$, we can calculate $P(k, l, \evolt)$ as follows:
\begin{equation}
  \label{eq:prob_sol}
  P(k, l, \evolt) = \abs{\bra{l}e^{-i\HL\evolt}\ket{k}}^2
  = \sum_{p,q=1}^{L} e^{-i(\lambda_p-\lambda_q) \evolt}
  v^{(p)}_l v^{(p)\ast}_k  v^{(q)\ast}_l v^{(q)}_k,
\end{equation}
where $\lambda_p$'s are the eigenvalues of $\HL$ --- each with the corresponding eigenstate $\ket{v^{(p)}} = \sum_{j=1}^{L} v^{(p)}_j \ket{j}$.
The eigenvalues and the eigenstates of $\HL$ have a closed-form expression.
That is, $\lambda_p = 2\cos(\frac{p\pi}{L+1})$ and  $v^{(p)}_j = \sqrt{\frac{2}{L+1}}\sin(\frac{jp\pi}{L+1})$ \cite{nagaj2010fast}
\footnote{In fact, there is a simpler form of $P(k,l,\evolt)$: when $\abs{l-k}$ is even, $$P(k,l,\evolt) = \big(\sum_{p} \cos\left(2\evolt\cos\left(\frac{p\pi}{L+1}\right)\right)\sin\left(\frac{kp\pi}{L+1}\right)\sin\left(\frac{lp\pi}{L+1}\right) \big)^2,$$ and when $\abs{l-k}$ is odd, $$P(k,l,\evolt) = \big(\sum_{p} \sin\left(2\evolt\cos\left(\frac{p\pi}{L+1}\right)\right)\sin\left(\frac{kp\pi}{L+1}\right)\sin\left(\frac{lp\pi}{L+1}\right)\big)^2.$$}.

We use the propagation of the wave function of a free particle to analogize the quantum walk.\footnote{Consider an extreme case that the distance between two adjacent vertices in space goes to zero and the length $L$ goes to infinity. The system is reduced to free space.}
For example, we plot the result of the quantum walk on a segment of length $L=100$ in Figure~\ref{fig:id_walk_result}.
The initial state is $\ket{k=1}$ and we focus on the time interval $\evolt\in[0, L/2]$.
Figure~\ref{fig:id_walk_result}(a) shows $P(1, l, \evolt)$, the probability of obtaining the measurement outcome $\ket{l}$, for every $l$ at different time $\evolt$.
We see that at time $\evolt$, the wavefront reaches $l \approx 2\evolt$.
Figure~\ref{fig:id_walk_result}(b) shows the probability of getting the measurement outcome $\ket{l}$ for a \emph{fixed} $l$ versus time.
We see that the probability is extremely small when $\evolt\ll l/2$ and reaches the maximum at $\evolt\approx l/2$.
Finally, it behaves like a damped oscillation when $\evolt \gtrsim l/2$.
These observations suggest that the wavefront propagates at a constant speed, which gives a hint that the particle reaches the vertex $l = \Theta(\evolt)$ at time $\evolt$.\footnote{This corresponds to the fact that the uncertainty of the position of a free particle is linear in $\evolt$. See, for example, \cite{SN20}.}

\begin{figure}[ht]
    \centering
    \begin{subfigure}{0.48\textwidth}
    \includegraphics[width=\textwidth]{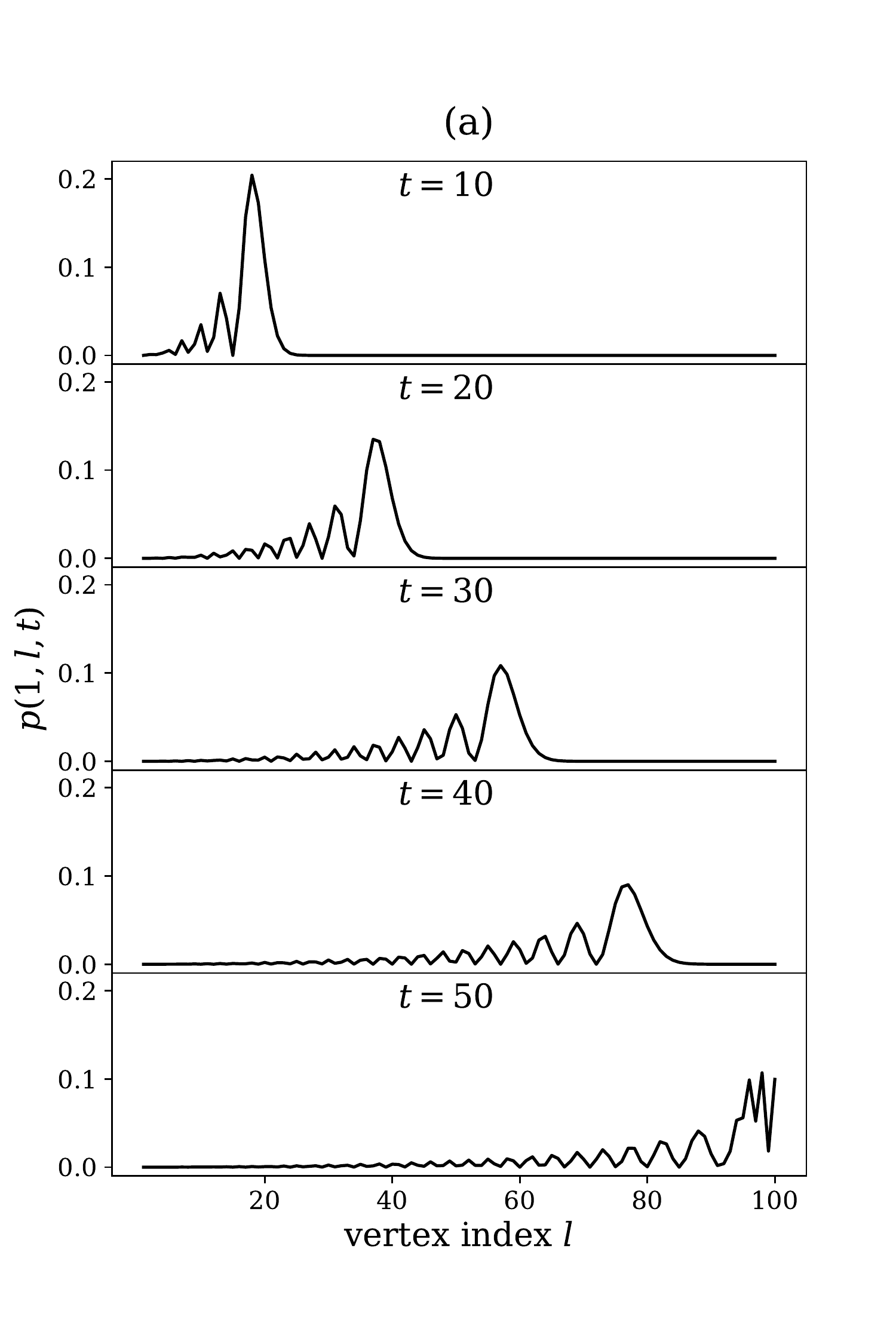}
    \end{subfigure}
    \begin{subfigure}{0.48\textwidth}
    \includegraphics[width=\textwidth]{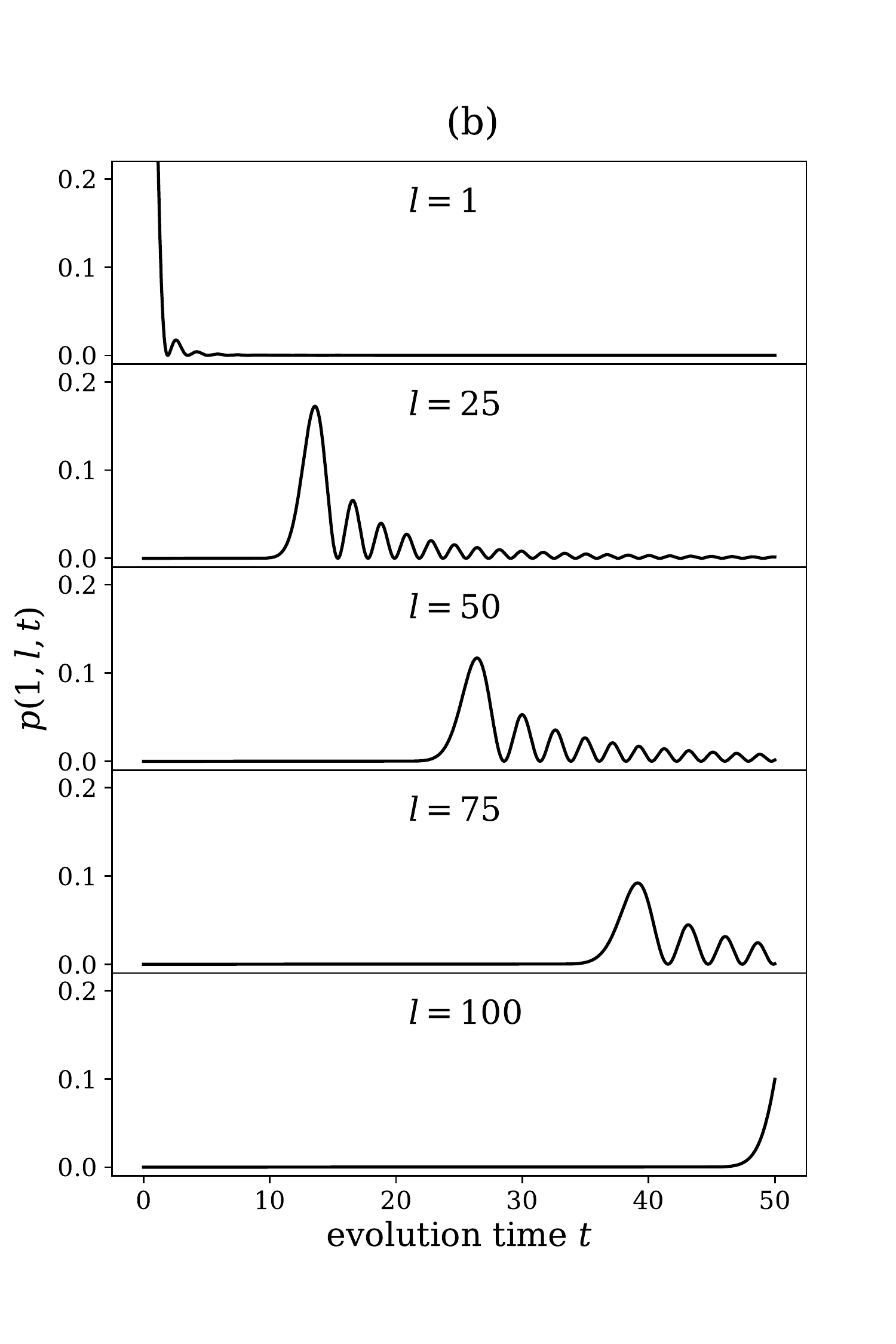}
    \end{subfigure}
    \caption{The result of quantum walk on a segment with $L=100$ for the evolution time $\evolt\in[0, L/2]$. The initial state is $\ket{1}$. (a) The probability of getting the outcomes $\ket{l}$ for every node $l$ at the evolution time $\evolt = 10, 20, 30, 40$ and $50$ respectively. (b) The probability of getting the outcomes $l=1, 25, 50, 75$ and $100$ versus evolution time $\evolt\in[0, L/2]$.}
    \label{fig:id_walk_result}
\end{figure}

Next, we are going to prove Lemma~\ref{lem:simple_h}.
We take another approach instead of diagonalizing $\HL$ directly.
We follow the approach in \cite{Childs_2003}.
Similar to solving ``the particle in a box model'' in quantum mechanics, we first find the homogeneous solution in free space and then find the particular solution that satisfies the boundary conditions and the initial conditions. (See, for example, \cite{SN20}.)

Consider the quantum walk on an infinite line which is illustrated in Figure~\ref{fig:1d_walk}(b).
The Hamiltonian of the quantum walk on an infinite line is defined by
\begin{equation}
  \label{eq:H_infty}
  H_{\infty} \coloneqq \sum_{j=-\infty}^{\infty} \ket{j+1}\bra{j} +\ket{j}\bra{j+1},
\end{equation}
and we define the propagator
\begin{equation}
  G(k,l,\evolt) := \bra{l}e^{-iH_{\infty}\evolt}\ket{k}.
\end{equation}
The (sub-normalized) eigenstate of $H_{\infty}$ is the momentum state $\ket{p}$.
The momentum state has the following property
\begin{equation}
  \inner{j}{p} = e^{ipj}, -\pi\leq p \leq \pi.
\end{equation}
The corresponding eigenvalue of $\ket{p}$ is $E_p = 2\cos p$.
Hence, we have
\begin{equation}
  \bra{l}e^{-iH_{\infty}\evolt}\ket{k}
  = \int_{-\pi}^{\pi} dp e^{-i2\evolt\cos p + i p(l-k)}
  = i^{(l-k)}J_{l-k}(2\evolt),
\end{equation}
where $J_n(\cdot)$ is the Bessel function of order $n$. (See \eqref{eq:jn_integral}.)

Now we are ready to calculate the propagator of the quantum walk on a \emph{finite} segment.
We use $\wt{G}(k, l, \evolt) \coloneqq \bra{l}e^{-i\HL\evolt}\ket{k}$ to denote the propagator of the quantum walk on a finite segment.
The propagator $\wt{G}$ is a superposition of $G$ and $\wt{G}$ that satisfies the boundary conditions: $\wt{G}(k, 0, \evolt) = 0 = \wt{G}(k, L+1, \evolt)$, and the initial condition $\wt{G}(k, l, 0) =\delta_{kl}$.

The solution is
\begin{equation}
  \label{eq:G_tilde_finite}
  \wt{G}(1, l, \evolt) = \sum_{m=-\infty}^{\infty}
  G(1, l+2m(L+1), \evolt) - G(1, -l+2m(L+1), \evolt).
\end{equation}
The above equation \eqref{eq:G_tilde_finite} can be interpreted as the wave reflecting between the boundaries $j=0$ and $j=L+1$. 

We set the starting point $j=1$.
In the time interval that we are interested in, namely, $\evolt \in[0, L/2]$, we have $G(1, \pm l + 2m(L+1), t) = J_{2m(L+1)\pm l -1 }(2t)$ is exponentially small in $L$ for $m\neq 0$.
This is because the order $\abs{2m(L+1)\pm l -1} >L$ for $m\neq 0$ and the argument $2\evolt\leq L$ for $\evolt \leq L/2$. (See \eqref{eq:jn_large_order_xltn}.)

\noindent Thus,
\begin{align}
\begin{split}
  \wt{G}(1, l,\evolt) &\approx G(1,l,\evolt) - G(1,-l, \evolt)\\
  &= i^{l-1} J_{l-1}(2\evolt) - i^{-(l+1)} J_{-(l+1)}(2\evolt)\\
  &=  i^{l-1} J_{l-1}(2\evolt) - (i^{-(l+1)})(-1^{l+1}) J_{l+1}(2\evolt)\\
  &=  i^{l-1} J_{l-1}(2\evolt) - (-i)^{l+1})  J_{l+1}(2\evolt)\\
  &=  i^{l-1} (J_{l-1}(2\evolt) - (-i)^{2} J_{l+1}(2\evolt))\\
  &=  i^{l-1} (J_{l-1}(2\evolt) + J_{l+1}(2\evolt))\\
  &=  i^{l-1} \frac{l}{\evolt} J_{l}(2\evolt).\label{eq:G_tilde_sol}
  \end{split}
\end{align}
The third equation is due to the relation of negative order \eqref{eq:jn_minus_n} of the Bessel function,
and the last equation uses the recursion property \eqref{eq:jn_recursion} of the Bessel function.
Then we have

\begin{equation}
  \label{eq:p_using_bessel}
  P(1, l, \evolt)
  = \abs{\wt{G}(1, l, \evolt)}^2
  \approx \left(\frac{l}{\evolt}\right)^2 J^2_{l}(2\evolt).
\end{equation}
As a remark, the probability $P(1, l, \evolt)$ is almost independent of $L$ when $\evolt\in [L/2]$.
It can be interpreted as the following: before the wavefront reaches the boundary, the wave propagates as in free space.
Finally, we prove Lemma~\ref{lem:simple_h}.

\begin{proof}[Proof of Lemma~\ref{lem:simple_h}]
We directly calculate the probability $\sum_{l=1}^{\floor{\evolt}} P(1, l, \evolt)$ as follows:
\begin{align*}
\sum_{l=1}^{\floor{\evolt}} P(1,l, \evolt)
= \sum_{l=1}^{\floor{\evolt}} \left(\frac{l}{\evolt}\right)^2 J^2_{l}(2\evolt)
\leq \sum_{l=1}^{\floor{\evolt}} \left(\frac{l}{\evolt}\right)^2 \frac{2}{\pi}\frac{1}{l}
= \frac{2}{\pi} \sum_{j=l}^{\floor{\evolt}} \frac{l}{\evolt^2}
\leq \frac{2}{\pi},
\end{align*}
where the first equation follows from \eqref{eq:p_using_bessel} and the second inequality follows from Lemma~\ref{lem:bessel_tail}. \\
\noindent As a result, we conclude that
\begin{equation*}
\sum_{l=\ceil{\evolt}}^{L} P(1, l, \evolt) 
= 1- \sum_{l=1}^{\floor{\evolt}} P(1, l, \evolt)
>\frac{1}{3}.
\end{equation*}
\end{proof}
\section{No Fast-forwarding in Oracle Model: Unconditional Result}
\label{sec:oracle}

In this section, we are going to investigate the parallel lower bound of Hamiltonian simulation in the oracle model.
In the oracle model, the Hamiltonian is expressed by a Hermitian matrix.
There are many algorithms that can efficiently simulate a Hamiltonian in the oracle model if the Hamiltonian matrix is sparse \cite{berry2007efficient, CB12, BCC+14, BCK15, LC17, LC19}.
As a result, we are interested in the lower bound of simulating a sparse Hamiltonian.
Besides, we normalize the Hamiltonian by setting the absolute value of every element of the Hamiltonian to be at most $1$.
The sparse Hamiltonian is defined as follows.

\begin{definition}[Sparse Hamiltonian]
  Let $H\in\mathbb{C}^{N\times N}$ denote a Hamiltonian acting on the Hilbert space with dimension $N$.
  We say $H$ is $d$-sparse if there are at most $d$ nonzero entries in every row.
\end{definition}

In the oracle setting, the simulation algorithm can only obtain the description of the Hamiltonian via oracle queries.
In most of the models of the algorithms, there are two oracles that can be accessed:
First, the \emph{entry oracle}, denoted by $\OracleH$, answers the value of the matrix element.
Second, the \emph{sparse structure oracle}, denoted by $\OracleSparse$, answers the index of the nonzero entry.
Let the Hamiltonian $H$ that we want to simulate be acting on an $N$-dimensional Hilbert space and be $d$-sparse.
When the entry oracle $\OracleH$ is queried on the index $(j, k)$ where $j, k\in[N]$, it returns the element value $H_{jk}$.
When the sparse structure oracle is queried on $(j, s)$ where $j\in[N]$ and $s\in[d]$, it returns $k$ where $H_{jk}$ is the $s$-th nonzero entry of the $j$-th row.

The algorithm can query these two oracles in superposition respectively.
In the standard quantum oracle model, these two oracles are written as:

\begin{equation}
  \label{eq:oracle_h}
  \OracleH \ket{j,k,z} = \ket{j, k, z\oplus H_{jk}},
\end{equation}
and
\begin{equation}
  \label{eq:oracle_sparse}
  \OracleSparse \ket{j, s} = \ket{j, k},
\end{equation}
where $k$ is the index of the $s$-th nonzero entry in the $j$-th row.

We are going to prove that simulating a quantum system for evolution time $t$ requires at least $\Omega(t)$ parallel quantum queries.
We have the following result.

\begin{theorem}[Simulation lower bound in the oracle model]
  \label{thm:lower_bound_oracle}
  For any integer $n$, any polynomial $T(\cdot)$ and $\pquery = \poly(n)$, there exists a time-independent Hamiltonian $H\in\mathbb{C}^{(2^nT(n))\times(2^nT(n))}$ satisfies the following.
  For any quantum algorithm that can make $\pquery$-parallel queries to the entry oracle $\OracleH$ (defined in \eqref{eq:oracle_h}) and the sparse structure oracle $\OracleSparse$ (defined in \eqref{eq:oracle_sparse}), simulating $H$ for an evolution time $t\in [0, T(n)/2]$ within an error $\epsilon \leq 1/4$ needs at least $\Omega(t)$ $\pquery$-parallel queries to $\OracleH$ and $\OracleSparse$ in total.
  Furthermore, $H$ is $2$-sparse and $\abs{H_{jk}}\le 1$ for every $j,k\in [2^nT(n)]$.
\end{theorem}

Theorem~\ref{thm:lower_bound_oracle} can be interpreted as simulating a system with $n+O(\log n)$ qubits for an evolution time $t<\poly (n)$ cannot be fast-forwarded.

Before the formal proof of Theorem~\ref{thm:lower_bound_oracle}, we first sketch our proof strategy.
We modify the proof of the query lower bound in \cite{berry2007efficient}.
In \cite{berry2007efficient}, the parity problem is reduced to the Hamiltonian simulation problem.
In particular, it is shown that if one can fast-forward the Hamiltonian simulation, then one can find the parity of an $N$-bit string with $o(N)$ queries.
However, this technique cannot be extended to prove the parallel lower bound since finding the parity of a string is not parallel-hard.
Instead, we reduce the permutation chain problem, of which the parallel hardness was already proven in Section~\ref{sec:hashchain}, to the Hamiltonian simulation.
We are going to show that there exists a specific Hamiltonian such that simulating the Hamiltonian implies solving the permutation chain problem.

We restate the permutation chain problem and its hardness below.
\begin{definition}[Permutation chain]
\label{def:hash_chain}
  Let $n\in\mathbb{N}$ and $\pquery, L = \poly(n)$. 
  For each $j\in [L]$, let $\hash_{j}:\ \{ 0, 1\} ^{n}\rightarrow\{ 0, 1\} ^{n}$ be a random permutation and let $\hash_{j}^{-1}$ be the inverse of $\hash_{j}$.
  Let $\hashstep^{(j)}(\cdot)$ denote $\hash_j(\hash_{j-1}(\cdots \hash_1(\cdot)))$.
  A quantum algorithm can make $\pquery$-parallel query to both $\hash_{j}$ and $\hash_{j}^{-1}$ for each $j\in [L]$ respectively and is asked to output $ x_{\querynum} \in \{ 0, 1\} ^{n}$ such that $x_{\querynum} = \hashstep^{(\querynum)} (0^n)$, where $q\in [L]$.
\end{definition}

\begin{corollary}[Hardness of permutation chain]
\label{thm:hardness_hash_chain}
    Let $n\in\mathbb{N}$, and $\pquery, L = \poly(n)$.
    For each $j\in [L]$, let $\Pi_j$ and $\Pi^{-1}_j$ be a random permutation over $n$-bit strings and its inverse.
    Let $\hashstep^{(j)}(\cdot) := \hash_j(\hash_{j-1}(\cdots \hash_1(\cdot)))$ be the function defined in Definition~\ref{def:hash_chain}.
    For any $t,q\in[L]$ and any quantum algorithm $\cA$ that makes $t$ $\pquery$-parallel queries to $\hash_j$ and $\hash^{-1}_{j}$, the probability that $\cA$ outputs $x_q\in\{0,1\}^{n}$ satisfying $x_q = f^{(q)}(0^n)$ and $t<q$ is negligible in $n$. 
\end{corollary}
\begin{proof}
    Let $\bar{x}_q \coloneqq f^{(q)}(0^n)$ for each $q\in [L]$.
    The probability that $\cA$ outputs $\bar{x}_q$ such that $t<q$ is given by $$\sum_{j=t+1}^{L}\Pr[\cA \textrm{ outputs }\bar{x}_j].$$
    By Theorem~\ref{thm:hashchain_main}, for any quantum algorithm $A$ that makes $t$ $\pquery$-parallel queries to $\hash_j$ and $\hash^{-1}_{j}$, the probability that $\cA$ outputs $x_j$ such that $j>t$ is $O(t\sqrt{\pquery/2^n})$. Hence, the probability $$\sum_{j=t+1}^{L}\Pr[\cA \textrm{ outputs }\bar{x}_j] = \poly(n)\cdot O\left(t\sqrt{\frac{\pquery}{2^n}}\right)$$ is negligible in $n$.
\end{proof}

Similar to \cite{berry2007efficient}, we use quantum walk on a graph to solve the underlying hard problem.
We construct a graph that consists of $L$ columns where there are $2^n$ vertices in each column.
Each vertex in the $j$-th column is labelled by $(j, x)$, where $j\in \{0,1,\dots L\}$ and $x\in \{ 0, 1\} ^{n}$.
The label is translated as follows: after $j$ queries, the output string is $x$.
The vertices in the $j$-th column are only adjacent to the vertices that are in the $(j\pm 1)$-th columns.
Furthermore, the vertices $(j, x)$ and $(j+1, x')$ (\resp $(j-1,x')$) are adjacent if and only if $x'=\hash_{j+1}(x)\ (\mathrm{resp.}\ \hash_{j}^{-1}(x))$.
Because each $\hash_{j}$ is a permutation, the graph consists of $2^n$ disconnected lines of length $L$.
If the vertices $(j, x)$ and $(0, x_0)$ are connected, it holds that $x = \hashstep^{(j)}(x_0)$.

In Figure~\ref{fig:hash_graph}, we presents a toy example: let $\hash_{j} :\{ 0, 1 \} ^{2} \rightarrow \{ 0, 1 \} ^{2}$ for each $j\in [L]$ and each $\hash_{j} = \hash$ has the same truth table, \ie
\begin{equation*}
  \hash(00) = (01),\ \hash(01)=(10),\ \hash(10) = (11),\  \mathrm{and}\ \hash(11) = (00).
\end{equation*}

\begin{figure}[ht]
  \begin{adjustbox}{max width = \textwidth}
\begin{tikzpicture}
  \filldraw [black] (0,3) circle (2pt) node[above]{$(0, 00)$};
  \filldraw [black] (0,2) circle (2pt) node[above]{$(0, 01)$};
  \filldraw [black] (0,1) circle (2pt) node[above]{$(0, 10)$};
  \filldraw [black] (0,0) circle (2pt) node[above]{$(0, 11)$};

  \filldraw [black] (3,3) circle (2pt) node[above]{$(1, 00)$};
  \filldraw [black] (3,2) circle (2pt) node[above]{$(1, 01)$};
  \filldraw [black] (3,1) circle (2pt) node[above]{$(1, 10)$};
  \filldraw [black] (3,0) circle (2pt) node[above]{$(1, 11)$};

  \draw (0,3) -- (3,2);
  \draw (0,2) -- (3,1);
  \draw (0,1) -- (3,0);
  \draw (0,0) -- (3,3);
  \filldraw [black] (6,3) circle (2pt) node[above]{$(2, 00)$};
  \filldraw [black] (6,2) circle (2pt) node[above]{$(2, 01)$};
  \filldraw [black] (6,1) circle (2pt) node[above]{$(2, 10)$};
  \filldraw [black] (6,0) circle (2pt) node[above]{$(2, 11)$};

  \draw (3,3) -- (6,2);
  \draw (3,2) -- (6,1);
  \draw (3,1) -- (6,0);
  \draw (3,0) -- (6,3);
  \draw (6,3) -- (6.3,2.9);
  \draw (6,2) -- (6.3,1.9);
  \draw (6,1) -- (6.3,0.9);
  \draw (6,0) -- (6.3,0.1);
  \filldraw [black] (8.2, 1.5) circle (1pt);
  \filldraw [black] (8.5, 1.5) circle (1pt);
  \filldraw [black] (8.8, 1.5) circle (1pt);
  \draw (10.7, 2.1) -- (11,2);
  \draw (10.7, 1.1) -- (11,1);
  \draw (10.7, 0.1) -- (11,0);
  \draw (10.7, 2.9) -- (11,3);
  \filldraw [black] (11,3) circle (2pt) node[above]{$(L-1, 00)$};
  \filldraw [black] (11,2) circle (2pt) node[above]{$(L-1, 01)$};
  \filldraw [black] (11,1) circle (2pt) node[above]{$(L-1, 10)$};
  \filldraw [black] (11,0) circle (2pt) node[above]{$(L-1, 11)$};
  \filldraw [black] (14,3) circle (2pt) node[above]{$(L, 00)$};
  \filldraw [black] (14,2) circle (2pt) node[above]{$(L, 01)$};
  \filldraw [black] (14,1) circle (2pt) node[above]{$(L, 10)$};
  \filldraw [black] (14,0) circle (2pt) node[above]{$(L, 11)$};

   \draw (11,3) -- (14,2);
   \draw (11,2) -- (14,1);
   \draw (11,1) -- (14,0);
   \draw (11,0) -- (14,3);

  \draw [decorate, decoration = {brace, amplitude=5pt},  very thick] (14,-0.2) --  (0,-0.2);
  \node at (7, -0.7) {$L+1$};
\end{tikzpicture}
\end{adjustbox}
  \caption{Using quantum walk to solve the permutation chain problem: a toy example.}
  \label{fig:hash_graph}
\end{figure}
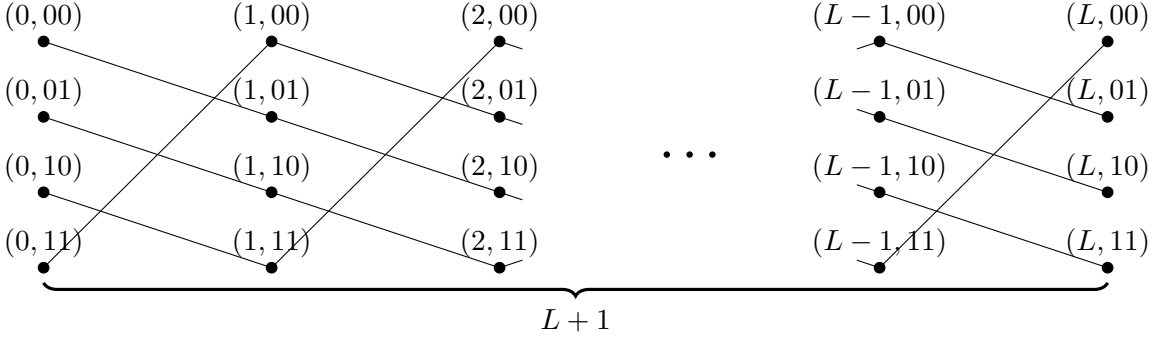

We let the Hamiltonian $H$ that determines the behavior of the quantum walk be the adjacency matrix of the graph.\footnote{Our Hamiltonian is different from that appears in \cite{berry2007efficient}, in which the graph is weighted.}
That is,
\begin{align}
    \begin{split}
    \label{eq:h_oracle_reduction}
    H &= \sum_{j=0}^{L-1}\sum_{x\in\{0,1\}^{n}} \ket{j+1, \hash_{j}(x)}\bra{j,x} + \ket{j, x}\bra{j+1,\hash_{j}(x)}\\
    &= \sum_{x_{0}\in\{0,1\}^{n}} \sum_{j=0}^{L-1} \ket{j+1, \hashstep^{(j+1)}(x_0)}\bra{j,\hashstep^{(j)}(x_0)} + \ket{j, \hashstep^{(j)}(x_0)}\bra{j+1,\hashstep^{(j+1)}(x_0)}.
    \end{split}
\end{align}
Because two vertices on different lines are decoupled, we have the following observation.
\begin{observation}
\label{obs:h_subspace}
If the random walk starts at the vertex $(0, x_0)$, then it always walks on the same line.
To be more precise, if a system evolves under the Hamiltonian $H$ described in~\eqref{eq:h_oracle_reduction} and the initial state is $\ket{0, x_0}$, then at any time $t$, the quantum state of the system is in the subspace $\spn\left(\left\{\ket{j, \hashstep^{(j)}(x_0)} \right\}_{j=0}^{L}\right)$.
\end{observation}
Observation~\ref{obs:h_subspace} can be verified by taking the Taylor expansion of the time evolution operator: $e^{-iHt} = \sum_{k = 0}^\infty (-iHt)^k/k!$.

To solve the permutation chain problem, we use a Hamiltonian simulation algorithm to simulate the quantum walk under the Hamiltonian $H$ with initial state $\ket{0, 0^n}$.
When we measure the system at time $t$ and get the outcome $(q, x)$.
The string $x$ is a potential solution to the permutation chain problem.
We aim to prove the following two statements.
First, the oracles $\OracleH$ and $\OracleSparse$ can be simulated efficiently by $\hash_{j}$ and $\hash_{j}^{-1}$.
Second, the probability of getting a measurement outcome $(q, x)$ at time $t$ such that $q\ge t$ is high.
Combining these two statements, we have the following conclusion. If an algorithm can simulate $H$ for an evolution time $t$ with $o(t)$ queries, then we can solve the permutation chain problem with $o(t)$ queries as well. However, this violates the hardness of the permutation chain problem.

Now we are ready to present the formal proof of Theorem~\ref{thm:lower_bound_oracle}.

\begin{proof}[Proof of Theorem~\ref{thm:lower_bound_oracle}]
We construct a time-independent Hamiltonian $H$ acting on a $2^{n}(L+1)$-dimensional Hilbert space where $L+1=f(n)$.
The basis vector of the $2^{n}(L+1)$-dimensional Hilbert space is denoted by $\ket{j, x}$ where  $j\in \{0,1,\dots, L\}$ and $x\in\{ 0, 1 \} ^{n}$.
The element of $H$ is defined as follows.
\begin{equation}
\bra{j',x'}H\ket{j,x} = \left\{
\begin{array}{ll}
      1, &\mathrm{if }j' = j+1\ \mathrm{and}\ x' = \hash_{j+1}(x)\\
      1, &\mathrm{if }j' = j-1\ \mathrm{and}\ x' = \hash_{j}^{-1}(x)\\
      0, &\mathrm{otherwise}.
\end{array}
\right.
\end{equation}
Notice that the Hamiltonian $H$ is $2$-sparse and the absolute value of every matrix element is at most $1$.

We are going to show the following. Suppose $\Hsimalg$ can simulate $H$ for an evolution time $t\in[0, (L+1)/2]$ within an error $\epsilon <1/4$ by making $o(t)$ $\pquery$-parallel queries to $\OracleH$ and $\OracleSparse$. Then we can construct a reduction $\reduction$ that makes $o(t)$ $\pquery$-parallel queries to $\hash_{j}$ and $\hash_{j}^{-1}$ and outputs a pair $(x_0, x_{\querynum})$ such that $x_{\querynum} = \hashstep^{(\querynum)}(x_0)$ and $ q> t$ with constant probability.
  The reduction $\reduction$ is described as follows:
  \begin{enumerate}
    \item Run the Hamiltonian simulation algorithm $\Hsimalg$ on inputs the Hamiltonian $H$, the evolution time $t\in [0, L/2]$ and the initial state $\ket{0, 0^{n}}$.

    \begin{itemize}
        \item When $\Hsimalg$ queries $\OracleH$ on the index $((j,x), (j',x'))$, the reduction $\reduction$ returns the response by the following rules.
    \begin{itemize}
      \item If $j' = j+1$ and $x'= \hash_{j+1}(x)$, then $\reduction$ returns $1$.
      \item If $j' = j-1$ and $x' = \hash_{j}^{-1}(x)$, then $\reduction$ returns $1$.
      \item Otherwise, $\reduction$ returns $0$.
    \end{itemize}

    \item When $\Hsimalg$ queries $\OracleSparse$ on $((j,x), s)$, reduction $\reduction$ returns the response by the following rules.
      \begin{itemize}
        \item If $j=0$ and $s=1$, then $\reduction$ returns $(1,\hash_{1}(x))$.
        \item If $j=L$ and $s=1$, then $\reduction$ returns $(L-1,\hash_{L}^{-1}(x))$
        \item If  $j\neq 0, L$ and $s=1$, then $\reduction$ returns $(j-1, \hash_{j}^{-1}(x))$.
        \item If $j\neq 0, L$ and $s=2$, then  $\reduction$ returns $(j+1, \hash_{j+1}(x))$.
      \end{itemize}
    \end{itemize}

      \item Measure the system in the $\{\ket{j,x}\}$ basis and obtain the outcome $(q, x_{q})$.
      \item Output $x_{q}$.
  \end{enumerate}

  Note that answering a query to the entry oracle $\OracleH$ can be implemented by $O(1)$ queries to $\hash_{j}$ and $\hash_{j}^{-1}$. Similarly, the sparse structure oracle $\OracleSparse$ can be simulated by $O(1)$ queries to $\hash_{j}$ and $\hash_{j}^{-1}$ as well.

  Next, we analyze the evolution under $H$.
  Let us define another Hamiltonian $H\vert_{0}$ restricted to the subspace $\spn\left(\left\{\ket{j, \hashstep^{(j)}(0)}\right\}_{j=0}^{L}\right)$:
  \begin{equation*}
  \label{eq:h_o_subspace}
      H\vert_{0} = \sum_{j=0}^{L-1} \ket{j+1, \hashstep^{(j+1)}(0)}\bra{j,\hashstep^{(j)}(0)} + \ket{j,\hashstep^{(j)}(0)}\bra{j+1, \hashstep^{(j+1)}(0)}.
  \end{equation*}
  By Observation~\ref{obs:h_subspace}, the time evolution under $H\vert_{0}$ is equivalent to the time evolution under $H$ with the initial state $\ket{0, 0^{n}}$.

  We first consider a perfect Hamiltonian simulation algorithm $\wt{\Hsimalg}$ that outputs the state $\ket{\wt{\psi}} \coloneqq e^{-iHt}\ket{0, 0^{n}} = e^{-iH\vert_{0}t}\ket{0, 0^{n}}$.  
  In Step~2, the measurement outcome $(q,x_q)$ satisfies $x_q = \hashstep^{(q)}(0^{n})$.
  Then by Lemma~\ref{lem:simple_h}, the probability that the measurement outcome satisfies $q > t$ is at least $1/3$.

  Next, we consider the general simulation algorithm that outputs a state $\ket{\psi}$ such that $\trdist(\ket{\psi}\bra{\psi},$ $ \ket{\widetilde{\psi}}\bra{\widetilde{\psi}}) \leq 1/4$. 
  By the property of the trace distance, the difference in probabilities that $\cR$ outputs a correct outcome by measuring $\ket{\psi}$ and $\ket{\wt{\psi}}$ is at most $1/4$.
  As a result, $\reduction$ outputs the accepted string $x_q$ with probability at least $1/3 - 1/4 = 1/12$.

  Combining everything together, if $\Hsimalg$ simulates $H$ for time $t$ within $\epsilon \leq 1/4$ by making $o(t)$ $\pquery$-parallel queries, then $\reduction$ will output $x_{q} = \hashstep^{(q)}(0^n)$ such that $q>t$ with constant probability by making $o(t)$ $\pquery$-parallel queries. This contradicts Corollary~\ref{thm:hardness_hash_chain}.
\end{proof}

\section{No Fast-forwarding in Plain Model}
\label{sec:plain}
In this section, we are going to investigate the parallel lower bound of Hamiltonian simulation in the plain model.
In the plain model, we are interested in the Hamiltonians that have a succinct description.
Typically, we consider the \emph{local Hamiltonians}.

\begin{definition}[Local Hamiltonian]
    \label{def:local_h}
    We say a Hamiltonian $H$ that acts on $n$ qubits is \emph{$k$-local} if $H$ can be written as
    \begin{equation*}
        \label{eq:def_local_h}
        H = \sum_{j} H_{j},
    \end{equation*}
    where each $H_j$ acts non-trivially on at most $k$ qubits.
\end{definition}

The \emph{geometrically local} Hamiltonians are another kind of Hamiltonians that often appear in physics models.
A geometrically local Hamiltonian is a local Hamiltonian with more constraints.
For a geometrically local Hamiltonian written by $H =\sum_{j}H_{j}$, each term $H_{j}$ acts non-trivially on the qubits that are near in space.
We are especially interested in one-dimensional geometrically local Hamiltonians.

\begin{definition}[One-dimension geometrical local Hamiltonians]
  \label{def:geometrical_local_h}
  Let a system consist of $n$ qubits that are aligned in space and each qubit is labeled by an integer $l\in [n]$.
  Let $H = \sum_j H_{j}$ be a $k$-local Hamiltonian that acts on $n$ qubits.
  We say $H$ is an \emph{one-dimension geometrically local} Hamiltonian if each $H_{j}$ acts non-trivially on at most $k$ consecutive indices.
\end{definition}

For example, consider Hamiltonians $H_1$ and $H_2$ acting on four qubits defined as follows:
\begin{equation*}
  H_1
  \coloneqq \sigma^{X} \otimes  \sigma^{Z} \otimes I \otimes I
  + \sigma^{Z} \otimes I \otimes I \otimes \sigma^{Z}
  + I \otimes I \otimes \sigma^{X} \otimes I,
\end{equation*}
and
\begin{equation*}
  H_2
  \coloneqq \sigma^{Z} \otimes \sigma^{Z} \otimes I \otimes I
  + I \otimes \sigma^{Z} \otimes \sigma^{Z} \otimes I
  + I \otimes I \otimes \sigma^{Z} \otimes \sigma^{Z},
\end{equation*}
where $\sigma^{X}$ and $\sigma^{Z}$ are Pauli operators and $I$ is the identity operator.
Hamiltonian $H_1$ is $2$-local but not geometrically local, but Hamiltonian $H_2$ is geometrically local.
We normalize the Hamiltonian by setting the spectral norm $\|H_j\|= O(1)$ for each $j$.

Having a succinct description gives the simulation algorithm more power than in the oracle model.
In this sense, we obtain a stronger lower bound.
On the other hand, our lower bound in the plain model relies on computational assumptions, which weakens the result. For our lower bound, we need to assume an iterative parallel-hard function, which is slightly modified from the definition of an iterative sequential function by Boneh et al. \cite{Boneh18VDP}.

\begin{definition}[Iterative parallel-hard functions/puzzles]\label{def:p-hard_funcs}
  A function $f\colon \mathbb{N}\times\hat{X}\to X$ where $\hat{X}\in X$ and $|\hat{X}| = 2^{\theta(\lambda)}$ is a (post-quantum) $(s,d)$-iterated parallel-hard function if there exists a function $g\colon X\to X$ such that
  \begin{itemize}
      \item $g$ can be computed by a quantum circuit with width $\lambda$ and size $s(\lambda)$.  Without loss of generality, we can let $s(\lambda)=\Omega(\lambda)$
      \item $f(k,x) = g^{(k)}(x)$.
      \item For all sufficiently large $k=2^{o(\lambda)}$, for any quantum circuit $C$ with depth less than $d(k)$ and size less than $\poly(t,d(k),\lambda)$, \begin{equation*}
          \Pr[C(x)=f(k,x)\mid x\leftarrow \hat{X}]\leq \negl(\lambda)
      \end{equation*}
      Without loss of generality, we assume that $d$ is non-decreasing.
  \end{itemize}

  We say that $f$ forms a (post-quantum) $(s,d)$-iterated parallel-hard puzzle if it only satisfies a weaker version of the third requirement as follows:
  \begin{itemize}
      \item For all $k=2^{o(\lambda)}$, for any uniform quantum circuit $C$ with depth less than $d(k)$ and size less than $\poly(t,d(k),\lambda)$, \begin{equation*}
          \Pr[C(x)=f(k',x) \text{ for some } k'\geq k/2 \mid x\leftarrow \hat{X}]\leq \negl(\lambda).
      \end{equation*}
  \end{itemize}
  Note that an $(s,d)$-iterated parallel-hard function is directly an $(s,d')$-iterated parallel-hard puzzle, where $d'(x):=d(x/2)$.
\end{definition}

Under the (quantum) random oracle heuristic~\cite{BR93,BDFLSZ11}, such parallel-hard puzzles can be heuristically obtained by instantiating the twisted hash chain with a cryptographic hash function.

\begin{assumption}
\label{thm:puzzle}
  With the random oracle heuristic, we can assume that the standard instantiation of the twisted hash chain is parallel-hard by Corollary~\ref{cor:twisted_hash_chain_puzzle}. Assuming the cryptographic hash function $h$ in the instantiation can be implemented by circuits of size $s(\lambda)$ on $\lambda$-bit inputs, the twisted hash chain directly gives an $(s+O(\lambda), d)$ iterative parallel-hard function with $d(x):=x-1$, which is an $(s+O(\lambda), d)$ iterative parallel-hard puzzle with $d(x):=\floor{\frac{x}{2}}-1$
\end{assumption}

We present the simulation lower bound for the local Hamiltonians in the following theorem.

\begin{theorem}[Simulation lower bound for local Hamiltonians in the plain model]
  \label{thm:lower_bound_plain}
  Assuming an $(s,d)$-iterated parallel-hard puzzle, for any integer $n$, there exists a time-independent $c$-local Hamiltonian $H$ acting on $n+(2s(n)T(n))^{1/c}$ qubits such simulating $H$ for an evolution time $t\in[0, s(n)T(n)]$ with error $\epsilon < 1/4$ needs a $\left(d(\floor{t/2s(n)})-O(s(n))\right)$-depth circuit, where $T(\cdot)$ is an arbitrary polynomial and $c$ is a constant.
\end{theorem}

We also have the lower bound for simulating geometrically local time-dependent Hamiltonians.

\begin{theorem}[Simulation lower bound for geometrically local Hamiltonians in the plain model]
    \label{thm:lower_bound_dep}
    Assuming an $(s,d)$-iterated parallel-hard function, for any integer $n$, there exists a piecewise-time-independent 1-D geometrically 2-local Hamiltonian $H$ acting on $n$ qubits such that simulating $H$ for an evolution time $t\in[0, ns(n)T(n)]$ with error $\epsilon(n)\leq 1-1/\poly(n)$ needs a\\ $\left(d(\floor{\frac{t}{ns(n)}})-O(ns(n))-\poly(\log(n),\log\log(1/\epsilon'(n))\right)$-depth circuit, where $T(\cdot)$ is an arbitrary polynomial and  $\epsilon'(n)<1-\epsilon(n)-1/\poly(n)$.
\end{theorem}

We sketch our proof strategy as follows.
The main idea is, again, to reduce the hard problem to the Hamiltonian simulation problem.
First, we consider a quantum circuit $C$ that computes an $(s,d)$-iterated parallel-hard puzzle, which according to the definition, can be written as a sequential composition of $\lambda$-qubit $s(\lambda)$-sized circuits.
Then we construct a Hamiltonian $H_{\mathrm{circuit}}$ to implement the circuit $C$ by the circuit to Hamiltonian reduction technique.
The circuit to time-independent reduction is introduced in Section~\ref{sec:c2h}, and the circuit to time-dependent reduction is introduced in Section~\ref{sec:c2timedep}.
Finally, we use the Hamiltonian simulation algorithm to simulate the Hamiltonian $H_{\mathrm{circuit}}$.
If we can fast-forward the Hamiltonian evolution under $H_{\mathrm{circuit}}$, then we can break the depth guarantee provided by the iterated parallel-hard puzzle.

\subsection{Circuit to time-independent Hamiltonian}
\label{sec:c2h}
Feynman suggested that we can implement a quantum circuit (which was called reversible computation at his time) by a time-independent Hamiltonian \cite{Fey85}.
About a decade later, Childs and Nagaj provided rigorous analyses for the implementation \cite{Chi04, nagaj2010fast}.
The idea is to introduce an extra register, which is called the clock register $\cH_{\mathrm{clock}}$, associated with the circuit register $\cH_{\mathrm{circuit}}$ to record the progress of the quantum circuit.
After introducing the clock register, we define a state $\ket{\psi_{j}} = \ket{\phi_j}_{\mathrm{circuit}}\otimes \ket{\gamma_{j}}_{\mathrm{clock}}$ to indicate that the $j$-th steps outcomes is $\ket{\phi_j}$.
We can construct a Hamiltonian acting on $\cH_{\mathrm{circuit}}\otimes \cH_{\mathrm{clock}}$ such that during the evolution, the system is in  $\spn\left(\{\ket{\psi_{j}}\}\right)$.
When we measure the clock register $\cH_{\mathrm{clock}}$ and get the outcome $\gamma_k$, the quantum state in the circuit register $\cH_{\mathrm{circuit}}$ collapses to $\ket{\phi_{k}}$.
And then we obtain the $k$-th step outcome of the quantum circuit.
Similar techniques appear in the proof of QMA completeness \cite{KSV02} and universality of adiabatic computation \cite{AvK+08}.

In \cite{nagaj2010fast}, Nagaj proved that for any quantum circuit $C$ with $L$ quantum gates, there is a Hamiltonian $H_{\mathrm{circuit}}$ such that evolving the system under the Hamiltonian $H_{\mathrm{circuit}}$ for time $O(L)$ and then measuring the system, we can get the final state of the quantum circuit $C$ with high probability.
We extend Nagaj's result.
In this section, we are going to prove that if the system evolves under $H_{\mathrm{circuit}}$ for time $t\in [0, L/2]$, we can get $\ket{\psi_{j}}$ where $j\geq t$ with high probability.
Our method is slightly different from Nagai's.
In \cite{nagaj2010fast}, the evolution time is uniformly sampled,
while we have an explicit evolution time.
Another difference is that in \cite{nagaj2010fast} it needs to pad $O(L)$ dummy identity gates at the end of the quantum circuit to amplify the probability of getting the output state.
In our construction, padding is not required.

Let a quantum circuit $C$ which acts on the register $\cH_{\mathrm{circuit}}$
consist of a sequence of $g$ quantum gates $U_1, U_2, \dots, U_L$.
Namely,
\begin{equation*}
  \label{eq:circuit}
	C = U_LU_{L-1}\cdots U_1.
\end{equation*}
After introducing the clock register $\cH_{\mathrm{clock}}$, and the clock state, which is a family of orthonormal states $\left\{\ket{\gamma_{j}}\right\}_{j = 0}^{L}$
where each $\ket{\gamma_{j}} \in \cH_{\mathrm{clock}}$,
we construct the following Hamiltonian
\begin{equation}
	\label{eq:h_circuit}
	H_{\mathrm{circuit}} := \sum_{j=1}^{L} H_j,
\end{equation}
where
\begin{equation}
	\label{eq:hj}
	H_j := U_j \otimes \ket{j}\bra{j-1} + U_j^{\dagger} \otimes \ket{j-1}\bra{j}.
\end{equation}
Let the input state of the circuit be $\ket{\phi_0^{(0)}}$,
and let $\ket{\phi_j^{(0)}}$ denote the quantum state of $j$-th step.
That is, $\ket{\phi_j^{(0)}} = U_jU_{j-1}\cdots U_1\ket{\phi_0^{(0)}}$.
Define the state
\begin{equation}
  \label{eq:histrical_state}
  \ket{\psi_j^{(0)}} := \ket{\phi_j^{(0)}}\otimes \ket{\gamma_j},
\end{equation}
which is the state after $j$ steps of $C$. 

Let $\{ \ket{\phi_{0}^{(1)}}, \dots , \ket{\phi_{0}^{(N-1)}}\}$  be the states that are orthogonal to $\ket{\phi_0^{(0)}}$ where $N$ denotes the dimension of $\mathcal{H}_{circuit}$.
Besides, let $\ket{\phi_j^{(m)}}:= U_jU_{j-1}\cdots U_1\ket{\phi_0^{(m)}}$ and $\ket{\psi_j^{(m)}} := \ket{\phi_j^{(m)}}\otimes \ket{\gamma_j}$ where $m\in[N-1]$.
We have
\begin{equation*}
		H_{\mathrm{circuit}}\ket{\psi_j^{(m)}} = \left\{\begin{array}{ll}
		\ket{\psi_{j+1}^{(m)}} &,\ \mathrm{if}\ j=0,\\
		\ket{\psi_{j-1}^{(m)}} + \ket{\psi_{j+1}^{m}} &,\ \mathrm{if}\ j =1,\dots, L-1,\\
		\ket{\psi_{j-1}^{(m)}} &,\ \mathrm{if}\  j=L,
	\end{array}\right.
\end{equation*}
for any $m\in\{0, 1, \dots, N-1\}$.

Again, the evolution under $H_{\mathrm{circuit}}$ can be viewed as quantum walks on a graph.
We construct a graph illustrated in Figure~\ref{fig:q_walk_c2H} whose  adjacency matrix is $H_{\mathrm{circuit}}$.

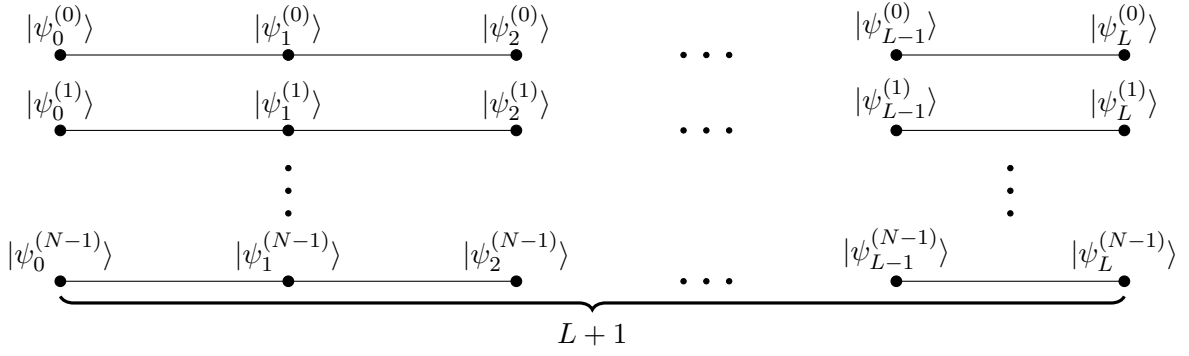
\begin{figure}[h]
    \centering
    \begin{adjustbox}{max width = \textwidth}
\begin{tikzpicture}
  \filldraw [black] (0,3) circle (2pt) node[above]{$\ket{\psi_0^{(0)}}$};
  \filldraw [black] (0,2) circle (2pt) node[above]{$\ket{\psi_{0}^{(1)}}$};
  
  \filldraw [black] (0,0) circle (2pt) node[above]{$\ket{\psi_{0}^{(N-1)}}$};

  \filldraw [black] (3,3) circle (2pt) node[above]{$\ket{\psi_1^{(0)}}$};
  \filldraw [black] (3,2) circle (2pt) node[above]{$\ket{\psi_{1}^{(1)}}$};

  \filldraw [black] (3,0) circle (2pt) node[above]{$\ket{\psi_{1}^{(N-1)}}$};

  \draw (0,3) -- (3,3);
  \draw (0,2) -- (3,2);

  \draw (0,0) -- (3,0);
  \filldraw [black] (6,3) circle (2pt) node[above]{$\ket{\psi_{2}^{(0)}}$};
  \filldraw [black] (6,2) circle (2pt) node[above]{$\ket{\psi_{2}^{(1)}}$};

  \filldraw [black] (6,0) circle (2pt) node[above]{$\ket{\psi_{2}^{(N-1)}}$};

  \draw (3,3) -- (6,3);
  \draw (3,2) -- (6,2);

  \draw (3,0) -- (6,0);

  \filldraw [black] (8.2, 3) circle (1pt);
  \filldraw [black] (8.5, 3) circle (1pt);
  \filldraw [black] (8.8, 3) circle (1pt);
  
  \filldraw [black] (8.2, 2) circle (1pt);
  \filldraw [black] (8.5, 2) circle (1pt);
  \filldraw [black] (8.8, 2) circle (1pt);
  
  \filldraw [black] (8.2, 0) circle (1pt);
  \filldraw [black] (8.5, 0) circle (1pt);
  \filldraw [black] (8.8, 0) circle (1pt);
  
  \filldraw [black] (3, 1.5) circle (1pt);
  \filldraw [black] (3, 1.2) circle (1pt);
  \filldraw [black] (3, 0.9) circle (1pt);
  
  \filldraw [black] (12.5, 1.5) circle (1pt);
  \filldraw [black] (12.5, 1.2) circle (1pt);
  \filldraw [black] (12.5, 0.9) circle (1pt);

  \filldraw [black] (11,3) circle (2pt) node[above]{$\ket{\psi_{L-1}^{(0)}}$};
  \filldraw [black] (11,2) circle (2pt) node[above]{$\ket{\psi_{L-1}^{(1)}}$};

  \filldraw [black] (11,0) circle (2pt) node[above]{$\ket{\psi_{L-1}^{(N-1)}}$};
  \filldraw [black] (14,3) circle (2pt) node[above]{$\ket{\psi_{L}^{(0)}}$};
  \filldraw [black] (14,2) circle (2pt) node[above]{$\ket{\psi_{L}^{(1)}}$};

  \filldraw [black] (14,0) circle (2pt) node[above]{$\ket{\psi_{L}^{(N-1)}}$};

   \draw (11,3) -- (14,3);
   \draw (11,2) -- (14,2);

   \draw (11,0) -- (14,0);

  \draw [decorate, decoration = {brace, amplitude=5pt},  very thick] (14,-0.2) --  (0,-0.2);
  \node at (7, -0.7) {$L+1$};
\end{tikzpicture}
\end{adjustbox}
    \caption{The quantum walk on a graph for the circuit to Hamiltonian reduction.}
    \label{fig:q_walk_c2H}
\end{figure}

We have the following lemma.

\begin{lemma}
  \label{lem:c2H_lemma}
  A system evolves under the Hamiltonian $H_\mathrm{circuit}$ described in \eqref{eq:h_circuit} with the initial state $\ket{\psi_0^{(0)}}$ described in \eqref{eq:histrical_state}.
  If the clock register is measured at time $t \in [0, L/2]$ in the $\{\gamma_j\}$ basis and get the outcome $l$,
  the probability that the circuit register collapses to $\ket{\phi_{l}^{(m)}}$ where $m=0$ and $l>t$ is at least $1/3$.
\end{lemma}
\begin{proof}
  By the similar argument of Observation~\ref{obs:h_subspace}, we define another Hamiltonian $H\vert_{\psi^{(0)}}$ restricted to the subspace $\spn\left(\{\ket{\psi_j^{(0)}}\}_{j=0}^{L}\right)$:
  \begin{equation}
    \label{eq:h_c2H_restricted}
    H\vert_{\psi^{(0)}} := \sum_{j=0}^{L-1} \ket{\psi_{j+1}^{(0)}}\bra{\psi_{j}^{(0)}} + \ket{\psi_{j}^{(0)}}\bra{\psi_{j+1}^{(0)}}.
  \end{equation}
  If the system is initially at $\ket{\psi_{0}^{(0)}}$, the time evolution under $H_{\mathrm{circuit}}$ is the same as the time evolution under  $H\vert_{\psi^{(0)}}$.
  As a result, we have $m=0$ for any time $t$.
  
  Because $\ket{\psi_{j}^{(0)}} = \ket{\phi_{j}^{(0)}}\otimes\ket{\gamma_j}$, 
  the measurement results of measuring clock register in $\{\ket{\gamma_j}\}$ basis is the same as measuring the entire system in $\{\ket{\psi_{j}^{(0)}}\}$ basis.
  The probability that $l>t$ can be obtained directly from Lemma~\ref{lem:simple_h}.
  This finishes the proof.
\end{proof}

Next, we present our construction of the clock state.

\begin{lemma}\label{construct:clock_state}
    For all $c,T\in\N$, there exists a construction that implements the clock state for time $T$ with locality $c$ and at most $O(T^{1/(c-1)})$ qubits.
\end{lemma}
\begin{proof}
    Let $n$ be the smallest integer such that $\binom{n}{c-1} \geq T$. Thus, $n = O(T^{1/(c-1)})$.
    Consider the system that consists of $n$ qubits indexed from $1$ to $n$. 
    Consider the Johnson graph $J_{n,{c-1}}$ (see Definition~\ref{def:Johnson_graph}), for each node $S \subseteq [n]$, define the $n$-qubit state $\ket{S}$ as
    \[
    \ket{S} \coloneqq \bigotimes_{i=1}^n \ket{I_S(i)}_i,
    \]
    where the subscript $i$ denotes the register of the $i$-th qubit; $I_S:[n]\to\bits$ is the indicator function that equals $1$ if $i\in S$ and $0$ otherwise.
    
    Choose an arbitrary Hamiltonian path of $J_{n,{c-1}}$, denoted by $(S_1,S_2,\dots,S_{\binom{n}{c-1}})$. Now, each time $j\in[\binom{n}{c-1}]$ corresponds to the $n$-qubit state $\ket{S_j}$.
    
    For every $j\in\left[ \binom{n}{{c-1}}-1 \right]$, the time transition $\ket{j+1}\bra{j}$ is implemented by 
    \[
    E_{j \to j+1} \coloneqq \bigotimes_{i=1}^n P_i,
    \] 
    where 
    \[
    P_i \coloneqq
    \begin{cases}
        \ket{1}\bra{1}_i, & \text{  if    } i \in S_{j}\cap S_{j+1} \\
        \ket{1}\bra{0}_i, & \text{  if    } i \in S_{j+1}\setminus S_{j} \\
        \ket{0}\bra{1}_i, & \text{  if    } i \in S_{j}\setminus S_{j+1} \\
        I, & \text{ otherwise}. \\
    \end{cases}
    \]
    It is easy to see that the transitions are $c$-local due to the definition of $J_{n,{c-1}}$. It remains to check the correctness of the transitions. That is, for every $j,j' \in\left[ \binom{n}{{c-1}}-1 \right]$, they should satisfy
    \[
    E_{j \to j+1}\ket{S_{j'}} = 
    \begin{cases}
        \ket{S_{j'+1}}, & \text{  if    } j = j' \\
        0, & \text{ otherwise}. \\
    \end{cases}
    \]
    When $j=j'$, the equality holds since there is an edge between $S_j$ and $S_{j+1}$. 
    When $j\neq j'$, there must exist an $i^*\in S_j$ such that $i^*\notin S_{j'}$. Therefore, $P_{i^*} = \ket{0}\bra{1}_{i^*}$ will vanish $\ket{S_{j'}}$ because the $i^*$-th qubit of $\ket{S_{j'}}$ is in the state $\ket{0}_i$. This verifies the correctness.
\end{proof}

\subsection{Proof of the lower bound for local Hamiltonians}
\label{sec:proof_of_the_lower_bound}

\begin{proof}[Proof of Theorem~\ref{thm:lower_bound_plain}]
  From an $(s,d)$-iterated parallel-hard puzzle $f(k,x) = g^{(k)}(x)$, we here show how to construct a time-independent $(c+2)$-local Hamiltonian $H$ acting on $n + (2s(n) T(n))^{1/c}$ qubits.
  
  The construction is direct. By definition, $g$ can be implemented by an $s(n)$-sized quantum circuit $C^g$. We can thus construct a circuit $C$ concatenating $T(n)$ copies of $C^g$, which computes $f(T(n),x)$ with size $s(n)T(n)$. Note that, if we denote $C(i,x)$ to be the intermediate output of $C(x)$ after applying the $i$-th gate, then we additionally have $C(ks(n),x) = \ket{f(k,x)}$ for all $k\leq T(n)$.
  
  Given such a circuit $C$, we can construct a Hamiltonian $H$ by the circuit-to-Hamiltonian reduction introduced in Section~\ref{sec:c2h}. We use the construction in Lemma~\ref{construct:clock_state} with locality $c$ and time-bound $2s(n)T(n)$ to implement the clock state. Hence $H = \sum_{j} H_j$, where $H_j := U_{j}\otimes \ket{j}\bra{j-1} + U_{j}^{\dagger}\otimes \ket{j-1}\bra{j}$, and $U_j$ is a unitary corresponding to the $j$-th gate of $C$. Note that $U_{j}$ is always an one- or two-qubit gate, and $\ket{j}\bra{j-1}$ is $c$-local. 
  Hence $H$ is $c+2$ local over $n + (2s(n)T(n))^{1/c}$ qubits. It is also direct to see that $\|H_j\| = O(1)$ for each $j$.

    For such $H$ and any $t\in[0,s(n)T(n)]$, if there is a quantum algorithm $\Hsimalg$ that computes $e^{-iHt}$ within depth $d_\Hsimalg$, we can indeed construct a quantum algorithm $\reduction$ with depth $d_\Hsimalg+O(n)$ that computes the underlying parallel-hard puzzle. The algorithm $\reduction$ is defined as follows.

  \begin{enumerate}
    \item \label{stp:one}
    For an input $x\in\set{0,1}^n$, run the algorithm $\Hsimalg$ with input Hamiltonian $H$ and input state $\ket{x}$, obtain the output state of $\Hsimalg$, denoted by $\ket{\psi}$.

    \item \label{stp:two}
    Measure the clock register $\mathcal{H}_{\mathrm{clock}}$ in $\{\ket{\gamma_j}\}$ basis and obtain some $l \in [s(n) f(n)]$.
    The residual state in the circuit register $\mathcal{H}_{\mathrm{circuit}}$ is denoted by $\ket{\phi_l}$.

    \item \label{stp:three}
    Let $m = \ceil{l/s(n)}$.
    Apply $U_{ms(n)}\cdots U_{l+1}$ on  $\ket{\phi_l}$.
    Let the final state be $\ket{\phi_m}$.
    
    \item \label{stp:four}
    Measure $\ket{\phi_m}$ on the computational basis and obtain the outcome $x_{m}$. 
    
    \item
    Output $x_{m}$

  \end{enumerate}
  
    In this construction, Step~2 and Step~4 can be done within depth $O(n)$ and Step~3 can be done within depth $s(n)$. It is easy to see that $\reduction$ can be implemented with depth $d_\Hsimalg+O(s(n))$.

    We claim that, with constant probability, $m>\frac{t}{s(n)}$ and $x_{m} = f(m,x)$. This implies that $\reduction$ can break the underlying parallel-hard puzzle on $k=2\floor{\frac{t}{s(n)}}$
    
    To prove the claim, we first consider a simplified case, where there exists an ideal  $\wt{\Hsimalg}$ that perfectly simulates $H$ for time $t$ and plugs it into our construction. We use $\ket{\wt{\psi}}$ to denote the output of $\wt{\Hsimalg}$ in Step~1, and similarly $\ket{\wt{\psi}_l}$ for output in Step~2. By Lemma~\ref{lem:c2H_lemma}, we have $\ket{\wt{\phi}_l} = U_l\cdots U_1\ket{x}$ for all $l$ and the probability that we get some $l>t$ in Step~2 is at least $1/3$. By the definition of $C$, the output state of Step~3 satisfies $\ket{\wt{\phi}_m} = U_{ms(n)}\cdots U_{1}\ket{x} = \ket{f(m,x)}$. Moreover, $m\geq \floor{\frac{t}{s(n)}}$ with probability at least $1/3$. This matches the required condition of our claim.

    Now we return to the general $\Hsimalg$ with simulation error $\epsilon < 1/4$. Let $\ket{\psi}$ be the output of $\Hsimalg$ with error $\epsilon$. Observe that we have $\trdist(\ket{\psi}\bra{\psi}, \ket{\wt{\psi}}\bra{\wt{\psi}}) < 1/4$. Thus, by the definition of the trace distance, the difference in probabilities of obtaining any outcome by applying the same procedure to two states should be at most $1/4$. Hence, with probability at least $1/3-1/4=1/12$, measuring $\ket{\phi_m}$ gives $f(m,x)$ with some $m\geq \floor{\frac{t}{s(n)}}$. This completes the proof of our claim.
    
    Finally, by the security guarantee of the $(s,d)$-iterated parallel-hard puzzle, any circuit computing the puzzle for $k=2\floor{\frac{t}{s(n)}}$ should have depth at least $d(2\floor{\frac{t}{s(n)}})$. This gives an lower bound that $d_\Hsimalg+O(s(n))> d(2\floor{\frac{t}{s(n)}})$, which completes the proof.
\end{proof}

\subsection{Circuit to time-dependent Hamiltonian}
\label{sec:c2timedep}
In this section, we will show how to encode a circuit into a time-dependent geometrically local circuit.

\begin{lemma} 
\label{lem:swap}
A quantum circuit $C$ (of 2-qubit gates) over $n$ qubits of size $s$ can be transformed into a circuit $C'$ over $n$ qubits of size $ns$, such that every gate in $C'$ acts only on consecutive qubits, and $C'(x)=\pi C(x)$ for some permutation $\pi$ on $n$ elements. We call such $C'$ a geometrically local circuit.
\end{lemma}

\begin{proof}
The proof of this small lemma is very direct. Given a circuit $C$ with (sequential) gates $G_1,G_2,\dots,G_s$, where gate $G_i$ acts on qubits $\alpha_i, \beta_i$, then we can rewrite $C$ as $S_{1,\alpha_1}, S_{1,\alpha_1+1},\dots, $ $S_{1,\beta_1-2},G'_1, G'_2,\dots, G'_s$, where $S_i$ is a swap gate acting on the $i$ and $(i+1)$-th qubits, and $G'_i$ is $G_i$ acting on the permuted qubits. Note that now $G'_1$ is acting on two consecutive qubits $\beta_1-1, \beta_1$. It is easy to see that applying $G_1$ and applying $S_{1,\alpha_1}, S_{1,\alpha_1+1},\dots, S_{1,\beta_1-2}, G'_1$ generate output states that differ up to a permutation. Through repeating such process $s$ times, we can obtain a circuit $C'$ that consists of at most $(n-1)s$ swap gates and $s$ permuted gates from $C$. Furthermore, every gate in $C'$ is acting only on consecutive qubits.
\end{proof}

\begin{theorem}
\label{thm:dep_trans}
Given a quantum circuit $C$ over $n$ qubits that consists of $s$ gates $U_1\dots U_s$, we can define a time-dependent Hamiltonian $H$ such that $H(t):=-i\log U_i$ for all $t\in[i-1,i)$ and $U_i:=I_n$ for $i>s$. Note that $H$ obviously satisfies $e^{iHt}\ket{\phi} = C\ket{\phi}$ for $t>s$.
\end{theorem}

\begin{remark}
    A Hamiltonian $H$ obtained from a geometrically local circuit is 1-D geometrically 2-local.
\end{remark}

\begin{remark}
    Such time-independent Hamiltonians that remain constant in each time segment are also called  \emph{piecewise constant Hamiltonians} in~\cite{Haah_2021}.
\end{remark}

\subsection{Proof of the lower bound for geometrically local Hamiltonians}
Before proving the theorem, we will need the existence of a Hamiltonian simulation algorithm for 2-local \emph{time-independent} Hamiltonians. While there are plenty of proposals in the literature, we  use the one in~\cite{zhang2021parallel} which provides a good dependency on $\epsilon$. The choice of the simulation algorithm will only affect an additive term of our bound.

\begin{theorem}[\cite{zhang2021parallel}, with some parameter specified] 
\label{thm:existing_Hsim}
A 2-local Hamiltonian acting on $n$ qubits can be simulated for time $t$ within precision $\epsilon$ by a quantum circuit of depth $\poly(\log\log(1/\epsilon), \log(n), t)$.
\end{theorem}

\begin{proof}[Proof of Theorem~\ref{thm:lower_bound_dep}]
The proof is very similar to the proof of Theorem~\ref{thm:lower_bound_plain}. From an $(s,d)$-iterated parallel-hard function $f(k,x)=g^{(k)}(x)$, we will construct a geometrically 2-local Hamiltonian $H$ over $n$ qubits.

Again, since $g$ can be implemented by an $s(n)$ sized circuit $C^g$, we can construct $C^f$ by concatenating $T(n)$ copies of $C^g$, which computes $f(T(n),x)$ with size $s(n)T(n)$. Then, by applying Lemma~\ref{lem:swap}, we can obtain a geometrically local circuit $C$ of depth at most $ns(n)T(n)$. Note that Lemma~\ref{lem:swap} basically constructs $C$ by adding at most $n-1$ swap gates before each gate of $C^f$. For notation simplicity, we add extra dummy gate before each gate so that there are exactly $(n-1)$ gates added before each gate of $C^f$. Thus the depth of $C$ is exactly $ns(n)T(n)$, and if we denote $C(i,x)$ to be the intermediate output of $C(x)$ after applying the $i$-th gate, we have $C(kns(n),x)=\pi_k\ket{f(k,x)}$ for all $k\leq T(n)$, where $\pi_k$ is some (known) permutation on $n$ bits.

With such a geometrically local circuit $C$, we can apply Theorem~\ref{thm:dep_trans} to obtain a time-dependent geometrically $2$-local Hamiltonian $H$ over $n$ qubits. For all $t\in[0,ns(n)T(n)]$, if there is a quantum algorithm $\Hsimalg$ that computes $e^{-iHt}$ with precision $\epsilon$ of depth $d_\Hsimalg$, we can indeed construct a quantum algorithm $\reduction$ of depth $d_\Hsimalg+O(s(n))+\poly(\log\log(1/\epsilon'),\log n)$ with $\epsilon'<1-\epsilon-1/\poly(n)$. The algorithm $\reduction$ is defined as follows.
\begin{enumerate}
    \item For an input $x\in\set{0,1}^n$, run $\Hsimalg$ on inputs the Hamiltonian $H$ and the initial state $\ket{x}$ to obtain the output state $\ket{\phi}$.
    \item 
    Run the Hamiltonian simulation algorithm $\cS$ in Theorem~\ref{thm:existing_Hsim} with inputs the Hamiltonian $H(t)$ (which is time-independent) 
    , the initial state $\ket{\phi}$, the evolution time $\ceil{t}-t$, and the precision parameter $\epsilon'$ to obtain the output state $\ket{\phi_{\ceil{t}}}$.
    
    \item Let $m=\ceil{\frac{t}{ns(n)}}$. Apply $U_{mns(n)}\cdot U_{\ceil{t}+1}$ to $\ket{\phi_{\ceil{t}}}$. Let the final state be $\ket{\phi_m}$.
    \item Measure the permuted state $\pi_m\ket{\phi_m}$ on the computational basis and obtain the outcome $x_m$.
    \item Output $x_m$.
\end{enumerate}
    
    In the construction, Step~3 can be done within depth $O(ns(n))$ and Step~4 can be done within depth $O(1)$. Thus, $R$ can be instantiated within depth $d_\Hsimalg + O(ns(n)) + \poly(\log\log(1/\epsilon'), \log n)$.
    
    Now we show that if $\epsilon+\epsilon'\leq 1-1/\poly(n)$, then $x_m=f(k,m)$ holds with non-negligible probability.
    
    We consider the case in which the algorithms $\Hsimalg$ and $\cS$ both simulate the evolution perfectly. Denote the output state in each step of this experiment as $\ket{\tilde{\phi}}, \ket{\tilde{\phi}_{\ceil{t}}}$ and $\ket{\tilde{\phi}_m}$. In particular,
    \[
    \ket{\tilde{\phi}}=\exp_{\cT}\left(-i \int_{0}^{t} H(t')dt' \right)\ket{x}
    \]
    and 
    \begin{align*}
    \ket{\tilde{\phi}_{\ceil{t}}} 
    & = e^{-iH(t)(\ceil{t}-t)}\times\exp_{\cT}\left(-i \int_{0}^{t} H(t')dt' \right)\ket{x} \\
    & =  \exp_{\cT}\left(-i \int_{0}^{\ceil{t}} H(t')dt' \right)\ket{x}=\ket{C(\ceil{t},x)},
    \end{align*}
    where the last equation follows from the definition of $H$, while the second last equation follows from the fact that $H$ is constant on the interval $(t,\ceil{t})$. Hence, $\ket{\tilde{\phi}_m}=\ket{C(mns(n),x)}=\ket{f(m,x)}$ with probability $1$.
    
    Back to the actual construction, by the precision guarantee of $\Hsimalg$ and $\cS$, there exists some polynomial $p$ such that $\trdist(\ket{\psi_{\ceil{t}}}\bra{\psi_{\ceil{t}}}, \ket{\tilde{\psi}_{\ceil{t}}}\bra{\tilde{\psi}_{\ceil{t}}}) < 1-1/p(n)$. Thus, $x_m=f(m,x)$ holds with probability at least $1/p(n)$.
    
    Finally, by the security guarantee of the $(s,d)$-iterated parallel-hard function, any quantum circuit that computes $f(m,x)$ should have depth at least $d(m) = d(\ceil{\frac{t}{ns(n)}})$. This gives an lower bound that $d_\Hsimalg+O(ns(n))+\poly(\log\log(1/\epsilon), \log(n))<d(\ceil{\frac{t}{ns(n)}})$, which completes the proof.
\end{proof}

\section{No Fast-forwarding with Natural Simulators}

In previous sections, we have shown no-go theorems for using quantum circuits to parallelly fast-forward (geometrically) local Hamiltonian simulation. Here, we are going to generalize Theorem~\ref{thm:lower_bound_dep} and Theorem~\ref{thm:lower_bound_plain} by showing that simulators that are geometrically local Hamiltonians cannot do much better than quantum circuits. 

\begin{corollary}
\label{cor:lower_bound_tdh}
Assuming an $(s,d)$-iterated parallel-hard function, for any integer $n$, there exists a piecewise-time-independent 1-D geometrically 2-local Hamiltonian $H_A$ acting on $n$ qubits satisfying the following: 
For any geometrically constant-local Hamiltonian $H_B$ acting on $\poly(n)$ qubits, using $H_B$ to simulate $H_A$ for evolution time $t\in[0, ns(n)T(n)]$ with error $\epsilon(n)\leq 1-1/\poly(n)$ needs $(d(\floor{\frac{t}{ns(n)}})-O(ns(n))-\poly(\log(n),\log\log(1/\epsilon'(n)))/\polylog(tn/\epsilon)$ evolution time, where $T(\cdot)$ is an arbitrary polynomial and  $\epsilon'(n)<1-\epsilon(n)-1/\poly(n)$.
\end{corollary}

\begin{proof}

Let $H_A$ be the Hamiltonian we considered in the proof of Theorem~\ref{thm:lower_bound_dep}. Suppose there exists $H_B$ that can simulate $H_A$ for evolution time $t\in[0, ns(n)T(n)]$ and error $\epsilon(n)/2\leq 1-1/\poly(n)$ with  $t'<\frac{d(\floor{\frac{t}{ns(n)}})-O(ns(n))-\poly(\log(n),\log\log(1/\epsilon'(n))}{\polylog(tn/\epsilon)}$ evolution time, where  $\epsilon'(n)<1-\epsilon(n)/2-1/\poly(n)$.

Recall that the algorithm in~\cite{Haah_2021} can simulate a geometrically constant-local Hamiltonian with quantum circuit depth $t\cdot \polylog(tn/\epsilon)$, where $n$ is the number of qubits and $\epsilon$ is the precision parameter.

Then, we apply the algorithm in~\cite{Haah_2021} to simulate $H_B$ with evolution time $t'$ and precision $\epsilon/2$. This leads to a simulation algorithm for $H_A$ with error $\epsilon$ and  circuit depth strictly less than $d(\floor{\frac{t}{ns(n)}})-O(ns(n))-\poly(\log(n),\log\log(1/\epsilon'(n))$, which contradicts Theorem~\ref{thm:lower_bound_dep}. 
\end{proof}

\begin{corollary}
\label{cor:lower_bound_plain}
Assuming an $(s,d)$-iterated parallel-hard function, for any integer $n$, there exists a time-independent $c$-local Hamiltonian $H$ acting on $n+(2s(n)T(n))^{1/c}$ qubits satisfying the following: 
For any geometrically constant-local Hamiltonian $H_B$ acting on $\poly(n)$ qubits, using $H_B$ to simulate $H_A$ for evolution time $t\in[0, s(n)T(n)]$ with error $\epsilon(n)< 1/4$ needs $d(\floor{t/2s(n)})-O(s(n))/\polylog(tn/\epsilon)$ evolution time, where $T(\cdot)$ is an arbitrary polynomial and $c$ is a constant.
\end{corollary}

The proof for Corollary~\ref{cor:lower_bound_plain} is similar to the proof for Corollary~\ref{cor:lower_bound_tdh}.

\bibliographystyle{alphaurl}
\bibliography{References.bib}

\newcommand{\etalchar}[1]{$^{#1}$}
\begin{thebibliography}{BMP{\etalchar{+}}99}

\bibitem[AA17]{Atia_2017}
Yosi Atia and Dorit Aharonov.
\newblock Fast-forwarding of hamiltonians and exponentially precise
  measurements.
\newblock {\em Nature Communications}, 8(1), 2017.
\newblock \href {https://arxiv.org/abs/1610.09619} {\path{arXiv:1610.09619}},
  \href {https://doi.org/10.1038/s41467-017-01637-7}
  {\path{doi:10.1038/s41467-017-01637-7}}.

\bibitem[Als12]{alspach2012johnson}
Brian Alspach.
\newblock Johnson graphs are {H}amilton-connected.
\newblock {\em Ars Mathematica Contemporanea}, 6(1):21--23, 2012.
\newblock \href {https://doi.org/10.26493/1855-3974.291.574}
  {\path{doi:10.26493/1855-3974.291.574}}.

\bibitem[AvK{\etalchar{+}}08]{AvK+08}
Dorit Aharonov, Wim {van Dam}, Julia Kempe, Zeph Landau, Seth Lloyd, and Oded
  Regev.
\newblock Adiabatic {{Quantum Computation Is Equivalent}} to {{Standard Quantum
  Computation}}.
\newblock {\em SIAM Review}, 50(4):755--787, 2008.
\newblock \href {https://arxiv.org/abs/quant-ph/0405098}
  {\path{arXiv:quant-ph/0405098}}, \href {https://doi.org/10.1137/080734479}
  {\path{doi:10.1137/080734479}}.

\bibitem[BACS07]{berry2007efficient}
Dominic~W Berry, Graeme Ahokas, Richard Cleve, and Barry~C Sanders.
\newblock Efficient quantum algorithms for simulating sparse {H}amiltonians.
\newblock {\em Communications in Mathematical Physics}, 270(2):359--371, 2007.
\newblock \href {https://arxiv.org/abs/quant-ph/0508139}
  {\path{arXiv:quant-ph/0508139}}, \href
  {https://doi.org/10.1007/s00220-006-0150-x}
  {\path{doi:10.1007/s00220-006-0150-x}}.

\bibitem[BBBF18]{Boneh18VDP}
Dan Boneh, Joseph Bonneau, Benedikt B{\"{u}}nz, and Ben Fisch.
\newblock Verifiable delay functions.
\newblock In {\em Advances in Cryptology - {CRYPTO} 2018}, volume 10991 of {\em
  Lecture Notes in Computer Science}, pages 757--788. Springer, 2018.
\newblock \href {https://doi.org/10.1007/978-3-319-96884-1\_25}
  {\path{doi:10.1007/978-3-319-96884-1\_25}}.

\bibitem[BBBV97]{BBBV97}
Charles~H Bennett, Ethan Bernstein, Gilles Brassard, and Umesh Vazirani.
\newblock Strengths and weaknesses of quantum computing.
\newblock {\em SIAM journal on Computing}, 26(5):1510--1523, 1997.
\newblock \href {https://arxiv.org/abs/quant-ph/9701001}
  {\path{arXiv:quant-ph/9701001}}, \href
  {https://doi.org/10.1137/S0097539796300933}
  {\path{doi:10.1137/S0097539796300933}}.

\bibitem[BCC{\etalchar{+}}14]{BCC+14}
Dominic~W. Berry, Andrew~M. Childs, Richard Cleve, Robin Kothari, and
  Rolando~D. Somma.
\newblock Exponential improvement in precision for simulating sparse
  {Hamiltonians}.
\newblock In {\em Proceedings of the Forty-Sixth Annual ACM Symposium on Theory
  of Computing}, STOC '14, page 283–292. Association for Computing Machinery,
  2014.
\newblock \href {https://arxiv.org/abs/1312.1414} {\path{arXiv:1312.1414}},
  \href {https://doi.org/10.1145/2591796.2591854}
  {\path{doi:10.1145/2591796.2591854}}.

\bibitem[BCK15]{BCK15}
Dominic~W. Berry, Andrew~M. Childs, and Robin Kothari.
\newblock Hamiltonian {Simulation} with {Nearly} {Optimal} {Dependence} on all
  {Parameters}.
\newblock In {\em {IEEE} 56th Annual Symposium on Foundations of Computer
  Science}, pages 792--809. {IEEE}, 2015.
\newblock \href {https://arxiv.org/abs/1501.01715} {\path{arXiv:1501.01715}},
  \href {https://doi.org/10.1109/FOCS.2015.54}
  {\path{doi:10.1109/FOCS.2015.54}}.

\bibitem[BDF{\etalchar{+}}11]{BDFLSZ11}
Dan Boneh, {\"O}zg{\"u}r Dagdelen, Marc Fischlin, Anja Lehmann, Christian
  Schaffner, and Mark Zhandry.
\newblock Random oracles in a quantum world.
\newblock In {\em Advances in Cryptology--ASIACRYPT 2011}, volume 7073 of {\em
  Lecture Notes in Computer Science}, pages 41--69. Springer, 2011.
\newblock \href {https://arxiv.org/abs/1008.0931} {\path{arXiv:1008.0931}},
  \href {https://doi.org/10.1007/978-3-642-25385-0\_3}
  {\path{doi:10.1007/978-3-642-25385-0\_3}}.

\bibitem[BMP{\etalchar{+}}99]{BMP+99}
P~Oscar Boykin, Tal Mor, Matthew Pulver, Vwani Roychowdhury, and Farrokh Vatan.
\newblock On universal and fault-tolerant quantum computing: a novel basis and
  a new constructive proof of universality for {S}hor's basis.
\newblock In {\em 40th Annual Symposium on Foundations of Computer Science
  (Cat. No. 99CB37039)}, pages 486--494. IEEE, 1999.
\newblock \href {https://arxiv.org/abs/quant-ph/9906054}
  {\path{arXiv:quant-ph/9906054}}, \href
  {https://doi.org/10.1109/SFFCS.1999.814621}
  {\path{doi:10.1109/SFFCS.1999.814621}}.

\bibitem[BR93]{BR93}
Mihir Bellare and Phillip Rogaway.
\newblock Random oracles are practical: A paradigm for designing efficient
  protocols.
\newblock In {\em Proceedings of the 1st ACM Conference on Computer and
  Communications Security}, CCS '93, pages 62--73. Association for Computing
  Machinery, 1993.
\newblock \href {https://doi.org/10.1145/168588.168596}
  {\path{doi:10.1145/168588.168596}}.

\bibitem[CB12]{CB12}
Andrew~M. Childs and Dominic~W. Berry.
\newblock Black-box {Hamiltonian} simulation and unitary implementation.
\newblock {\em Quantum Information and Computation}, 12(1\&2):29--62, 2012.
\newblock \href {https://arxiv.org/abs/0910.4157} {\path{arXiv:0910.4157}},
  \href {https://doi.org/10.26421/QIC12.1-2-4}
  {\path{doi:10.26421/QIC12.1-2-4}}.

\bibitem[CCD{\etalchar{+}}03]{Childs_2003}
Andrew~M. Childs, Richard Cleve, Enrico Deotto, Edward Farhi, Sam Gutmann, and
  Daniel~A. Spielman.
\newblock Exponential algorithmic speedup by a quantum walk.
\newblock In {\em Proceedings of the thirty-fifth {ACM} symposium on Theory of
  computing - {STOC} {\textquotesingle}03}. {ACM} Press, 2003.
\newblock \href {https://arxiv.org/abs/quant-ph/0209131}
  {\path{arXiv:quant-ph/0209131}}, \href
  {https://doi.org/10.1145/780542.780552} {\path{doi:10.1145/780542.780552}}.

\bibitem[CFHL21]{chung2021compressed}
Kai-Min Chung, Serge Fehr, Yu-Hsuan Huang, and Tai-Ning Liao.
\newblock On the compressed-oracle technique, and post-quantum security of
  proofs of sequential work.
\newblock In {\em Advances in Cryptology -- EUROCRYPT 2021}, volume 12697 of
  {\em Lecture Notes in Computer Science}, pages 598--629. Springer, 2021.
\newblock \href {https://arxiv.org/abs/2010.11658} {\path{arXiv:2010.11658}},
  \href {https://doi.org/10.1007/978-3-030-77886-6\_21}
  {\path{doi:10.1007/978-3-030-77886-6\_21}}.

\bibitem[Chi04]{Chi04}
Andrew~MacGregor Childs.
\newblock {\em Quantum Information Processing in Continuous Time}.
\newblock Thesis, Massachusetts Institute of Technology, 2004.
\newblock URL: \url{http://hdl.handle.net/1721.1/16663}.

\bibitem[CMP18]{CMP18}
Toby Cubitt, Ashley Montanaro, and Stephen Piddock.
\newblock Universal {{Quantum Hamiltonians}}.
\newblock {\em Proceedings of the National Academy of Sciences},
  115(38):9497--9502, September 2018.
\newblock \href {https://arxiv.org/abs/1701.05182} {\path{arXiv:1701.05182}},
  \href {https://doi.org/10.1073/pnas.1804949115}
  {\path{doi:10.1073/pnas.1804949115}}.

\bibitem[CRT21]{CRT21}
Jorge Ch{\'{a}}vez{-}Saab, Francisco Rodr{\'{\i}}guez{-}Henr{\'{\i}}quez, and
  Mehdi Tibouchi.
\newblock Verifiable isogeny walks: Towards an isogeny-based postquantum {VDF}.
\newblock In {\em {SAC}}, volume 13203 of {\em Lecture Notes in Computer
  Science}, pages 441--460. Springer, 2021.
\newblock \href {https://doi.org/10.1007/978-3-030-99277-4_21}
  {\path{doi:10.1007/978-3-030-99277-4_21}}.

\bibitem[CZ12]{CZ12}
J.~Ignacio Cirac and Peter Zoller.
\newblock Goals and opportunities in quantum simulation.
\newblock {\em Nature Physics}, 8(4):264--266, April 2012.
\newblock \href {https://doi.org/10.1038/nphys2275}
  {\path{doi:10.1038/nphys2275}}.

\bibitem[EFKP20]{EFK20}
Naomi Ephraim, Cody Freitag, Ilan Komargodski, and Rafael Pass.
\newblock Continuous verifiable delay functions.
\newblock In {\em {EUROCRYPT} {(3)}}, volume 12107 of {\em Lecture Notes in
  Computer Science}, pages 125--154. Springer, 2020.
\newblock \href {https://doi.org/10.1007/978-3-030-45727-3_5}
  {\path{doi:10.1007/978-3-030-45727-3_5}}.

\bibitem[Fey85]{Fey85}
Richard~P. Feynman.
\newblock Quantum {{Mechanical Computers}}.
\newblock {\em Optics News}, 11(2):11--20, February 1985.
\newblock \href {https://doi.org/10.1364/ON.11.2.000011}
  {\path{doi:10.1364/ON.11.2.000011}}.

\bibitem[FMPS19]{FMP19}
Luca~De Feo, Simon Masson, Christophe Petit, and Antonio Sanso.
\newblock Verifiable delay functions from supersingular isogenies and pairings.
\newblock In {\em {ASIACRYPT} {(1)}}, volume 11921 of {\em Lecture Notes in
  Computer Science}, pages 248--277. Springer, 2019.
\newblock \href {https://doi.org/10.1007/978-3-030-34578-5\_10}
  {\path{doi:10.1007/978-3-030-34578-5\_10}}.

\bibitem[GS{\c{S} }21]{Gu_2021}
Shouzhen Gu, Rolando~D. Somma, and Burak {\c{S} }ahino{\u{g}}lu.
\newblock Fast-forwarding quantum evolution.
\newblock {\em Quantum}, 5:577, 2021.
\newblock \href {https://arxiv.org/abs/2105.07304} {\path{arXiv:2105.07304}},
  \href {https://doi.org/10.22331/q-2021-11-15-577}
  {\path{doi:10.22331/q-2021-11-15-577}}.

\bibitem[HHKL18]{Haah_2021}
Jeongwan Haah, Matthew~B. Hastings, Robin Kothari, and Guang~Hao Low.
\newblock Quantum algorithm for simulating real time evolution of lattice
  {Hamiltonians}.
\newblock In {\em 59th {IEEE} Annual Symposium on Foundations of Computer
  Science}, pages 350--360. {IEEE}, 2018.
\newblock \href {https://arxiv.org/abs/1801.03922} {\path{arXiv:1801.03922}},
  \href {https://doi.org/10.1109/FOCS.2018.00041}
  {\path{doi:10.1109/FOCS.2018.00041}}.

\bibitem[JMdW14]{JMD13}
Stacey Jeffery, Fr{\'{e}}d{\'{e}}ric Magniez, and Ronald de~Wolf.
\newblock Optimal parallel quantum query algorithms.
\newblock In Andreas~S. Schulz and Dorothea Wagner, editors, {\em Algorithms -
  {ESA} 2014 - 22th Annual European Symposium}, volume 8737 of {\em Lecture
  Notes in Computer Science}, pages 592--604. Springer, 2014.
\newblock \href {https://arxiv.org/abs/1309.6116} {\path{arXiv:1309.6116}},
  \href {https://doi.org/10.1007/978-3-662-44777-2\_49}
  {\path{doi:10.1007/978-3-662-44777-2\_49}}.

\bibitem[Kra06]{Kra06}
Ilia Krasikov.
\newblock Uniform bounds for {B}essel functions.
\newblock {\em Journal of Applied Analysis}, 12(1):83--91, 2006.
\newblock \href {https://doi.org/10.1515/JAA.2006.83}
  {\path{doi:10.1515/JAA.2006.83}}.

\bibitem[KSV02]{KSV02}
A.~Kitaev, A.~Shen, and M.~Vyalyi.
\newblock {\em Classical and {{Quantum Computation}}}, volume~47 of {\em
  Graduate {{Studies}} in {{Mathematics}}}.
\newblock {American Mathematical Society}, {Providence, Rhode Island}, May
  2002.
\newblock \href {https://doi.org/10.1090/gsm/047} {\path{doi:10.1090/gsm/047}}.

\bibitem[LC17]{LC17}
Guang~Hao Low and Isaac~L. Chuang.
\newblock Optimal {H}amiltonian simulation by quantum signal processing.
\newblock {\em Phys. Rev. Lett.}, 118:010501, January 2017.
\newblock \href {https://arxiv.org/abs/1606.02685} {\path{arXiv:1606.02685}},
  \href {https://doi.org/10.1103/PhysRevLett.118.010501}
  {\path{doi:10.1103/PhysRevLett.118.010501}}.

\bibitem[LC19]{LC19}
Guang~Hao Low and Isaac~L. Chuang.
\newblock Hamiltonian {S}imulation by {Q}ubitization.
\newblock {\em {Quantum}}, 3:163, July 2019.
\newblock \href {https://arxiv.org/abs/1610.06546} {\path{arXiv:1610.06546}},
  \href {https://doi.org/10.22331/q-2019-07-12-163}
  {\path{doi:10.22331/q-2019-07-12-163}}.

\bibitem[LR88]{Luby88}
Michael Luby and Charles Rackoff.
\newblock How to construct pseudorandom permutations from pseudorandom
  functions.
\newblock {\em SIAM Journal on Computing}, 17(2):373--386, 1988.
\newblock \href {https://doi.org/10.1137/0217022} {\path{doi:10.1137/0217022}}.

\bibitem[Nag10]{nagaj2010fast}
Daniel Nagaj.
\newblock Fast universal quantum computation with railroad-switch local
  {H}amiltonians.
\newblock {\em Journal of Mathematical Physics}, 51(6):062201, 2010.
\newblock \href {https://arxiv.org/abs/0908.4219} {\path{arXiv:0908.4219}},
  \href {https://doi.org/10.1063/1.3384661} {\path{doi:10.1063/1.3384661}}.

\bibitem[NC10]{nielsen2010quantum}
Michael~A. Nielsen and Isaac~L. Chuang.
\newblock {\em Quantum Computation and Quantum Information: 10th Anniversary
  Edition}.
\newblock Cambridge University Press, 2010.
\newblock \href {https://doi.org/10.1017/CBO9780511976667}
  {\path{doi:10.1017/CBO9780511976667}}.

\bibitem[Pie19]{Pie19}
Krzysztof Pietrzak.
\newblock Simple verifiable delay functions.
\newblock In {\em {ITCS}}, volume 124 of {\em LIPIcs}, pages 60:1--60:15.
  Schloss Dagstuhl - Leibniz-Zentrum f{\"{u}}r Informatik, 2019.
\newblock \href {https://doi.org/10.4230/LIPIcs.ITCS.2019.60}
  {\path{doi:10.4230/LIPIcs.ITCS.2019.60}}.

\bibitem[RSW96]{RSW96}
Ronald~L Rivest, Adi Shamir, and David~A Wagner.
\newblock Time-lock puzzles and timed-release crypto.
\newblock {\em Massachusetts Institute of Technology. Laboratory for Computer
  Science}, 1996.

\bibitem[SN20]{SN20}
J.~J. Sakurai and Jim Napolitano.
\newblock {\em Modern {{Quantum Mechanics}}}.
\newblock {Cambridge University Press}, third edition, September 2020.
\newblock \href {https://doi.org/10.1017/9781108587280}
  {\path{doi:10.1017/9781108587280}}.

\bibitem[Wes20]{Wes20}
Benjamin Wesolowski.
\newblock Efficient verifiable delay functions.
\newblock {\em J. Cryptol.}, 33(4):2113--2147, 2020.
\newblock \href {https://doi.org/10.1007/s00145-020-09364-x}
  {\path{doi:10.1007/s00145-020-09364-x}}.

\bibitem[Zal99]{zalka1999grover}
Christof Zalka.
\newblock Grover’s quantum searching algorithm is optimal.
\newblock {\em Physical Review A}, 60(4):2746, 1999.
\newblock \href {https://arxiv.org/abs/quant-ph/9711070}
  {\path{arXiv:quant-ph/9711070}}, \href
  {https://doi.org/10.1103/PhysRevA.60.2746}
  {\path{doi:10.1103/PhysRevA.60.2746}}.

\bibitem[Zha13]{zhandry13}
Mark Zhandry.
\newblock A note on the quantum collision and set equality problems, 2013.
\newblock URL: \url{https://arxiv.org/abs/1312.1027}.

\bibitem[Zha19]{zhandry2019record}
Mark Zhandry.
\newblock How to record quantum queries, and applications to quantum
  indifferentiability.
\newblock In {\em Annual International Cryptology Conference}, pages 239--268.
  Springer, 2019.
\newblock \href {https://doi.org/10.1007/978-3-030-26951-7_9}
  {\path{doi:10.1007/978-3-030-26951-7_9}}.

\bibitem[ZWY21]{zhang2021parallel}
Zhicheng Zhang, Qisheng Wang, and Mingsheng Ying.
\newblock Parallel quantum algorithm for hamiltonian simulation, 2021.
\newblock \href {https://arxiv.org/abs/2105.11889} {\path{arXiv:2105.11889}}.

\end{thebibliography}

\end{document}